\documentclass[twoside,leqno,twocolumn]{article}
\usepackage[T1]{fontenc} % for hyphenated characters

\usepackage{ltexpprt}
\usepackage[cmex10]{amsmath}
\usepackage{amssymb}
\usepackage{algorithm}
\usepackage{algorithmic}
\usepackage{enumerate}
\usepackage{array}
\usepackage{ctable} % provides toprule, bottomrule, midrule
\usepackage{float}
\usepackage{graphicx}
\usepackage[caption=false,font=footnotesize]{subfig}
\usepackage{color}

\usepackage{etoolbox}
\newtoggle{fullversion}
\toggletrue{fullversion}  %Comment this line out for conference version

\newcommand{\expect}{\mathrm{Exp}}
\newcommand{\hide}[1]{}

\newcommand{\codeurl}{http://tinyurl.com/l3lgsq7}

\newcommand{\algorithmicinput}{\textbf{Input:}}
\newcommand{\INPUT}{\item[\algorithmicinput]}
\newcommand{\algorithmicoutput}{\textbf{Output:}}
\newcommand{\OUTPUT}{\item[\algorithmicoutput]}

\newcommand{\opt}{\mathrm{OPT}}
\newcommand{\eopt}{E_{\mathrm{OPT}}}
\newcommand{\edeg}{e_{\mathrm{deg}}}
\newcommand{\prob}{\textsc{Spectral Radius Minimization}}
\newcommand{\degr}{\mathrm{d}}
\newcommand{\Vol}{\mathrm{Vol}}
\newcommand{\wmax}{w_{\max}}

\DeclareMathOperator*{\nodes}{nodes}
\DeclareMathOperator*{\walks}{walks}
\DeclareMathOperator*{\ct}{count}
\newcommand{\superscript}[1]{\ensuremath{^{\textrm{#1}}}}
\def\ws{\superscript{*}}
\def\wdg{\superscript{\dag}}

\newcommand{\Input}[1]{\textbf{input:} #1\\}
\newcommand{\Output}[1]{\textbf{output:} #1\\}
\newcommand{\While}[1]{\textbf{while} #1\\}

\newcommand{\tcp}[1]{\texttt{//#1}}

\begin{document}
\title{Approximation Algorithms for Reducing the Spectral Radius To Control Epidemic Spread}
%\title{Near-optimal algorithms for reducing spectral radius of networks}

\hide{
\author{
Sudip Saha\ws\wdg, Abhijin Adiga\ws, B. Aditya Prakash\wdg, Anil Kumar S. Vullikanti\ws\wdg\\
\end{tabular}\newline\begin{tabular}{c}
\begin{tabular}{cc}
   \ws Network Dynamics and Simulation Science Laboratory & \wdg Department of Computer Science\\
   Virginia Bioinformatics Institute & Virginia Tech,\\
   Virginia Tech, Blacksburg, VA 24060 &
   Blacksburg, VA 24060\\
\end{tabular}\\
Email: \{ssaha,abhijin,akumar\}@vbi.vt.edu, badityap@cs.vt.edu
}}

\author{Sudip Saha\thanks{NDSSL, Virginia Bioinformatics Institute, Virginia Tech.}~\thanks{Department of Computer Science, Virginia Tech.\newline Email:\{ssaha, abhijin, akumar\}@vbi.vt.edu, badityap@cs.vt.edu}\\
\and
Abhijin Adiga\footnotemark[1] \\
\and
B. Aditya Prakash\footnotemark[2] \\
\and
Anil Kumar S. Vullikanti\footnotemark[1]~\footnotemark[2]
}

%\title{Near-optimal algorithms for reducing the spectral radius to control epidemic spread}
%% \numberofauthors{5}
%% %
%% \author{
%% Sudip Saha\ws, Abhijin Adiga\ws, B. Aditya Prakash\wdg, Rajmohan
%% Rajaraman\wddg, Anil Vullikanti\ws\wdg\\
%% \end{tabular}\newline\begin{tabular}{c}
%% \begin{tabular}{cc}
%%    \affaddr{\ws Network Dynamics and Simulation Science Laboratory} & \affaddr{\wdg Department of Computer Science}\\
%%    \affaddr{Virginia Bioinformatics Institute,} & \affaddr{Virginia Tech,}\\
%%    \affaddr{Virginia Tech, Blacksburg, VA 24060} &
%%    \affaddr{Blacksburg, VA 24060}\\
%% \end{tabular}
%% \and
%% \affaddr{\wddg College of Computer and Information Science}\\
%% \affaddr{Northeastern University, Boston MA 02115}\\
%% \email{\{ssaha,abhijin,akumar\}@vbi.vt.edu, badityap@cs.vt.edu, rraj@ccs.neu.edu}
%% }
\date{}
\maketitle

\begin{abstract}
{\small
The largest eigenvalue of the adjacency matrix of a network (referred to as the spectral radius) is an important metric
in its own right. Further, for several models of epidemic spread on networks (e.g., the `flu-like' SIS model), it has been shown that
an epidemic dies out quickly if the spectral radius of the graph is below a
certain threshold that depends on the model parameters. This motivates
a strategy to control epidemic spread by reducing the spectral radius
of the underlying network.

In this paper, we develop a suite of provable approximation
algorithms for reducing the spectral radius by removing the minimum cost
set of edges (modeling quarantining) or nodes (modeling vaccinations), with
different time and quality tradeoffs. Our main algorithm, \textsc{GreedyWalk}, is based on the idea of hitting
closed walks of a given length, and gives an $O(\log^2{n})$-approximation,
where $n$ denotes the number of nodes; it also performs much better in
practice compared to all prior heuristics proposed for this problem. We further present a novel sparsification method to improve its running time.

In addition, we give a new primal-dual based algorithm with an even better approximation guarantee ($O(\log n)$),
albeit with slower running time. We
also give lower bounds on the worst-case performance of some of the popular
heuristics. Finally we demonstrate the applicability of our algorithms and the properties of our solutions
via extensive experiments on multiple synthetic and real networks. } %and also discuss which demographic properties can be directly utilized.}
\end{abstract}

%\noindent
%\textbf{Keywords}:

%
\section{Introduction}

Given a contact network, which contacts should we remove to contain
the spread of a virus?  Equivalently, in a computer network, which
connections should we cut to prevent the spread of malware? Designing
effective and low cost interventions are fundamental challenges in
public health and network security.
%Interventions such as vaccinations and quarantining are common strategies in controlling the
%spread of epidemics, such as diseases in human populations and spread of malware; designing
Epidemics are commonly modeled by stochastic diffusion processes, such
as the so-called `SIS' (flu-like) and `SIR' (mumps-like) models on networks (more in Section~\ref{sec:preliminaries}).  An important result that highlights the
impact of the network structure on the dynamics is that epidemics die
out ``quickly'' if $\rho(G)\leq T$, where $\rho(G)$ is the spectral
radius (or the largest eigenvalue) of graph $G$, and $T$ is a
threshold that depends on the disease model \cite{ganesh+topology05,Wang03Epidemic,aditya12}.  This motivates the following strategy for controlling an
epidemic: remove edges (quarantining) or nodes (vaccinating) to reduce
the spectral radius below a threshold $T$---we refer to this as the
spectral radius minimization (\textsc{SRM}) problem, with variants depending
on whether edges are removed (the \textsc{SRME} problem) or whether
nodes are removed (the \textsc{SRMN} problem).
Van Mieghem et al. \cite{vanmieghem:ton12} and
Tong et al. \cite{tong:cikm12} prove that this problem is NP-complete. They also
study two heuristics for it, one based on the components of the first eigenvector
(\textsc{EigenScore}) and another based on degrees (\textsc{ProductDegree}).
However, no rigorous approximations were known for the \textsc{SRME} or
the \textsc{SRMN} problems.

\iffalse
In this paper we present the first provable approximation algorithms for
the \textsc{SRME} and the \textsc{SRMN} problems. Our algorithms are based on
hitting closed walks, and perform very well in practice as well.

From a public policy perspective,
it is often important that the individuals who are selected for
vaccination or the edges selected for quarantining be identified by
simple attributes, such as demographics or geographical locations (in
the case of nodes) or time labels and activities (in the case of
edges), in order for such interventions to be implementable.
We develop a formulation of the \textsc{SRM} problems in
{\em node and edge labeled graphs}\/ to model and compute
such implementable interventions in a unified manner.
\fi
\smallskip

\noindent
\textbf{Our main contributions}.
\smallskip

\iffalse
\noindent
\textbf{1. Formulation}: We present a unified formulation of the spectral
radius minimization (\textsc{SRM}) problems in labeled graphs with arbitrary costs for
nodes, edges and labels.  Our framework generalizes
\cite{vanmieghem:ton12,tong:cikm12} and incorporates the practical
aspect of implementability of the interventions. The variants of the \textsc{SRM}
problems we consider are summarized in Table \ref{table:srm-problems};
our algorithms have the same
basic structure for these different variants.
\smallskip
\fi

\noindent
\textbf{1.~Lower bounds on the worst-case performance of heuristics}:
We show that the \textsc{ProductDegree}, \textsc{EigenScore} and \textsc{Pagerank}
heuristics (defined formally in Section \ref{sec:preliminaries}) can perform quite poorly
in general. We demonstrate graph instances where these heuristics give
solutions of cost $\Omega(\frac{n}{T^2})$ times the optimal,
where $n$ is the number of nodes in the graph.

\iftoggle{fullversion}
{
\noindent
\textbf{2.~Provable approximation algorithms}: We present two bicriteria
approximation algorithms for the \textsc{SRME} and \textsc{SRMN} problems,
with varying approximation quality and running time tradeoffs. Our first
algorithm, \textsc{GreedyWalk}, is based on hitting
closed walks in $G$. We show this algorithm has an approximation bound of
$O(\log{n}\log{\Delta})$ times optimal for the cost of edges removed,
while ensuring that the spectral radius becomes at most
$(1 + \epsilon)$ times the threshold, for $\epsilon$
arbitrarily small (here $\Delta$ denotes the maximum node degree in
the graph). We also design a variant, \textsc{GreedyWalkSparse}, that
performs careful sparsification of the graph, leading to similar
asymptotic guarantees, but better running time, especially when the threshold $T$ is small.
%\textbf{Slightly better bounds for small $T$-- decide how to state this}
We then develop algorithm \textsc{PrimalDual}, which
improves this approximation bound to an $O(\log{n})$
using a more sophisticated primal-dual approach, at the expense of a
slightly higher (but polynomial) running time.
%%
%% \iffalse
%% We propose two bicriteria
%% approximation algorithms for the spectral radius minimization problem
%% with varying approximation quality and running time tradeoffs. First,
%% we develop algorithm \textsc{GreedyWalk}, that is based on hitting
%% closed walks in $G$, which incurs an intervention cost of at most
%% $O(\log{n}\log{\Delta})$ times optimal, while reducing the spectral
%% radius to at most $(1 + \epsilon)$ times the threshold, for $\epsilon$
%% arbitrarily small (here $\Delta$ denotes the maximum node degree in
%% the graph).  We then develop algorithm \textsc{PrimalDual}, which
%% improves this approximation bound to an $O(\log{n})$-approximation
%% using a more sophisticated primal-dual approach, at the expense of a
%% slightly higher (but polynomial) running time.
%% %% Finally, we describe algorithm \textsc{SDPRound}, which uses semidefinite programming based
%% %% rounding and gives a constant factor approximation, but is the slowest
%% %% to run.
%% Though these algorithms all involve fundamentally different
%% techniques, they are based on partial hitting of closed walks, thereby
%% unifying their analysis. \textcolor{red}{Can we highlight and elaborate on
%% the following?} We also identify sparsification techniques
%% that can be combined with all these algorithms, without altering the
%% asymptotic bounds, but enabling faster running times.
%% \fi
%% \smallskip

\noindent
\textbf{3. Extensions}:
We consider two natural extensions of the \textsc{SRME} problem: (i) non-uniform
transmission rates on edges and (ii) node version \textsc{SRMN}.
We show that our methods extend to these variations too.

\noindent
\textbf{4. Empirical analysis}: We conduct an extensive experimental
evaluation of \textsc{GreedyWalk}, a simplified version of
\textsc{PrimalDual} and different heuristics that have been
proposed for epidemic containment on a diverse collection of synthetic and real
networks.  These heuristics involve picking edges $e=(i,j)$ in
non-increasing order of some kind of score; the specific heuristics
we compare include: (i)~\textsc{ProductDegree}, (ii)~\textsc{EigenScore},
(iii)~\textsc{LinePagerank}, and (iv)~\textsc{Hybrid},
which picks the edge based on either the eigenscore or the
product-degree ordering, depending on the maximum decrease in
eigenvalue.
We find that \textsc{GreedyWalk} performs better than all the
heuristics in all the networks we study. We analyze \textsc{GreedyWalk}
for walks of length $k=\Theta(\log{n})$; in practice, we
found that the performance degrades significantly as $k$ is reduced.
}%iftoggle{fullversion}
{
\noindent
\textbf{2.~Provable approximation algorithms}: We present two bicriteria
approximation algorithms for the \textsc{SRME} and \textsc{SRMN} problems,
with varying approximation quality and running time tradeoffs. Our first
algorithm, \textsc{GreedyWalk}, is based on hitting
closed walks in $G$. We show this algorithm has an approximation bound of
$O(\log{n}\log{\Delta})$ times optimal for the cost of edges removed,
while ensuring that the spectral radius becomes at most
$(1 + \epsilon)$ times the threshold, for $\epsilon$
arbitrarily small (here $\Delta$ denotes the maximum node degree in
the graph). We also design a variant, \textsc{GreedyWalkSparse}, that
performs careful sparsification of the graph, leading to similar
asymptotic guarantees, but better running time, especially when the threshold $T$ is small.
%\textbf{Slightly better bounds for small $T$-- decide how to state this}
We then develop algorithm \textsc{PrimalDual}, which
improves this approximation bound to an $O(\log{n})$
using a more sophisticated primal-dual approach, at the expense of a
slightly higher (but polynomial) running time.
%%\textcolor{red}{We show that our methods extend to the node version of the problem, i.e., \textsc{SRMN} too.}

\noindent
\textbf{3. Empirical analysis}: We conduct an extensive experimental
evaluation of \textsc{GreedyWalk}, a simplified version of
\textsc{PrimalDual} and different heuristics that have been
proposed for epidemic containment on a diverse collection of synthetic and real
networks.  For \textsc{SRME}, these heuristics involve picking edges $e=(i,j)$ in
non-increasing order of some kind of score; the specific heuristics
we compare include: (i)~\textsc{ProductDegree}, (ii)~\textsc{EigenScore},
(iii)~\textsc{LinePagerank}, and (iv)~\textsc{Hybrid},
which pick the edge based on either the eigenscore or the
product-degree ordering, depending on the maximum decrease in
eigenvalue.
We find that \textsc{GreedyWalk} performs better than all the
heuristics in all the networks we study. In the experiments, we analyze \textsc{GreedyWalk}
for walks of length $k=\Theta(\log{n})$.
%in practice, we found that the performance degrades significantly as $k$ is reduced.

}%!iftoggle{fullversion}

\noindent
\textbf{Organization}.
The background and notation are defined in Section \ref{sec:preliminaries}.
Sections \ref{sec:greedywalk}, \ref{sec:sparse} and \ref{sec:primaldual}
cover \textsc{GreedyWalk}, \textsc{GreedyWalkSparse} and
\textsc{PrimalDual} algorithms, respectively, for the \textsc{SMRE} problem; the \textsc{SRMN} problem is discussed in section~\ref{sec:extensions}. Some of the algorithmic details and proofs are omitted for brevity and are available in \cite{spectral-approx-extended}.
Lower bounds for some heuristics and the experimental results
are discussed in Sections \ref{sec:lb} and \ref{sec:experiments}, respectively.
We discuss the related work in Section \ref{sec:related}
and conclude in Section \ref{sec:conc}.

%%
\iffalse
\begin{figure}[ht]
\centering
\includegraphics[width=.48\textwidth]{fig/labelled_graph.eps}
\caption{An example of a portion of a labeled contact graph where
nodes representing people are assigned labels (age, gender, set of
locations) while edges are assigned (contact time intervals in hours). This network can represent the
following scenario: the leftmost node corresponds to a school teacher teaching
10-year-olds and each child is linked to two nodes representing
their parents. $S$ represents school. Also, two parents share a common work place.}
\label{fig:labelled_graph}
\end{figure}
\fi

\section{Preliminaries}
\label{sec:preliminaries}
\begin{table}[ht]
\centering
\caption{Notations}
\label{table:notations}
\footnotesize
\begin{tabular}{p{1.8cm} p{5.8cm}} \toprule
$G=(V,E)$ & Graph representing a contact network\\
$n=|V|$ & Total number of nodes in $G$\\
$d(v,G)$ & Degree of node $v$ in $G$\\
$\Delta(G)$ & Maximum node degree in $G$\\
$A=A^G$ & Adjacency matrix of $G$\\
$G[E']$ & Subgraph of $G$ induced on $E'\subseteq E$\\
$\lambda_i(G)$ & $i$th largest Eigenvalue of $A^G$\\
$\rho(G)=\rho(A)$ & \iffalse$\max_i \lambda_i(G)$\fi$\lambda_1(G)$, spectral radius of $G$\\\midrule
$c(\cdot)$ & Cost of a vertex or edge of $G$\\
$\beta$ & Infection rate\\
$\delta$ & Recovery rate\\
$T$ & Epidemic Threshold, $T=\frac{\delta}{\beta}$\\
$\tau$ & Time to epidemic extinction\\\midrule
$\mathcal{W}_k(G)$ & Set of closed walks of length $k$ in $G$ \\
$W_k(G)$ & $W_k(G) = |\mathcal{W}_k(G)|$ \\
$\nodes(w)$ & number of distinct nodes in walk $w$ \\
$\walks(x,G,k)$ & Number of closed $k$-walks in $G$ containing edge (or
vertex) $x$\\
%% $\walks(v,G,k)$ & Number of closed $k$-walks in $G$ containing $v\in V$\\
$\eopt(T)$ & Optimal solution to \textsc{SRME}$(G, c(\cdot), T)$\\\bottomrule
\end{tabular}
\normalsize
\vspace{-.4cm}
\end{table}
We consider undirected graphs $G=(V,E)$, and interventions to control the spread of epidemics---
vaccination (modeled by removal of nodes) and quarantining (modeled by removal of edges).
There can be different costs for the removal of nodes and edges (denoted by $c(v)$ and
$c(e)$, respectively), e.g., depending
on their demographics, as estimated by \cite{medlock2009}. For a set $E'\subset E$,
$c(E')=\sum_{e\in E'} c(e)$ denotes the total cost of the set $E'$ (similarly for node subsets).
%We extend these notions
%and the problems we study to labeled graphs (see Figure \ref{fig:labelled_graph});
%however, for notational simplicity, we
%focus the discussion on unlabeled graphs, and explain how our formulations
%and results extend.

There are a number of models for epidemic spread; we focus on
the fundamental SIS (Susceptible-Infectious-Susceptible) model, which is defined in the following manner.
Nodes are in susceptible (S) or infectious (I) state.
Each infected node $u$ (in state I) causes each susceptible neighbor $v$ (in state S)
to become infected at rate $\beta_{uv}$. Further, each infected node $u$ switches to
the susceptible state at rate $\delta$. In this paper\iffalse most of the paper\fi, we assume a uniform
rate $\beta_{uv}=\beta$ for all $(u,v)\in E$; in this case, we define a threshold $T=\delta/\beta$,
which characterizes the time to extinction.
Let $A=A^G$ denote the adjacency matrix of $G$, and let $n=|V|$. Let
$\lambda_i(G)$ denote the $i$th largest eigenvalue of $A$, and let
$\rho(A)=\lambda_1(A)$ denote the spectral radius of $A$. Since $G$ is undirected, it follows
that all eigenvalues are real, and $\rho(A)>0$ (see, e.g., Chapter 3 of \cite{vanmieghem-spectral-graphs}).
Ganesh et al. \cite{ganesh+topology05} showed that the epidemic dies out in time
$O(\frac{\log{n}}{1-\rho(A)/T})$, if $\rho(A)<T$ in the SIS model, with high probability;
this threshold was also observed by \cite{Wang03Epidemic}.
Prakash et al. \cite{aditya12} show this condition holds for a broad class of other epidemic models,
including the SIR model (which contains the `Recovered' state). Now we
formally define the \textsc{SRM} problem.
\begin{Definition}{\prob{} problems (SRME and SRMN):}
Given an undirected graph $G=(V, E)$, with cost $c(e)$ for each edge
$e$, and a threshold $T$, the goal of the \textsc{SRME}$(G, c(\cdot),
T)$ problem is to find the cheapest subset $E'\subseteq E$ such that
$\lambda_1(G[E\setminus E'])<T$. We refer to the node version of this
problem as \textsc{SRMN}$(G, c(\cdot),T)$.
\end{Definition}
%% Similarly, we also consider the node version of this problem, denoted by \textsc{SRMN}.
We discuss some notation that will be used in the rest of the paper.
$\eopt(T)$ denotes an optimal solution to the \textsc{SRME}$(G, c(\cdot), T)$ problem.
%%
%%
%% \iffalse
%% \begin{table}[ht]
%% \centering
%% \caption{Different versions of the \textsc{SRM} problems studied in
%% this paper, and the algorithms which can be adapted for them.}
%% \label{table:srm-problems}
%% \small
%% \begin{tabular}{p{2.2cm} p{3.2cm} p{1.9cm}} \toprule
%% Problem variant & Specification & Approx. algorithms\\ \midrule
%% \textsc{SRME} & Unlabeled, edge removal & All algorithms\\ \midrule
%% \textsc{SRME-Labeled} & Edge labeled, edge removal & \textsc{GreedyWalk}, \textsc{PrimalDual}\\ \midrule
%% \textsc{SRME-nonuniform} & Unlabeled, edge removal, non-uniform
%% transmission rates & \textsc{GreedyWalk}\\ \midrule\midrule
%% \textsc{SRMN} & Unlabeled, node removal & \textsc{GreedyWalk}, \textsc{PrimalDual}\\ \midrule
%% \textsc{SRMN-Labeled} & Node labeled, node removal & \textsc{GreedyWalk}, \textsc{PrimalDual}\\ \bottomrule
%% \end{tabular}
%% \end{table}
%% \fi
%%
Let $\mathcal{W}_k(G)$ denote the set
of closed walks of length $k$ in $G$; let $W_k(G) = |\mathcal{W}_k(G)|$. For a walk $w$,
let $\nodes(w)$ denote the number of distinct nodes in $w$.
A standard result (see, e.g., Chapter 3 of \cite{vanmieghem-spectral-graphs}) is the following:
\begin{align}
\sum_{w\in \mathcal{W}_k(G)} \nodes(w) = \sum_i A^k_{ii}=\sum_{i=1}^n\lambda_i(G)^k.
\label{eqn:kthmoment}
\end{align}
%%
%It is easy to see that $\sum_{w\in \mathcal{W}_k(G)} \nodes(w) =\sum_i A^k_{ii}$.
The number of walks in $\mathcal{W}_k(G)$ containing a node $i$ is $A^k_{ii}$.
For a graph $G$, let $\walks(e, G,k)$ denote the number of closed $k$-walks
in $G$ containing $e=(i,j)$. Then, $\walks(e,G,k)=A^{k-1}_{ij}$. We say
that an edge set $E'$ hits a walk $w$ if $w$ contains an edge from $E'$.
Similarly, for a node $v$, let $\walks(v, G, k)$ denote the number of closed $k$-walks in $G$ containing
$v$. Then, $\walks(i, G,k) = A^k_{ii}$.
Table~\ref{table:notations} summarizes the frequently
used notations.
%%
%%
%% %Since $\walks(e, G) = W_k(G) - W_k(G[E\setminus\{e\}])$, it can be computed easily.
%% %It is easy to compute $W_k(G)$ in time polynomial in $n$ and $k$, since $W_k(G)=\sum_{i=1}^n A^k_{ii}$.
%% %
%%
%% %%%%%%%%%%%%%%
%% \iffalse
%% %%% not seems to be used here. check in lower bound section
%% The following bounds for the first eigenvalue are used frequently in the
%% paper (see \cite{vanmieghem-spectral-graphs}):
%% \begin{itemize}
%% \item
%% $\sqrt{\Delta(G)}\leqslant\lambda_1(G)\leqslant\Delta(G)$
%% \item
%% $\lambda_1(G)\leqslant\max_{u\sim v}\sqrt{\degr(u,G)\degr(v,G)}$, where
%% $u\sim v$ means $u$ is adjacent to $v$.
%% \end{itemize}
%% \fi
%% %%%%%%%%%%%%
%%
%% \iffalse
%% For an edge $(i,j)$, $c(i,j)$
%% denotes the cost of removing $(i,j)$ and for a set of edges $A$, $c(A)$
%% is the sum of the costs of edges of $A$.
%% \fi
%% %%
%%
\section{\textsc{GreedyWalk}: $O(\log {n} \log \Delta)$-approximation}
\label{sec:greedywalk}

\noindent
\textbf{Main idea}. Our starting point is
the connection between the number of closed walks in
a graph and the sum of powers of the eigenvalues in (\ref{eqn:kthmoment}). We try to
reduce the spectral radius by reducing the number of closed walks of length $k$
in the graph, by removing edges (see Algorithm 1\iffalse note this \fi). This, in turn, can be viewed as a partial covering
problem.\footnote{This is a variation of the set cover problem, in which an instance
consists of (i) a set $H$ of elements, (ii) a collection $\mathcal{S}=\{S_1,\ldots, S_m\}\subseteq 2^H$
of sets, (iii) cost$(S_i)$ for each $S_i\in \mathcal{S}$, and (iv) a parameter $r\leq |H|$. The
objective is to find the cheapest collection of sets from $\mathcal{S}$ which cover at least $r$ elements.
Slavic \cite{slavic:ipl97} shows that a greedy algorithm gives an
$O(\log{|H|})$ approximation.}
Our basic idea extends to other versions, as discussed later in Section
\ref{sec:extensions}.
%%%that achieves an $O(\log n \log \Delta)$ approximation
%%%in edge removal cost, while exceeding the spectral threshold by at
%%%most a factor of $(1 + \eps)$, for $\eps > 0$ that can be made
%%%arbitrarily small.  We then extend the algorithm and the bounds to the
%%%node version, and the versions for labeled graphs.
%%%Our algorithms are
%%%based on a greedy approach, that exploits the relationship
%%%(\ref{eqn:kthmoment}), and reduces this problem to a partial covering problem.
%%
\begin{algorithm}{}
\label{alg:greedywalk}
\caption{\textsc{GreedyWalk} (high level description)}
\small
\begin{algorithmic}[1]
\INPUT $G$, $T$, $c(\cdot)$, $k$ even\\
\OUTPUT Edge set $E'$\\
\STATE Initialize $E'\leftarrow \phi$\\
\WHILE{$W_k(G[E\setminus E'])\geq nT^k$}
\STATE $r \leftarrow W_k(G[E\setminus E']) - nT^k$\\
\STATE Pick $e\in E\setminus E'$ that maximizes $\frac{\min \{r, \walks(e, G[E\setminus E'],k)\}}{c(e)}$\\ \label{line:pickEdge}
\STATE $E'\leftarrow E'\cup\{e\}$
\ENDWHILE
\end{algorithmic}
\end{algorithm}

The Lemma below proves the approximation bound for any solution (say $E'$) from \textsc{GreedyWalk}.
Let $G'=G[E\setminus E']$ denote the graph resulting after the removal
of edges in $E'$.  Our proof involves three steps:
(1) Proving the bound on $\lambda_1(G')$;
(2) Relating $c(E')$ to the cost of
the optimum solution to the partial covering problem
which ensures that the number of walks in the residual graph is at most $nT^k$;
(3) Showing that the optimum solution to the \textsc{SRME} problem also ensures that
at most $nT^k$ remain in the residual graph.

\begin{lemma}
\label{lemma:greedywalk}
Let $E'$ denote the set of edges found by Algorithm
\textsc{GreedyWalk}.  Given any constant $\epsilon > 0$, let $k$ be an even integer
larger than $\frac{\log{n}}{\log(1+\epsilon/3)}$.
%%%larger than $(1+\epsilon/2)\log{n}$.
Then, we have $\lambda_1(G[E\setminus E']) \leq (1+\epsilon)T$,
and $c(E') = O(c(\eopt(T))\log n\log \Delta)$.
\end{lemma}
\begin{proof}
\hide{
Let $G'=G[E\setminus E']$ denote the graph resulting after the removal
of edges in $E'$.  Our proof involves three steps:
(1) Proving the bound on $\lambda_1(G')$;
(2) Relating $c(E')$ to the cost of
the optimum solution to the partial covering problem
which ensures that the number of walks in the residual graph is at most $nT^k$;
(3) Showing that the optimum solution to the \textsc{SRME} problem also ensures that
at most $nT^k$ remain in the residual graph.
} We follow the proof scheme mentioned above.
By the stopping condition of the algorithm, we have $W_k(G')\leqslant nT^k$.
%%%Since $E'$ hits at least $W_k(G)-nT^k$, the number of closed $k$-walks in $G'$, $W_k(G')\leqslant nT^k$.
From (\ref{eqn:kthmoment}), we have $\sum_{i=1}^n\lambda_i(G')^k =\sum_i A_{ii}^k = \sum_{w\in\mathcal{W}(G')}
\nodes(w)\leq kW_k(G')$, which implies $\sum_{i=1}^n\lambda_i(G')^k\leqslant nkT^k$.
Further, since $k$ is even (by assumption), $\lambda_i(G')\geq 0$,
so that $\lambda_1(G')^k\leq \sum_{i=1}^n\lambda_i(G')^k\leq nkT^k$.
This implies $\lambda_1(G')\leqslant e^{(\log n+\log{k})/k}T$. Since $k=\log{n}/\log{(1+\epsilon/3)}$,
we have
$(\log n+\log{k})/k\leq 2\log{(1+\epsilon/3)}$, so that $\lambda_1(G')\leq (1+\epsilon/3)^2T\leq (1+\epsilon)T$.
%%$\leqslant (1+\epsilon)T$ for $k\geqslant (\log n)/\log(1+\epsilon)$ and $n$
%%sufficiently large.

Next, we derive a bound for $c(E')$. Observe that the algorithm can be viewed as solving a partial
cover problem, in which (i) the set $H$ of elements corresponds to walks in $\mathcal{W}_k(G)$, and
(ii) there is a set corresponding to each edge $e\in E$ consisting of all the walks in
$\mathcal{W}_k(G)$ that contain $e$.  Following the analysis of the greedy
algorithm for partial cover \cite{slavic:ipl97}, we have
$c(E')=O(c(E_\text{HITOPT}) \log|H|)$, where $E_\text{HITOPT}$
denotes the optimum solution for this covering instance.
Since $\Delta$ denotes the maximum node degree, we have $H=W_k(G) \leqslant n{\Delta}^k$. We show below
that $c(E_\text{HITOPT})\leq c(\eopt(T))$; it follows that
$c(E')=O(c(\eopt(T))\log n\log \Delta)$.

Finally, we prove that $c(E_\text{HITOPT})\leq c(\eopt(T))$. By definition of
$\eopt(T)$, we have $\lambda_1(G[E-\eopt(T)])\leq T$. Let $G''=G[E-\eopt(T)]$. Then,
we have
\begin{align*}
W_k(G'')\leq\sum_{i=1}^n\lambda_i(G'')^k<n\lambda_1(G'')^k\leq nT^k.
\end{align*}
This implies $\eopt(T)$ hits at least $W_k(G)-nT^k$ walks, so that $c(E_\text{HITOPT})\leq c(\eopt(T))$.
\end{proof}

\noindent
\textbf{Effect of the walk length $k$}.
We set the walk length $k=a\log{n}$ for some constant $a$ in Algorithm \textsc{GreedyWalk}; understanding the
effect of $k$ is a natural question. From the proof of Lemma \ref{lemma:greedywalk}, it follows
that $\lambda_1(G[E\setminus E'])$ can be bounded by $(nk)^{1/k}T$ for any choice of $k$,
as long as it is even. This bound becomes worse as $k$ becomes smaller, e.g., it is
$O(\sqrt{n})$ for $k=2$.
\iftoggle{fullversion}
{This is borne out in the experiments in Section \ref{sec:experiments}.}
{}

In order to complete the description of \textsc{GreedyWalk} (Algorithm 1 \iffalse note this \fi),
we need to design an efficient method to determine the edge which maximizes the
quantity in line 4. We discuss two methods below.

%%\subsection{\textsc{GreedyWalk} in small networks using matrix multiplication}
\subsection{Matrix multiplication approach for implementing \textsc{GreedyWalk}.}
\label{sec:matrix}
%\iffalse
%We discuss two methods: method one uses a
%naive approach based on Matrix multiplication while the other uses dynamic
%programming and is particularly faster on sparse graphs.
%
%%\subsubsection{Direct approach based on matrix multiplication}
%Implementation of Algorithm \textsc{GreedyWalk} requires the computation of $W_k(G)$.
%We first describe a matrix multiplication based approach.
%All walks in $\mathcal{W}_k(G)$ that
%do not contain node $i$ must be walks in $\mathcal{W}_k(G[V-\{i\}])$. Therefore, we have the
%following recurrence,
%\[
%W_k(G) = A(G)^k_{nn} + W_k(G[V-\{n\}]),
%\]
%which can be computed by a simple dynamic program, using $n$ matrix multiplications;
%we refer to this dynamic programming algorithm by \textsc{CountWalks}$(G)$.
%\comment{Is the above paragraph required?}
%
%\noindent
%\emph{Running time}.
%\fi
%%
Note that $\walks(e, G,k)=A^{k-1}_e$.  We use matrix multiplication to compute
$A^{k-1}$ once for each iteration of the while loop in line 2 of
Algorithm 1 \iffalse note this \ref{alg:greedywalk}\fi. In line 4, we iterate over all edges, in order to
compute the edge $e$ that maximizes the given ratio.
For $k=O(\log n)$, $A^{k-1}$ can be computed in time $O(n^{\omega}\log\log{n})$, where
$\omega< 2.37$ is the exponent for the running time of the best matrix
multiplication algorithm~\cite{williams2012multiplying}.
Therefore, each iteration involves $O(n^{\omega}\log\log{n} + m)=O(n^{\omega}\log\log{n})$ time.
This gives a total running time of $O(n^{\omega}\log\log{n}|E_{OPT}| \log^2 n)$,
since only $O(|E_{OPT}| \log^2n)$ edges are removed.
One drawback with this approach is the high (super-linear) space complexity,
even with the best matrix multiplication methods, in general.
%%
%% \subsection{\textsc{GreedyWalk} in large sparse networks using a dynamic programming approach}

\subsection{Dynamic programming approach for implementing \textsc{GreedyWalk}.}
\label{sec:dynamic}
When the graphs are very sparse ($\Theta(n)$ edges), we adapt a dynamic programming
approach to compute $\walks(e, G,k)$ for an edge $e$ and more efficiently select the edge that
maximizes $\walks(e, G[E\setminus E'],k)/c(e)$ in line 4 of Algorithm 1 \iffalse \ref{alg:greedywalk} note this \fi.
Although, potentially $\walks(e, G,k)$ needs to be computed for each edge $e\in E\setminus E'$, in practice it suffices to compute it for only a small subset of $E\setminus E'$. We make use of the fact that
$\walks(e,G',k)\le\walks(e,G,k)$ for any subgraph $G'$. The approach is briefly as follows. Initially we compute $\walks(e, G,k)$ for each $e\in E$ and arrange the edges in non-ascending order of their $\walks(e, G,k)$ value, $e_1,e_2,...,e_{|E|}$. After the first edge ( i.e. $e_1$ in the first iteration) is removed, $\walks(e, G',k)$ is computed on the residual graph $G'$ only for some consecutive edges in that order upto some $e_i$ such that $\walks(e_i,G',k)> \walks(e_{i+1},G,k)$. Edges $e_2,...,e_i$ are reordered based on the recomputed walk numbers, $\walks(e_i,G',k)$ and then the same steps are repeated. The approach takes $O(n)$ space and $O(n^2 k)$ time assuming the number of edges is $\Theta(n)$ in real world large networks.
\iftoggle{fullversion}
{The detailed algorithm and the analysis is given in the appendix \ref{appsubsec:dp}.}
%{The detailed algorithm and the analysis is given in the full version \cite{spectral-approx-extended}.}
{The algorithmic details and proofs are omitted for brevity}

%\subsection{Varying walk length}
%Discussion on what happens when we vary $k$

\section{Using sparsification for faster running time: Algorithm \textsc{GreedyWalkSparse}}
\label{sec:sparse}
%%
%% \iffalse
%% The efficiency of Algorithm \textsc{GreedyWalk} can be improved
%% if the number of edges in the graph can be reduced. We now discuss
%% Algorithm \textsc{GreedyWalkSparse} that combines two pruning steps with
%% \textsc{GreedyWalk}, leading to sparser graphs, without affecting the
%% asymptotic approximation guarantees. The algorithm is given below. It
%% refers to the $T$-core of a graph which denotes the the maximal subgraph
%% of $G$ with minimum degree $T$ (see, e.g., \cite{kcore}).
%% \fi
The efficiency of Algorithm \textsc{GreedyWalk} can be improved
if the number of edges in the graph can be reduced. This can be achieved by
   two pruning steps - pruning edges such that in the residual graph (i)~no node has degree more than $T^2$, and (ii) there is no $T$-core; the $T$-core of a graph denotes the maximal subgraph
of $G$ with minimum degree $T$ (see, e.g., \cite{kcore}). We will refer to these steps as \textsc{MaxDegreeReduction} and \textsc{DensityReduction} respectively. This leads to sparser graphs, without affecting the
asymptotic approximation guarantees.
\iftoggle{fullversion}
{
The algorithm involves two prunning steps: \textsc{MaxDegreeReduction} and \textsc{DensityReduction}; the procedure is described in Algorithm \textsc{GreedyWalkSparse}.

\begin{algorithm}
\label{alg:GreedyWalkSparse}
\caption{Algorithm \textsc{GreedyWalkSparse}}
\begin{algorithmic}[1]
\INPUT $G, T, c(\cdot)$
\OUTPUT Edge set $E'$
\STATE Initialize $G_r=G$.\\

\STATE \tcp{\textbf{Pruning step~1}: \textsc{MaxDegreeReduction}}\\
\STATE  Let $V_{T^2}=\{v: \degr(v,G)\geqslant T^2\}$.
\FOR {$v\in V_{T^2}$}
\IF{$\degr(v,G_r)\ge T^2$}
\STATE  Let $e_{v,1},\ldots,e_{v,\degr(v,G_r)}$ be the edges incident on $v$ ordered so that $c(e_{v,1})\leq\ldots \le c(e_{v,\degr(v,G_r)})$.
\STATE  Let $E_v=\{e_{v,1}, \dots, e_{v, \degr(v,G_r)-T^2+1}\}$.\\
\STATE  $E_1\leftarrow E_1\cup E_v$ and $E(G_r)\leftarrow E(G_r)\setminus E_v$.\\
\ENDIF
\ENDFOR
\STATE \tcp{\textbf{pruning step~2}: \textsc{DensityReduction}}
\STATE Let $C_T$ denote the $T$-core of $G_r$.\\
\STATE Order the edges $e_1,\ldots,e_{|E(C_T)|}$ in non-decreasing order of cost.\\
\STATE $E_2\leftarrow \{e_i\mid i\le |E(C_T)|-T|V(C_T)|/2+1\}$\\
\STATE \tcp{\textbf{\textsc{GreedyWalk} on Pruned Graph}:}
\STATE $E(G_r)\leftarrow E(G_r)-E_1-E_2$\\
\STATE $E_3=\textsc{GreedyWalk}(G_r,T,c(\cdot))$\\
\STATE $E'\leftarrow E_1\cup E_2\cup E_3$
\end{algorithmic}
\end{algorithm}
} %iftoggle{fullversion}
%{The algorithm is described in the full version \cite{spectral-approx-extended}.\iffalse appendix \ref{appsubsec:greedywalksparse}.\fi}
{
%\textcolor{red}{The algorithmic details are omitted for brevity.}
}
%%
%%
%%\iffalse
%%We discuss two pruning steps to sparsify the graph, which do not
%%affect the performance guarantees, and allows us to run Algorithm
%%\textsc{GreedyWalk} on a sparser graph.
%%\fi
%%
%%
\begin{lemma}
\label{lemma:prune}
Let $E_1$ and $E_2$ denote the set of edges removed in the pruning steps
\textsc{MaxDegreeReduction} and \textsc{DensityReduction}, respectively.
Then, $c(E_1)$ and $c(E_2)$ are both at most $2c(\eopt(T))$.
\end{lemma}
\iftoggle{fullversion}
{
\begin{proof}
Since $\sqrt{\Delta(G')}\leqslant\lambda_1(G')$
\cite{vanmieghem-spectral-graphs}, \iffalse as mentioned in Section
\ref{sec:preliminaries}, \fi
which implies $\Delta(G[E-\eopt(T)])\leq T^2$. Therefore,
$c(\{e\in N(v)\cap \eopt(T)\})\geq \sum_{j=1}^{\degr(v,G)-T^2+1} c(e_{v,j})$, where the
sum is the minimum cost of edges that can be removed to ensure that the degree of $v$
becomes at most $T^2$. Therefore,
\begin{align*}
c(E_1) &= \sum_{v\in V_{T^2}} \sum_{j=1}^{\degr(v,G)-T^2+1} c(e_{v,j})\\
&\leq \sum_{v\in V_{T^2}} c(\{e\in N(v)\cap \eopt(T)\})\\
&\leq c(\eopt(T))
\end{align*}
Recall that the second pruning step is applied on $G_r$.
For bounding $c(E_2)$, we use another lower bound for $\lambda_1$:
for any induced subgraph $H$ of $G_r$, $\sum_{v\in V(H)} \frac{d(v,
   H)}{|V(H)|}\leq \lambda_1(G_r)$.
Therefore, the existence of a $T$-core $C_T$ implies that $\lambda_1(G_r)\geq T$.
Since the average degree of $C_T$ in the residual graph
is at least $T$, it implies that at least
%must be at most $T$, at least implies
$|E(C_T)|-T|V(C_T)|/2+1$ edges must be removed from
$C_T$. Therefore,
\begin{align*}
c(E_2)=\sum_{j=1}^{|E(C_T)|-T|V(C_T)|/2+1} c(e_j)\le c(\eopt(T)\cap
E(C_T))\,,
\end{align*}
where, the $e_j$ correspond to the first $|E(C_T)|-T|V(C_T)|/2+1$ edges of
least cost. Hence proved.
\end{proof}
}
%{The proof of the lemma is given in the full version \cite{spectral-approx-extended}.\iffalse appendix \ref{subsec:lemma-prune}.\fi}
{The proof of the lemma and algorithmic details of sparsification are omitted for brevity.}
By Lemma \ref{lemma:prune}, it follows that the approximation bounds of
Lemma \ref{lemma:greedywalk} still hold. However, the pruning steps
reduce the number of edges, thereby speeding the implementation of
\textsc{GreedyWalk}. We discuss the empirical performance of pruning in
Section \ref{sec:experiments}. We show below that pruning also improves the
approximation factor marginally from $O(\log n\log \Delta)$ to $O(\log
n\log T)$ which could be significant when $n$ is large and
$T\ll\Delta$.
\begin{lemma}
\label{lemma:greedywalksparse}
Let $E'$ denote the set of edges found by Algorithm
\textsc{GreedyWalkSparse}.  Given any constant $\epsilon > 0$, let $k$ be
an even integer larger than $\frac{\log{n}}{\log(1+\epsilon/3)}$.
Then, we have $\lambda_1(G[E\setminus E']) \leq (1+\epsilon)T$,
and $c(E') = O(c(\eopt(T))\log n\log T)$.
\end{lemma}
\iftoggle{fullversion}
{
\begin{proof}
From Lemma~\ref{lemma:prune}, the number of edges removed is at most
$2c(\eopt)$. The residual graph $G_r$ has maximum degree less than $T^2$.
Therefore, applying Lemma~\ref{lemma:greedywalk} on $G_r$, it follows
that the number of edges removed is $O(c(\eopt(T))\log n\log T)$. Hence,
the total number of edges removed by \textsc{GreedyWalkSparse} is at
most $2c(\eopt(T))+O(c(\eopt(T))\log n\log T)=O(c(\eopt(T))\log n\log T)$.
\end{proof}

}
%{The proof of this lemma is presented in the full version \cite{spectral-approx-extended}.\iffalse appendix \ref{subsec:lemma-prune}.\fi}
{The proof is omitted for brevity.}
%%%%%%%%%%
%% moved to appendix
\iffalse
\begin{proof}
From Lemma~\ref{lemma:prune}, the number of edges removed is at most
$2c(\eopt)$. The residual graph $G_r$ has maximum degree less than $T^2$.
Therefore, applying Lemma~\ref{lemma:greedywalk} on $G_r$, it follows
that the number of edges removed is $O(c(\eopt(T))\log n\log T)$. Hence,
the total number of edges removed by \textsc{GreedyWalkSparse} is at
most $2c(\eopt(T)+O(c(\eopt(T))\log n\log T)=O(c(\eopt(T))\log n\log T)$.
\end{proof}
\fi
%%%%%%%%%%

\section{\textsc{PrimalDual}: $O(\log{n})$-approximation}
\label{sec:primaldual}
\noindent
\textbf{Main idea}:
The approach of \cite{gandhi:ja04} gives an $f$-approximation for the partial
covering problem, where $f$ denotes the maximum number of sets that contain any element
in the set system.  As in the proof of Lemma \ref{lemma:greedywalk},
in our reduction from the \textsc{SRME} problem to partial covering, elements correspond
to all the closed walks of length $k=O(\log{n})$, while sets correspond to edges; for an edge $e$,
the corresponding set $S_e$ consists of all the walks $w$ that are hit by $e$. In this reduction,
each walk $w$ lies in $k$ sets; therefore, $f=O(\log{n})$ for this set system. Therefore,
the approach of \cite{gandhi:ja04} could improve the approximation factor.  Unfortunately,
our set system has size $n^{O(\log{n})}$, so that the algorithm of \cite{gandhi:ja04} cannot
be used directly to get a polynomial time algorithm.

\begin{algorithm}{}
\label{alg:primaldual}
\caption{\textsc{PrimalDual}$(\mathcal{T}', \mathcal{S}', c', \sigma')$}
\small
\begin{algorithmic}[1]
\OUTPUT Edge set $E''$
\STATE Initialize $z_e=0$ for all $S_e\in \mathcal{S}'$, $C\leftarrow\phi$.\\
\STATE \tcp{$u(w)=0$ for all walks $w$ in $G'$.}\\
\WHILE{$C$ is not $\sigma'$-feasible}
\STATE {$x = \min_{e\in E\setminus E''}\{ \frac{c(e) - z_e}{\walks(e, G',k)}\}$; let $e$ be
an edge for which the minimum is reached.}\\
\STATE $C\leftarrow C\cup\{S_e\}$\\
\STATE For each $e'\in E\setminus E''$: $z_{e'} = z_{e'} + x\cdot \walks(e', G',k)$\\
\STATE \ \ \ \tcp{$u(w) = u(w)+x$ for all walks $w$ in
$G'$ that pass through $e'$}\\
\STATE $E'\leftarrow E''\cup \{e\}$
\ENDWHILE
\end{algorithmic}
\end{algorithm}

The algorithm of Gandhi et al. \cite{gandhi:ja04} uses a primal-dual approach,
which maintains dual variables $u(w)$ for each element (i.e., walk); these are
increased gradually, and a set (i.e., an edge) is picked if the sum of duals corresponding
to the elements in the set equals its cost.
We now discuss how to adapt this algorithm to run in polynomial time, and only focus on
polynomial time implementation of the \textsc{PrimalDual} subroutine of
\cite{gandhi:ja04} in detail here. However, we also present the set cover
algorithm \textsc{HitWalks} for completeness. This algorithm iterates over
all edges and invokes \textsc{PrimalDual} in each iteration to obtain a candidate
set of edges to remove and finally chooses the set with minimum cost.
$\mathcal{T}'$, $\mathcal{S}'$, $c'$ and  $\sigma'$ denote the set of elements (walks) to be covered,
the sets (corresponding to edges that can be chosen), the costs corresponding to the sets/edges
and the number of elements (walks) that need to be covered, respectively. A subset $C\subseteq \mathcal{S}'$
is $\sigma'$-feasible if $|\cup_{S_e\in C} S_e|\geq \sigma'$. Let $u(w)$ denote the dual variables
corresponding to the walks $w$; these are not maintained in the algorithm explicitly,
but assigned in the comments, for use in the analysis.
\begin{algorithm}{}
\label{alg:setcover}
\caption{\textsc{HitWalks}$(\mathcal{T}, \mathcal{S}, c, \sigma)$}
\small
\begin{algorithmic}[1]
\INPUT Set of all $k$-closed walks $\mathcal{T}$, walks corresponding to
edges $\mathcal{S}$, edge cost set $c$,  number of walks to hit $\sigma$
\OUTPUT Edge set $E'$
\STATE Sort the edges of $G$ in increasing order of their costs.\\
\STATE Initialize $\forall j$, $c'(e_j)\leftarrow\infty$\\
\FOR {$j\leftarrow 1$ to $m$}
\STATE   $c'(e_j)\leftarrow c(e_j)$ and compute $\walks(e_j,G,k)$\\
\STATE   \parbox{.45\textwidth}{$cs_j\leftarrow\infty$. \tcp{cost of edge set in this iteration}}
\IF{$|S_1\cup S_2\cup\cdots S_j|\geqslant \sigma$}
\STATE $E'_j~=~\{e_j\}~\cup~\textsc{PrimalDual}\big(\mathcal{T}~\setminus~S_j,~\mathcal{S}~\setminus~S_j,c',\sigma-\walks(e_j,G,k)\big)$\\
\STATE      $cs_j=c(E'_j)$\\
\ENDIF
\STATE   $i=\min_jcs_j$\\
\STATE   $E'=E'_i$
\ENDFOR
\end{algorithmic}
\end{algorithm}
%%%%%%%%%%%%%%%%%%%%%%%%
%%
This algorithm does not explicitly update the dual variables, but the edges
are picked in the same sequence as in \cite{gandhi:ja04}.
\begin{lemma}
Given any constant $\epsilon > 0$, let $k$ be
an even integer larger than $\frac{\log{n}}{\log(1+\epsilon/3)}$.
The dual variables $u(w)$ in algorithm \textsc{PrimalDual} are
maintained and updated as in \textup{\cite{gandhi:ja04}}, and the edge $e$
picked in each iteration is the same.
We have $c(E')=O(c(E_{OPT})\log{n})$ and $\lambda_1(G[E\setminus E'])\leq (1+\epsilon)T$.
\label{lemma:primaldual}
\end{lemma}
\iftoggle{fullversion}
{
\begin{proof}
Instead of updating the dual variable $u(w)$ for each element (walk) $w$, as done in \cite{gandhi:ja04}, the variable $z_e$ corresponding to each set (edge) $e$ is updated in algorithm \textsc{PrimalDual} at the end of each iteration. It is easy to see that, the following is an invariant at the end of each iteration, $z_e = \sum_{w\in S_e}u(w)$. Also note that, a set $e$ is picked into the cover in \textsc{PrimalDual}, whenever $z_e=c_e$.

Therefore, increasing the $u(w)$'s has the same effect as increasing the $z_e$'s in terms of picking the sets into the cover and both the algorithm \textsc{PrimalDual} and the one in \cite{gandhi:ja04} chooses the same set in each iteration.
\end{proof}
}
%{The proof of the lemma is given in the full version \cite{spectral-approx-extended}.\iffalse appendix \ref{appsubsec:primaldual}.\fi}
{The proof is omitted for brevity.}
%%
%\section{Extension}
\section{Node Version}
\label{sec:extensions}
\iftoggle{fullversion}
{
Our discussion so far has focused on the \textsc{SRME} problem. We now consider
extensions which capture two kinds of issues arising in practice.
%%
%\iffalse
%\begin{enumerate}
%\item
%\emph{Non-uniform transmission rates}:
%In general, the transmission rate $\beta$ is not constant for all the edges. The
%transmission rate $\beta_{ij}$ for edge $(i,j)$ depends
%on individual properties, especially the demographics of the end-points $i$ and $j$, such as
%age, e.g., \cite{medlock2009}. This gives us the \textsc{SRME-nonuniform} problem,
%which is defined later.  We extend the spectral radius characterization of
%\cite{ganesh+topology05,Wang03Epidemic,aditya12} to handle this setting, and show
%that \textsc{GreedyWalk} can also be adapted for solving \textsc{SRME-nonuniform},
%with the same guarantees.
%\item
%\emph{The node removal version (\textsc{SRMN} problem)}: We extend the \textsc{GreedyWalk}
%algorithm in a natural manner to work for \textsc{SRMN}, with the
%same approximation guarantees.
%\end{enumerate}
%\fi
%%
\paragraph{1.~Non-uniform transmission rates.}
In general, the transmission rate $\beta$ is not constant for all the edges. The
transmission rate $\beta_{ij}$ for edge $(i,j)$ depends
on individual properties, especially the demographics of the end-points $i$ and $j$, such as
age, e.g., \cite{medlock2009}.
Let $B=B(G)=(\beta_{ij})$ denote the matrix of the transmission rates.
This gives us the \textsc{SRME-nonuniform} problem,
which is defined as follows: Given an undirected graph $G=(V, E)$, with transmission rate $\beta_{ij}$ for each $(i,j)\in E$
and recovery rate $\delta$, find the smallest set $E'\subseteq E$ such that
$\rho(B(G[E-E'])) \leq\delta$.
We extend the spectral radius characterization of
\cite{ganesh+topology05,Wang03Epidemic,aditya12} to handle this setting, and show
that \textsc{GreedyWalk} can also be adapted for solving \textsc{SRME-nonuniform},
with the same guarantees.
\iftoggle{fullversion}
{The details of the algorithm, lemma and proofs are discussed in the appendix \ref{sec:non-uniform-appendix}.}
%{The details of the algorithm, lemma and proofs are discussed in the full version \cite{spectral-approx-extended}.}
{The algorithmic details and performance analysis are omitted for brevity.}
%%%Particularly, we show that the \textsc{GreedyWalk} adaption
%%%gives a solution with $O(\log n\log \Delta)$ approximation ratio while the spectral
%%%radius of the residual graph after edge removal is no more than
%%%$(1+\epsilon)\delta$, for a small constant $\epsilon>0$.
\paragraph{2.~The node removal version (\textsc{SRMN} problem).} We extend the \textsc{GreedyWalk}
algorithm in a natural manner to work for \textsc{SRMN}, with the
same approximation guarantees. For the details, please see the appendix \ref{appsubsec:srmn}.
}
{
Our discussion so far has focused on the \textsc{SRME} problem. We now consider
The \textsc{SRMN} problem. Recall the definition of $\walks(v, G, k)$ from Section \ref{sec:preliminaries}.
Let $G[V'']$ denote the subgraph of $G=(V, E)$ induced by subset $V''\subset V$.
We modify Algorithm \textsc{GreedyWalk} \iffalse \ref{alg:greedywalk} note this \fi to work for the \textsc{SRMN} problem
as shown in algorithm \textsc{GreedyWalkSRMN}.

\begin{algorithm}{}
\label{alg:srmn}
\caption{Algorithm \textsc{GreedyWalkSRMN}}
\begin{algorithmic}[1]
\STATE Initialize $V'\leftarrow \phi$
\WHILE {$W_k(G[V\setminus V'])\geq nT^k$}
\STATE $r \leftarrow W_k(G[E\setminus E']) - nT^k$
\STATE Pick $v\in V\setminus V'$ that maximizes $\frac{\min \{r, \walks(v, G[V\setminus V'],k)\}}{c(v)}$
\STATE $V'\leftarrow V'\cup\{v\}$
\ENDWHILE
\end{algorithmic}
\end{algorithm}

It can be shown on the same lines as Lemma \ref{lemma:greedywalk} that this gives
a solution of cost $O(c(\eopt(T))\log n\log \Delta)$, where $c(\eopt(T))$ denotes the
cost of the optimal solution to \textsc{SRMN} problem. Further, the same running time
bounds as in Sections \ref{sec:matrix} and \ref{sec:dynamic} hold.
}
\section{Popular heuristics and lower bounds}
\label{sec:lb}
A number of heuristics have been developed for controlling the spread of epidemics-- these
are discussed below. All these heuristics involve ordering the edges based on some kind
of score, and then selecting the top \iffalse$k$\fi few edges based on this score. We describe the score
function in each heuristic.
\begin{enumerate}[1.]

\item \textsc{ProductDegree} (\cite{vanmieghem:ton12}): The score for edge $e=(u,v)$ is defined as
$\deg(u)\times \deg(v)$.
Edges are removed in non-increasing order of this score.

\item \textsc{EigenScore} (\cite{vanmieghem:ton12, tong:cikm12}): Let $\mathbf{x}$ be the eigenvector corresponding to the
first eigenvalue of the graph. The score for edge $e=(u,v)$ is $|x(u)\times x(v)|$.

%% \iffalse
%% \item Method based on degree in linegraph (\textsc{LineDegree}): We set a degree score for an edge $(u,v)$, as its degree in the corresponding linegraph of the graph. We remove edges from the graph in non-increasing order of this score.
%% \fi

\item \textsc{LinePagerank}: This method uses the linegraph $L(G)=(E, F)$ of graph $G=(V,E)$,
where $(e,e')\in F$ if $e, e'\in E$ have a common endpoint. We define the score of edge
$e\in E$ as the pagerank of the corresponding node in $L(G)$.
%%
%% \iffalse
%% \item \textsc{Hybrid}: this heuristic picks the best of the
%% \textsc{EigenScore} and \textsc{ProductDegree} methods. The edges are ordered
%% in the following manner:
%% %\noindent
%% %Algorithm \textsc{Hybrid}
%% \begin{enumerate}[\textbullet]
%% \item
%% Let $\pi_1,\ldots,\pi_m$ and $\mu_1,\ldots,\mu_m$ be orderings of edges in the \textsc{Eigenscore}
%% and \textsc{ProductDegree} algorithms, respectively.
%% \item
%% Initialize $i=0$ and $j=0$.
%% \item
%% From the edges $\pi(i)$ and $\rho(j)$, remove the one
%% which decreases the max eigenvalue of the residual graph more. Increment the corresponding index.
%% \end{enumerate}
%% \fi
\end{enumerate}

As we find in Section \ref{sec:exp-results}, these heuristics work well for different
kinds of networks.
%These heuristics work quite well in practice, and we compare them with our
%algorithms in Section \ref{sec:experiments}.
We design another heuristic,
\textsc{Hybrid}, which picks the best of the
\textsc{EigenScore} and \textsc{ProductDegree} methods. The edges are ordered
in the following manner:
%\noindent
%Algorithm \textsc{Hybrid}
(1)~Let $\pi_1,\ldots,\pi_m$ and $\mu_1,\ldots,\mu_m$ be orderings of edges in the \textsc{Eigenscore}
and \textsc{ProductDegree} algorithms, respectively.
(2)~Initialize $i=0$ and $j=0$, and
(3)~from the edges $\pi(i)$ and $\rho(j)$, remove the one
which decreases the max eigenvalue of the residual graph more. Increment the corresponding index.

We have examined the worst case performance of these heuristics.
Two of these, namely, \textsc{EigenScore} and \textsc{ProductDegree}, have been used specifically
for reducing the spectral radius, e.g., \cite{vanmieghem:ton12, tong:cikm12}.
No formal analysis is known for any of these heuristics in the context of the \textsc{SRME} or
\textsc{SRMN} problems; some of them seem to work pretty
well on real world networks.
We show that the worst case performance of these heuristics can be quite poor, in general.
\begin{theorem}\label{thm:productdegree}
Given any sufficiently large positive integer $n$, there exists a
threshold $T'<a\sqrt{n}$, for some constant $a<1$ and a graph of size $n$
for which the number of edges removed by \textsc{ProductDegree},
\textsc{EigenScore}, \textsc{Hybrid} and \textsc{LinePagerank} is
$\Omega\big(\frac{n}{T'^2}\big)c(\eopt)$.
\end{theorem}
\iftoggle{fullversion}
{The proof is presented in appendix~\ref{sec:worstcaseproof}.}
%{The proof is presented in the full version \cite{spectral-approx-extended}.}
{We give here a proof sketch for theorem~\ref{thm:productdegree}. The detailed proof is omitted for brevity. We construct a graph $G$ for which the statement holds. For convenience
let us assume that $T'$ is a positive integer. $G$ contains (1)~a clique
$G_1$ on ${T'+1}$ nodes; (2)~a caterpillar tree $G_2$, which comprises of a path $v_1v_2\cdots v_{q-1}$
with $v_i$ adjacent to $T'$ leaves each and (3)~$G_3$, a star graph
with $(T'+1)^2$ leaves and central vertex denoted by $v_q$. We connect
$G_1$ to $G_2$ by $(v_0,v_1)$ where, $v_0$ is some node in $G_1$ and
$G_2$ is connected to $G_3$ by the edge $(v_q,v_{q-1})$. Note that
$q=\frac{n-(T'+1)^2-T'}{T'}$ and $\lambda_1(G)\ge\lambda_1(G_3)=T'+1$. Again,
here we assume that $q$ is an integer.

It can be shown that $c(\eopt)\le 2T'+3$.
Removing the edges $(v_0,v_1)$ and $(v_{q-1},v_q)$ isolates
the components $G_1$, $G_2$ and $G_3$. $G_1$ is a clique on $T'+1$
nodes and on removing one edge, its spectral radius decreases below
$T'$. $G_2$ is a star with $(T'+1)^2$ leaves and therefore, on removing
at most $(T'+1)^2-(T'^2+1)$ edges, its spectral radius decreases
below $T'$. It can be shown that $\lambda_1(G_2)\le \sqrt{T'}+2$.

It can be demonstrated that all the four algorithms score the edges
$(v_i,v_{i+1})$, $i=0,\ldots,q-2$ above any edge belonging to the clique
$G_1$. However, the spectral radius cannot be brought down below $T'$
until at least one edge in $G_1$ is removed. Therefore, at least $q$
edges will be removed by all the algorithms.  By the initial assumption
that $T'<c\sqrt{n}$, it follows that $q=\Omega\big(\frac{n}{T'}\big)$,
while, ${c(\eopt)}=O(T')$, hence the theorem holds.
}
%%%%%%%%%%
\iffalse
We construct a connected
graph $G$ which is composed of a ($T+1$)-clique $G_1$, a caterpillar tree
$G_2$ and a star graph $G_3$. We show that the optimal set of edges to remove to reduce the
spectral radius below $T$ is $O(T)$. Then we demonstrate that at least
$n/T$ edges of $G_2$ have higher score than the edges of $G_1$ in all the
four algorithms. However, at least one edge from $G_1$ has to be removed to
reduce the spectral radius below $T$. Hence, each method ends up removing at
least $\frac{n}{T}+1$ edges while $\opt=O(T)$.
\fi
%%
%%

\section{Experiments}
\label{sec:experiments}

\subsection{Methods and Dataset}

We evaluate the algorithms developed in the paper\footnote{All code at: \codeurl.}
-- \textsc{GreedyWalk},
\textsc{GreedyWalkSparse} and \textsc{PrimalDual} -- and compare their performance
with the heuristics from literature --
\textsc{EigenScore}, \textsc{ProductDegree}, \textsc{LinePagerank} and
\textsc{Hybrid} (described in Section~\ref{sec:lb}), as a more sophisticated baseline.
The networks which we considered in our empirical analysis are listed in Table~\ref{table:networks}
spanning infrastructure networks, social networks and random graphs.
%%
%\begin{table}[ht]
%\centering
%\caption{\small Networks and their sizes}
%\footnotesize
%\begin{tabular}{p{2.8cm} p{0.8cm} p{0.9cm}p{0.7cm}p{0.9cm}}%%{lrrrc}
%\toprule
%Network & nodes & edges & $\lambda_1$ & Source\\ \midrule
%Barabasi-Albert & $1000$ & $1996$ & $11.1$ & Synthetic \\
%Erdos-Renyi & $994$ & $2526$ & $6.38$ &  Synthetic \\
%P2P (Gnutella05) & $8846$ & $31839$ & $23.55$ & \cite{snap:01}\\
%P2P (Gnutella06) & $8717$ & $31525$ & $22.38$ & \cite{snap:01}\\
%Collab. Net (HepTh)  & $9877$ & $25998$ & $31.03$ & \cite{snap:01}\\
%Collab. Net (GrQc)  & $5242$ & $14496$ & $45.62$ & \cite{snap:01}\\
%AS (Oregon 1)  & $10670$ & $22002$ & $58.72$ & \cite{snap:01}\\
%AS (Oregon 2)  & $10900$ & $31180$ & $70.74$ & \cite{snap:01}\\
%%Irvine Social Net  & $1899$ & $13838$ & $48.14$ & \\
%Brightkite Net  & $58228$ & $214078$  & $101.49$ & \cite{snap:01}\\
%%Portland Cont. Net. & $1575861$ & $19481626$ & $27.23$ \iffalse $87.12$\fi & \cite{cinet:14}\\
%Youtube Network & $1134890$ & $2987624$ & $210.4$ & \cite{snap:01}\\
%Stanford Web graph & $281903$ & $1992636$ & $448.13$ & \cite{snap:01}\\ \bottomrule
%\end{tabular}
%\label{table:networks}
%\vspace{-0.18in}
%\end{table}

\begin{table}[ht]
\centering
\caption{\small Networks and their sizes. The first two are synthetic random networks; others are taken from \cite{snap:01} and \cite{cinet:14}}
\footnotesize
\begin{tabular}{p{2.8cm} p{0.8cm} p{0.9cm}p{0.7cm}}%%{lrrrc}
\toprule
Network & nodes & edges & $\lambda_1$ \\ \midrule
Barabasi-Albert & $1000$ & $1996$ & $11.1$  \\
Erdos-Renyi & $994$ & $2526$ & $6.38$  \\
P2P (Gnutella05) & $8846$ & $31839$ & $23.55$ \\
P2P (Gnutella06) & $8717$ & $31525$ & $22.38$ \\
Collab. Net (HepTh)  & $9877$ & $25998$ & $31.03$ \\
Collab. Net (GrQc)  & $5242$ & $14496$ & $45.62$ \\
AS (Oregon 1)  & $10670$ & $22002$ & $58.72$ \\
AS (Oregon 2)  & $10900$ & $31180$ & $70.74$ \\
%Irvine Social Net  & $1899$ & $13838$ & $48.14$ & \\
Brightkite Net  & $58228$ & $214078$  & $101.49$ \\
%Portland Cont. Net. & $1575861$ & $19481626$ & $27.23$ \iffalse $87.12$\fi & \cite{cinet:14}\\
Youtube Network & $1134890$ & $2987624$ & $210.4$ \\
Stanford Web graph & $281903$ & $1992636$ & $448.13$ \\ \bottomrule
\end{tabular}
\label{table:networks}
\vspace{-0.18in}
\end{table}

%%
%%%%%%%%%%%%%%%%%%%%%%%%%%%%%%%%%%%%%%%%%%%

\subsection{Experimental results}
\label{sec:exp-results}

$\newline$
\par
\noindent
\textbf{Performance of our algorithms and comparison with other heuristics:}
%\textsc{GreedyWalk} and \textsc{GreedyWalkSparse}}
%\label{sec:expt-greedywalk}
%%
We first compare the quality of solution from our algorithms with the
\textsc{EigenScore}, \textsc{ProductDegree}, \textsc{LinePagerank} and
\textsc{Hybrid} heuristics in Figure \ref{fig:simple_methods}. \iffalse Plots for more graphs, please see the full version \cite{spectral-approx-extended}.\fi \iffalse are shown in the appendix~\ref{subsec:other-plots}.\fi
We note that \textsc{GreedyWalk} is consistently better than all
other heuristics, \iffalse except in one case (the Brightkite network in Figure \ref{fig:6}), \fi
especially as the target threshold becomes smaller.
Compared to the \textsc{EigenScore}, \textsc{ProductDegree} and \textsc{LinePagerank} heuristics,
the spectral radius for the solution produced by \textsc{GreedyWalk},
as a function of the fraction of edges removed, is lower by at least 10-20\%.
\iffalse ;for the collaboration network (Figure \ref{fig:colgen}), the solution produced
by \textsc{GreedyWalk} is smaller by more than a factor of two. \fi Our improved
baseline, the \textsc{Hybrid} heuristic, works better than the other heuristics,
and comes somewhat close the \textsc{GreedyWalk} in many networks. \iffalse ; however,
except for the Brightkite network (Figure \ref{fig:6}), the
\textsc{Hybrid} performs slightly worse than \textsc{GreedyWalk}. Further, \fi \iffalse Although,
for the collaboration network, \textsc{Hybrid} remains significantly worse than \textsc{GreedyWalk}.\fi
%An interesting observation is that \textsc{ProductDegree} seems
%to work better on assortative graphs. Figures \ref{fig:pt} and \ref{fig:yt} show the results
%for the Portland and YouTube graphs; not all the methods are shown here, because of
%their time complexity.

Though \textsc{PrimalDual} gives a significantly better approximation guarantee, compared to
\textsc{GreedyWalk}, it has a much higher running time. Therefore, we only evaluate
it for one iteration of Algorithm \textsc{HitWalks}.
Figure \ref{fig:gw_vs_pd} shows that \textsc{PrimalDual} is quite close
to \textsc{GreedyWalk} after just one iteration; we expect running this algorithm fully would
further improve the performance, but additional work is needed to improve the running time.

\iftoggle{fullversion}
{
\begin{figure*}[ht]
\centering
\begin{tabular}{ccc}
\subfloat[AS Oregon-1]{\label{fig:5}
 \includegraphics[scale=0.23]{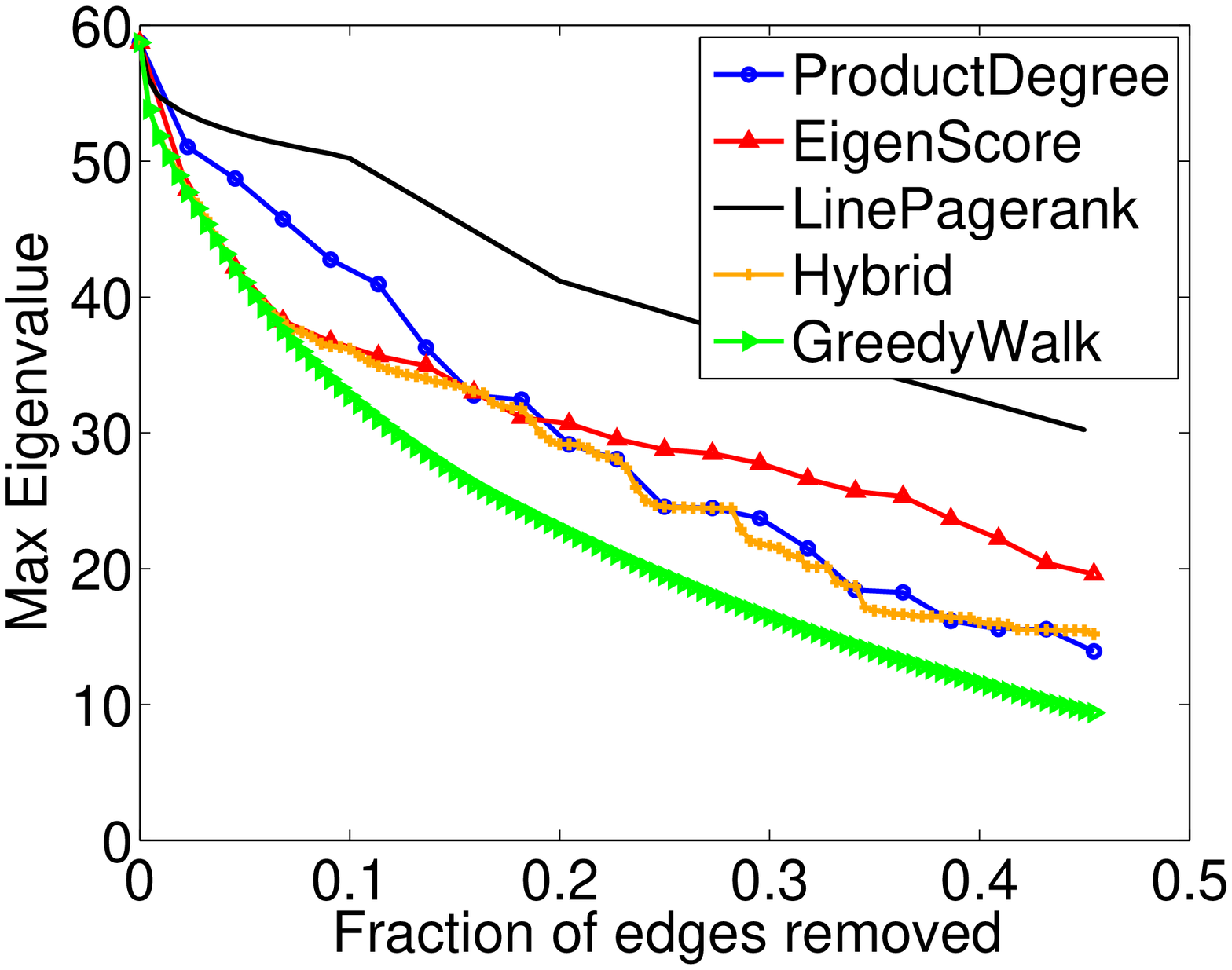}
} & \subfloat[AS Oregon-2]{\label{fig:as2}
 \includegraphics[width=0.23\textwidth]{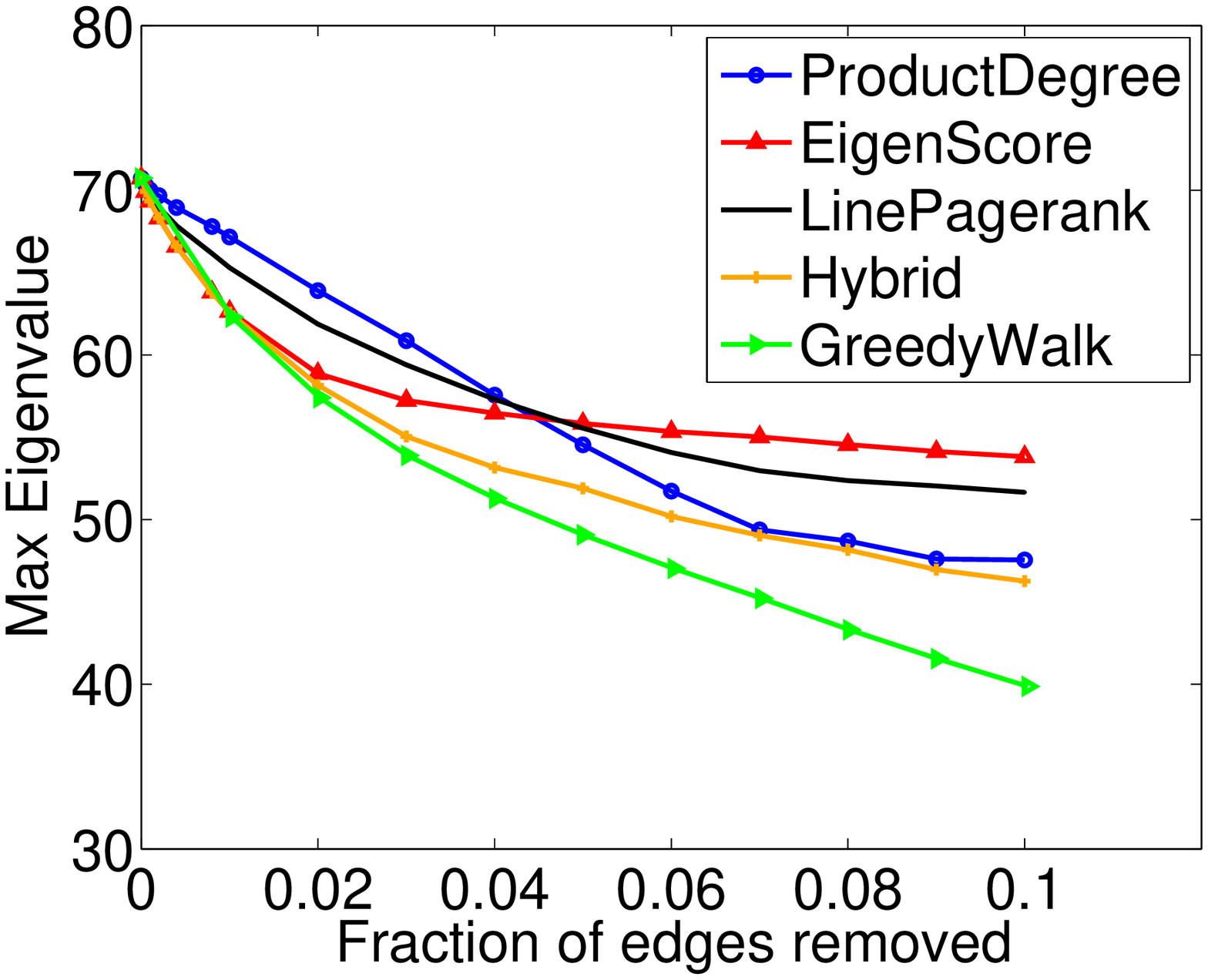}
} & \subfloat[Collaboration GrQc]{ \label{fig:colgen}
 \includegraphics[scale=0.23]{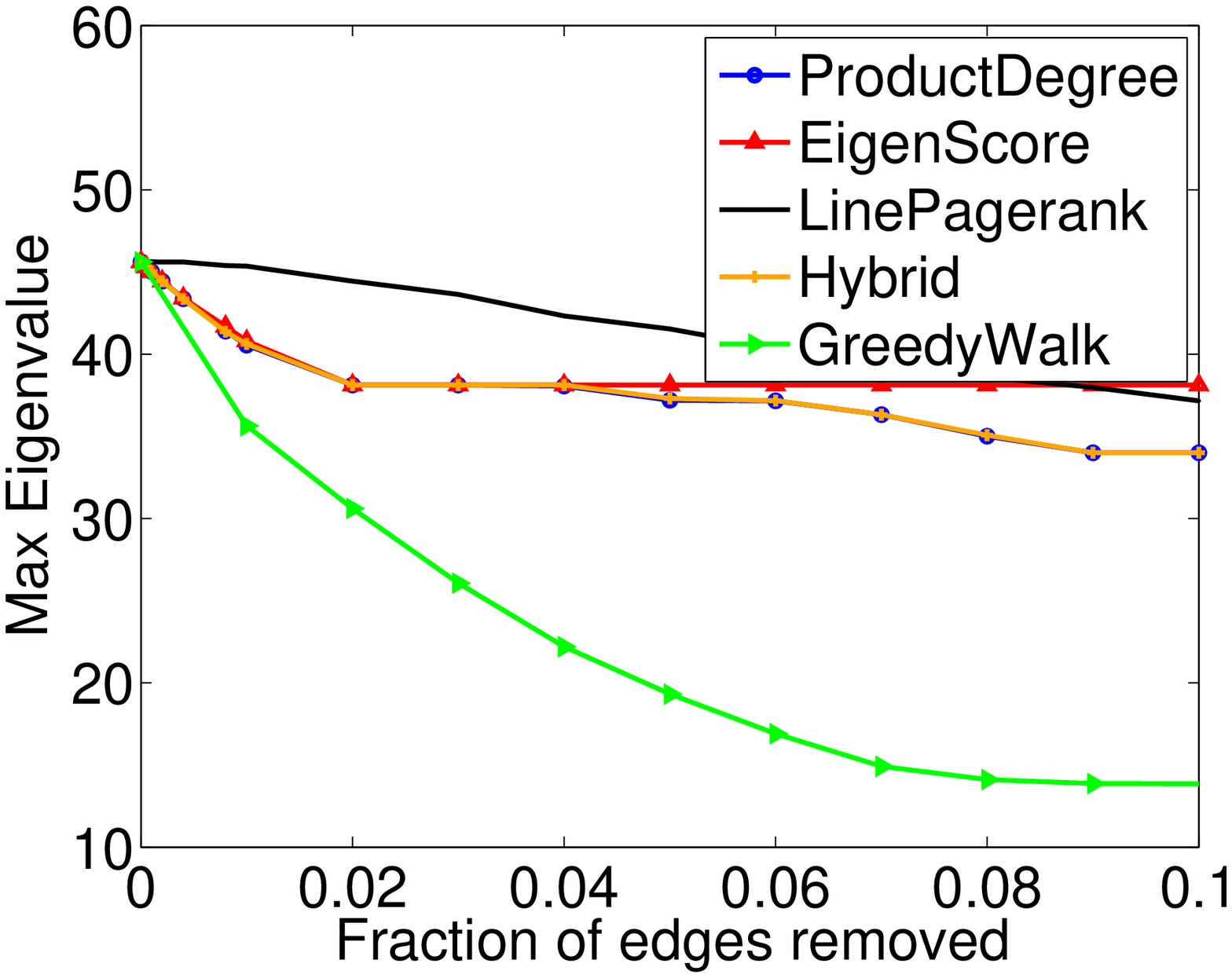}
}
\\ \subfloat[P2P Gnutella-5]{ \label{fig:9}
 \includegraphics[scale=0.23]{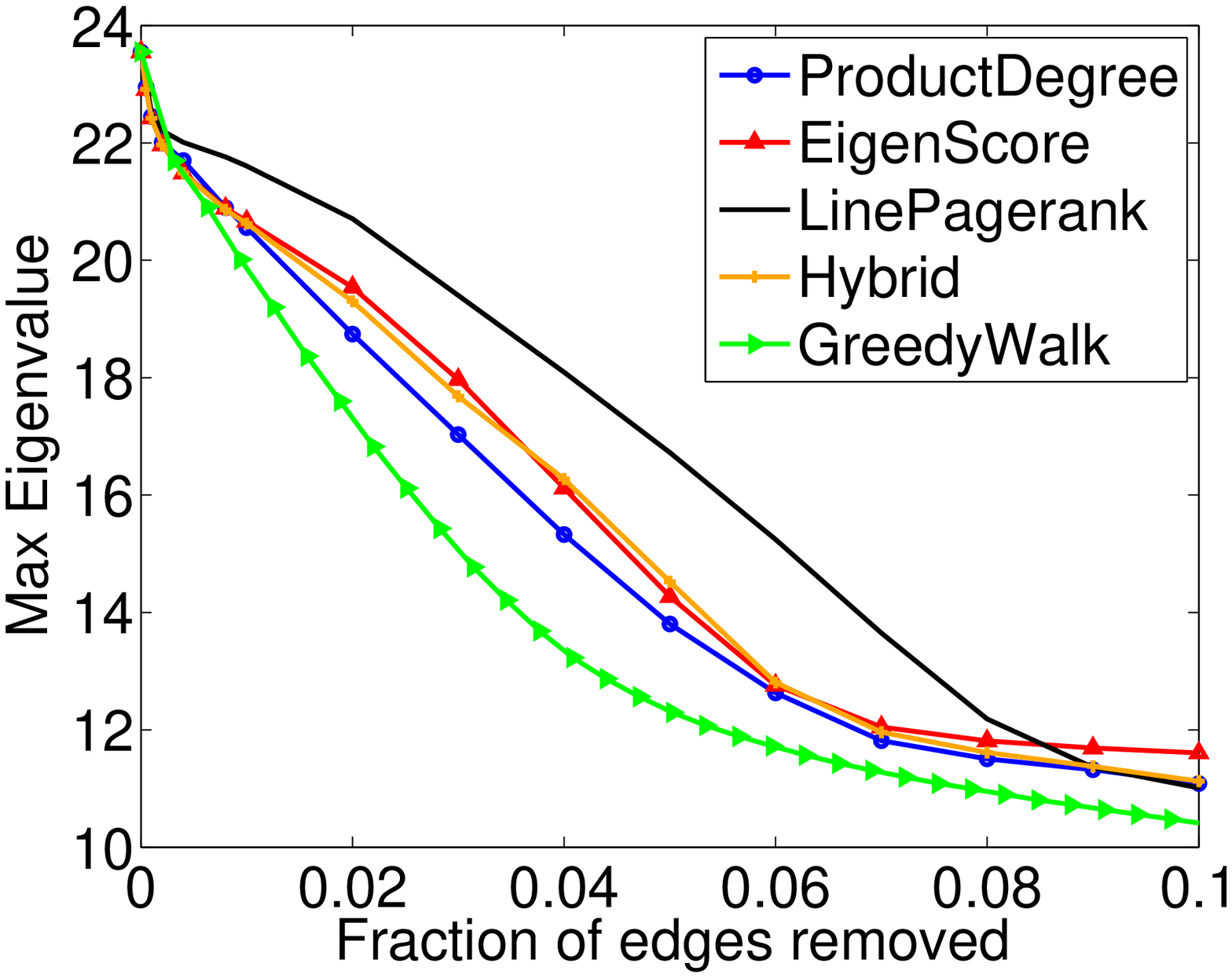}
} & \subfloat[P2P Gnutella-6]{ \label{fig:10}
 \includegraphics[width=0.23\textwidth]{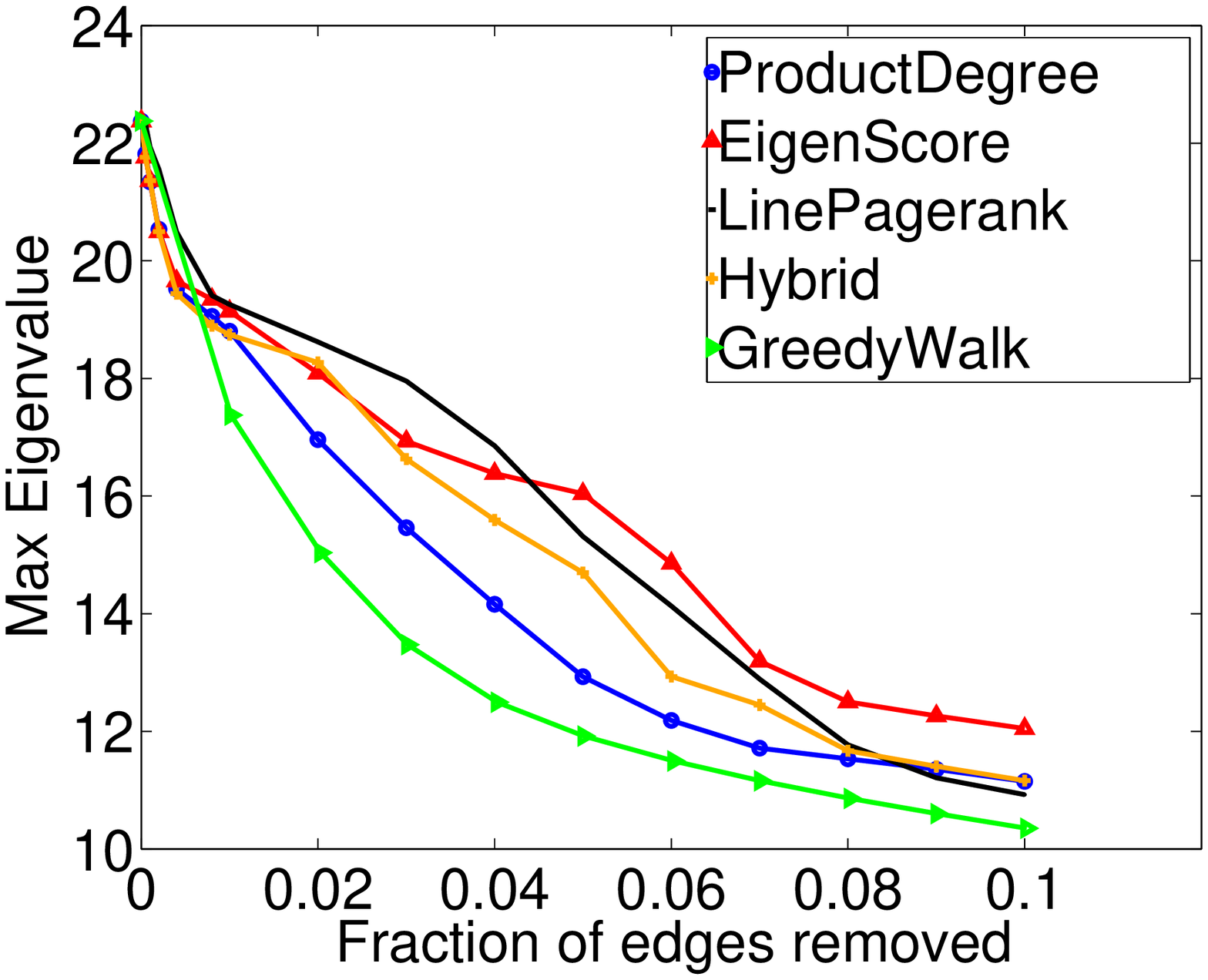}
} & \subfloat[Brightkite]{ \label{fig:6}
 \includegraphics[width=0.23\textwidth]{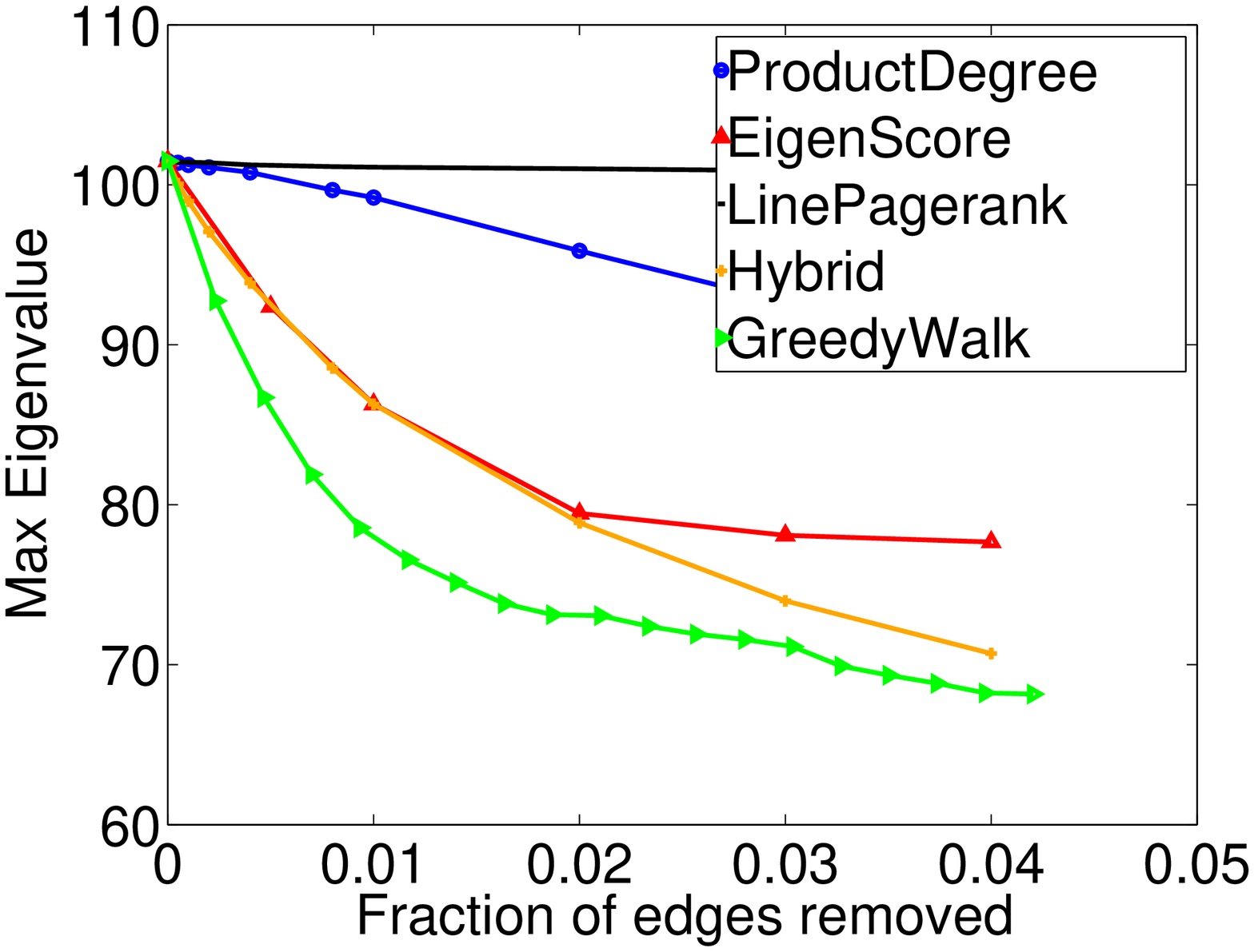}
}
\\ \subfloat[Portland]{ \label{fig:pt}
 \includegraphics[width=0.23\textwidth]{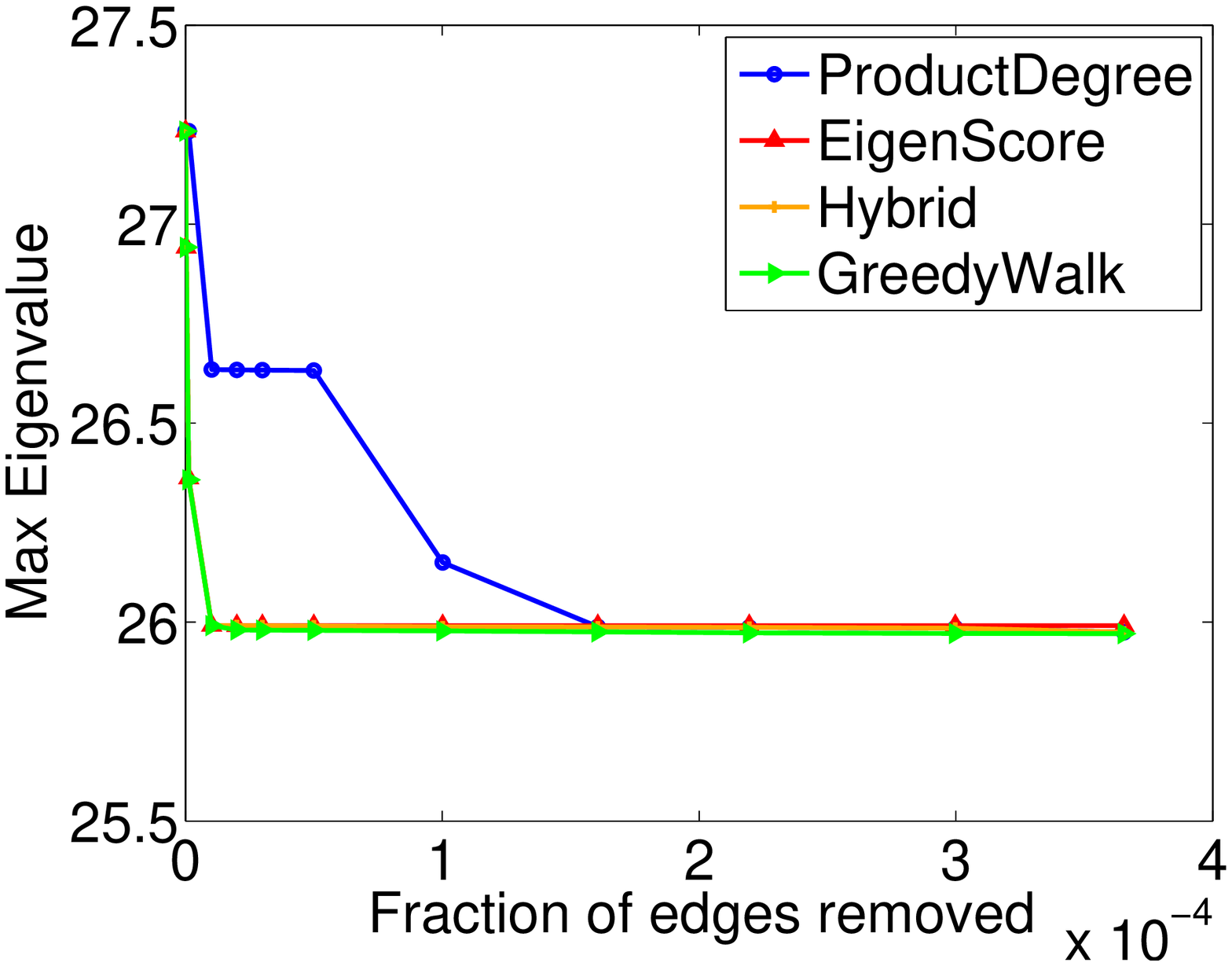}
} & \subfloat[Youtube]{ \label{fig:yt}
 \includegraphics[width=0.23\textwidth]{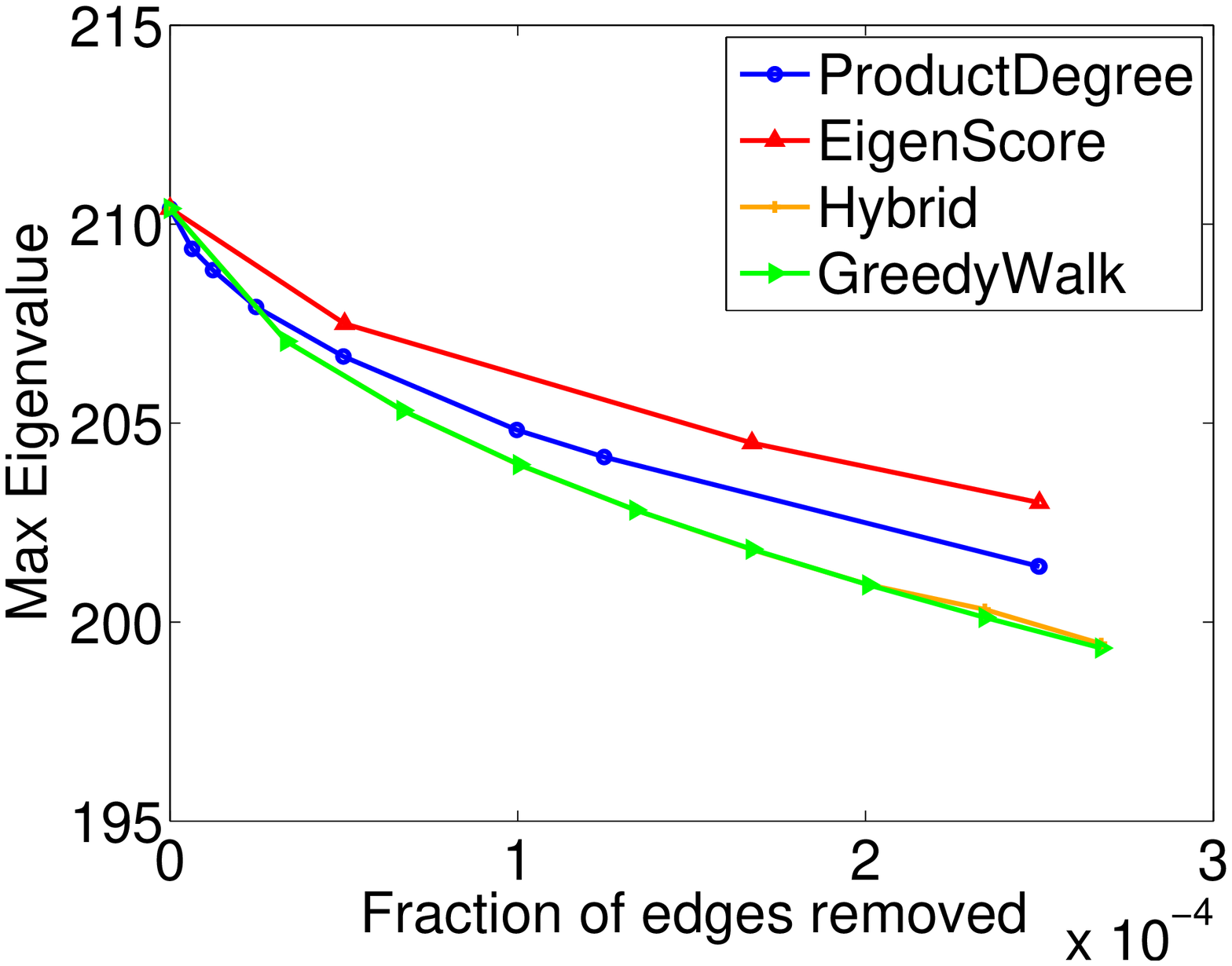}
} & \subfloat[Stanford Web]{ \label{fig:stn}
 \includegraphics[width=0.23\textwidth]{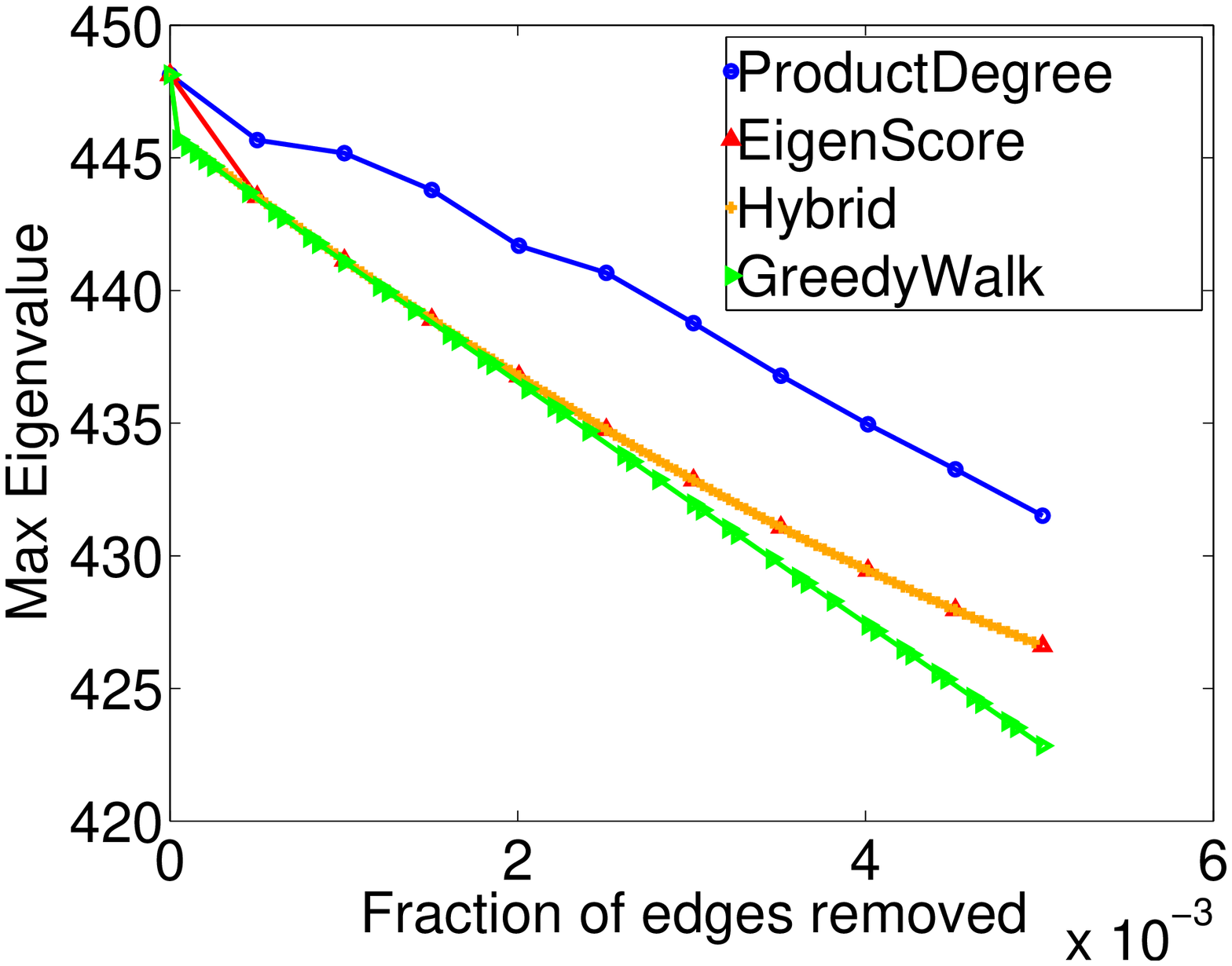}
}
\end{tabular}
\caption{ \label{fig:simple_methods}
Comparison between the \textsc{GreedyWalk}, \textsc{ProductDegree}, \textsc{Eigenscore}, \textsc{LinePagerank}
and \textsc{Hybrid} algorithms for different networks.
Each plot shows the spectral radius (y-axis)
as a function of the fraction of edges removed (x-axis).
%The assortativity of each network is also indicated by $r$.
The \textsc{LinePagerank} heuristic has not been evaluated in
\ref{fig:pt}, \ref{fig:yt} and \ref{fig:stn} because of the scale of these networks.}
\end{figure*}
}%iftoggle{fullversion}
{
\begin{figure*}[ht]
\centering
\begin{tabular}{ccc}
\subfloat[AS Oregon-2]{\label{fig:as2}
 \includegraphics[width=0.23\textwidth]{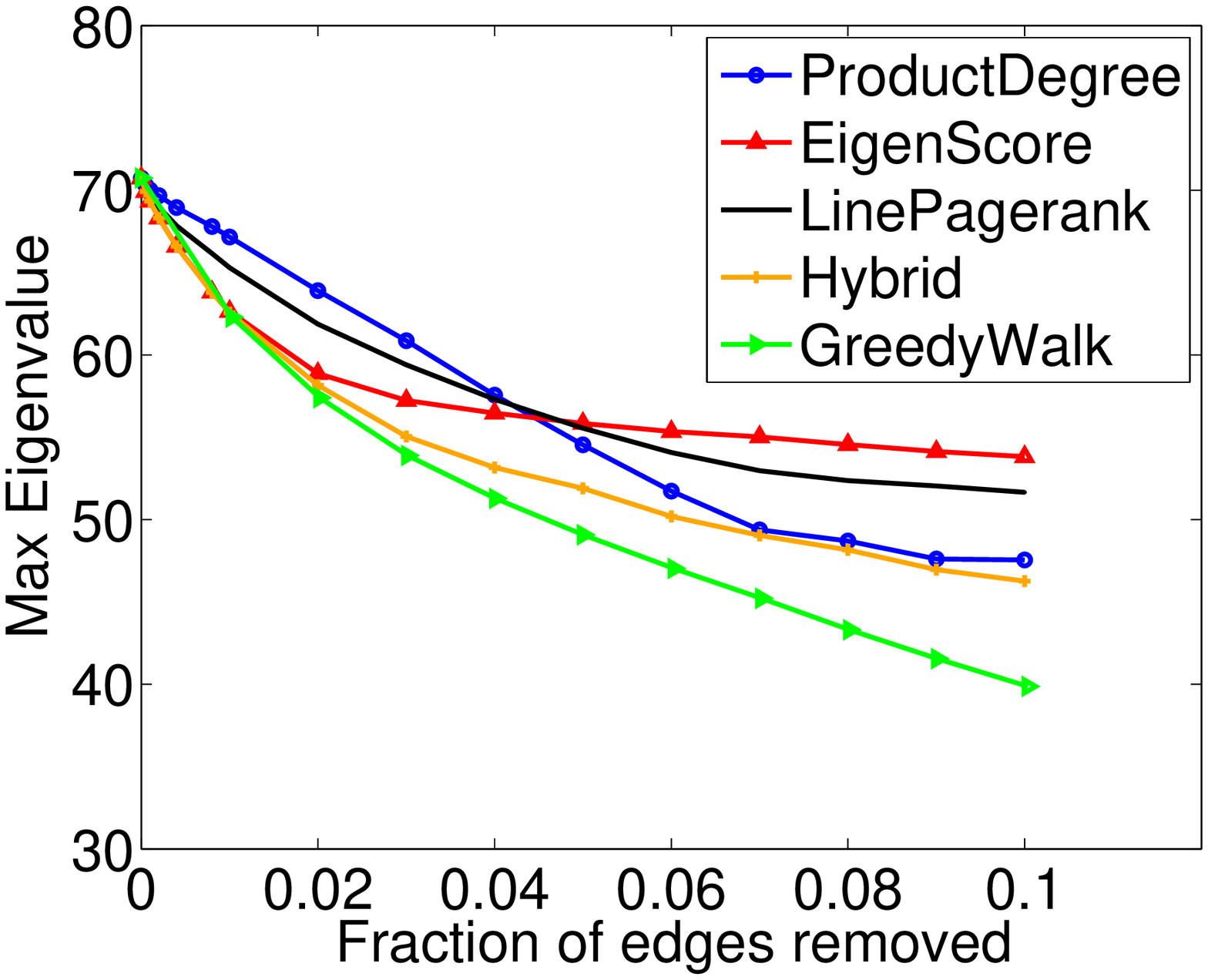}
} & \subfloat[P2P Gnutella-6]{ \label{fig:10}
 \includegraphics[width=0.23\textwidth]{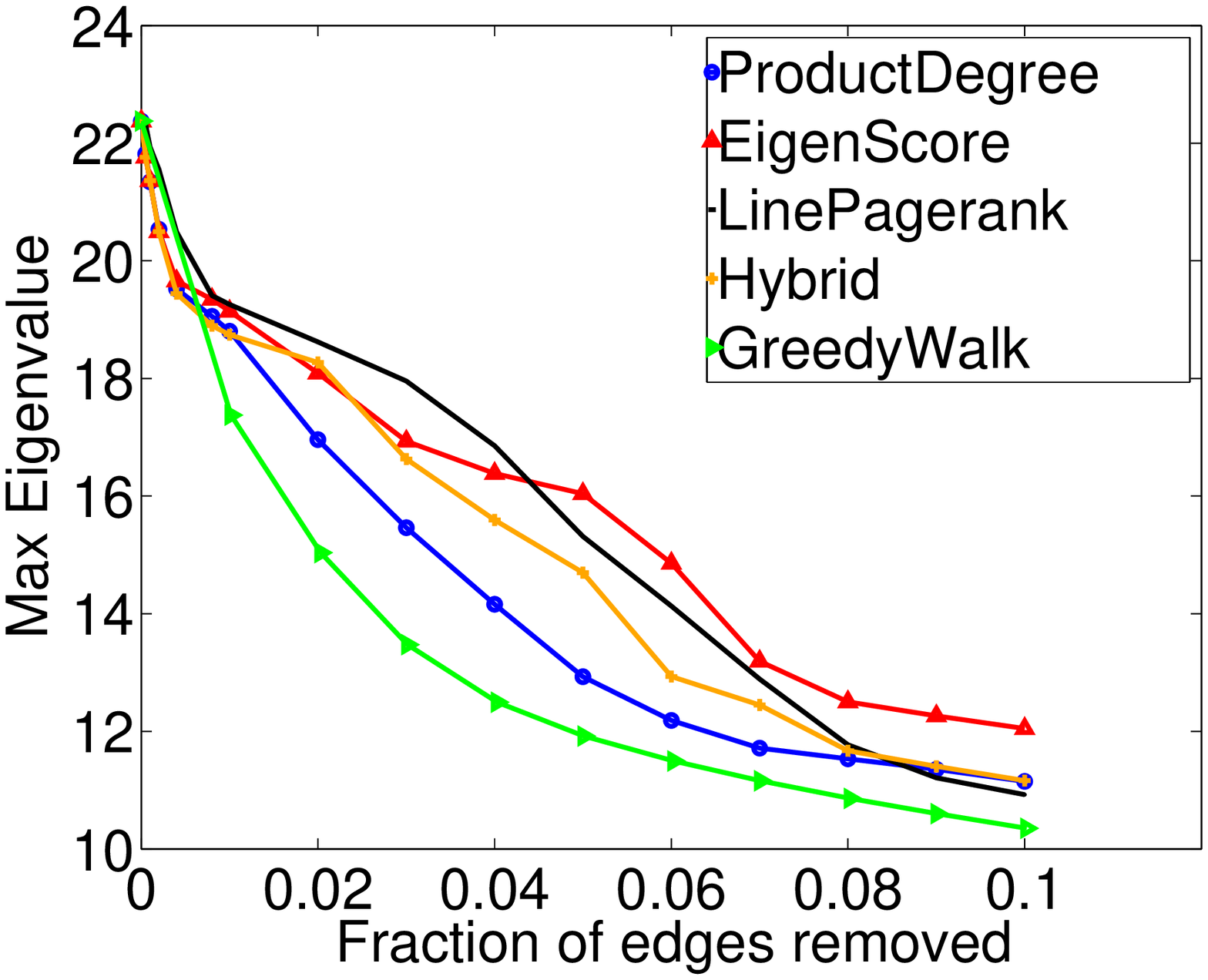}
} & \subfloat[Brightkite]{ \label{fig:6}
 \includegraphics[width=0.23\textwidth]{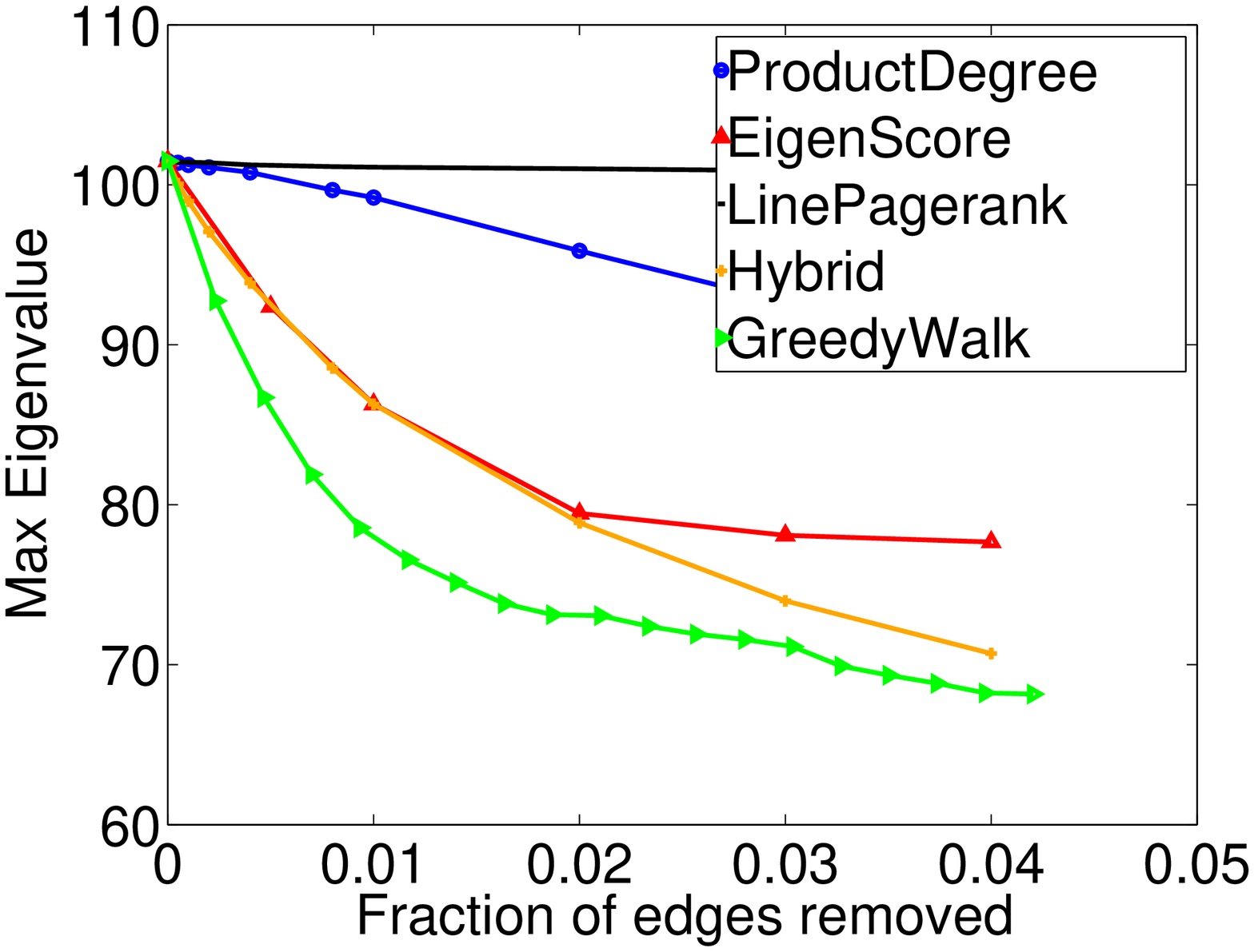}
}

\\ \subfloat[Collaboration GrQc]{ \label{fig:colgen}
 \includegraphics[width=0.23\textwidth]{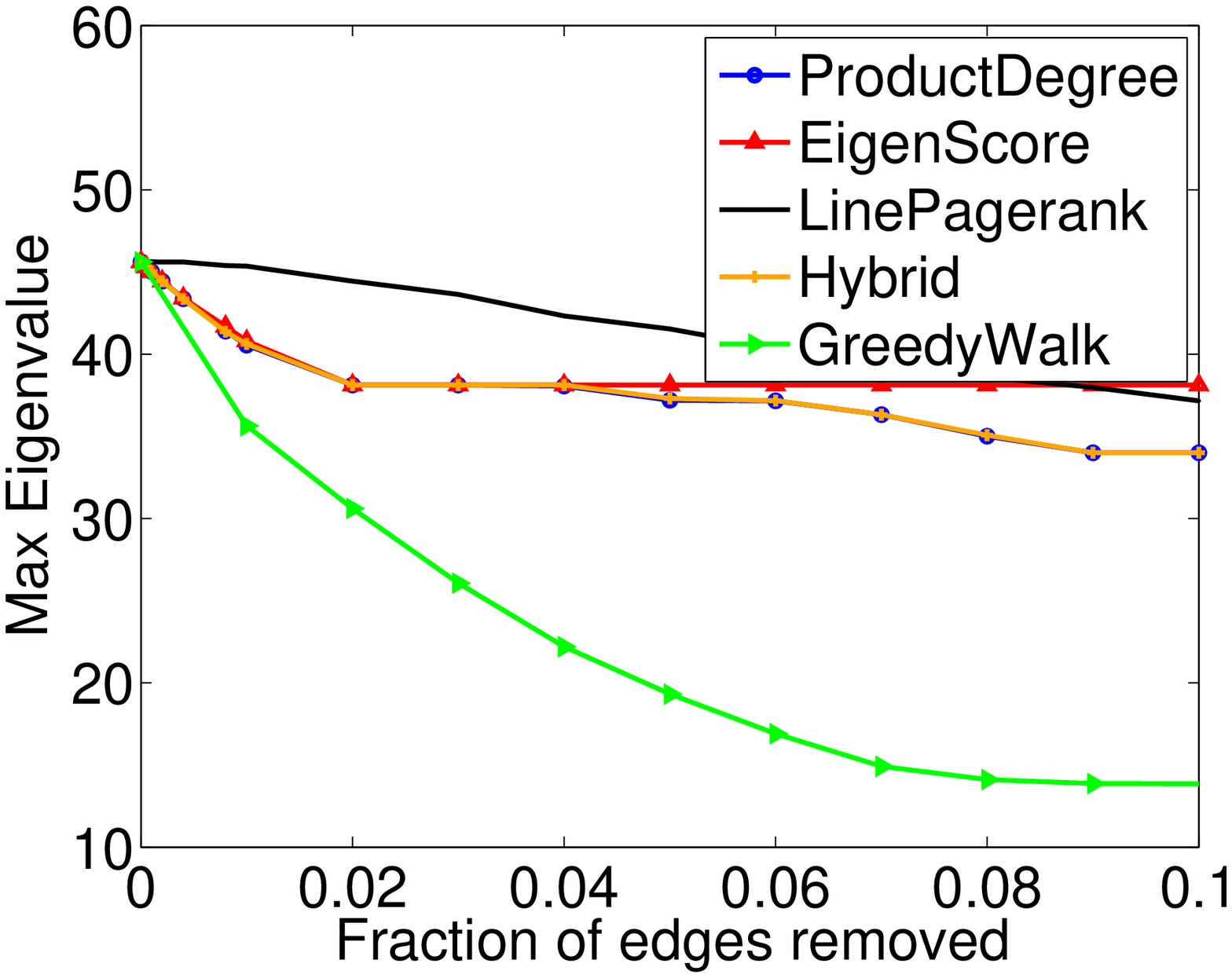}
} & \subfloat[Youtube]{ \label{fig:yt}
 \includegraphics[width=0.23\textwidth]{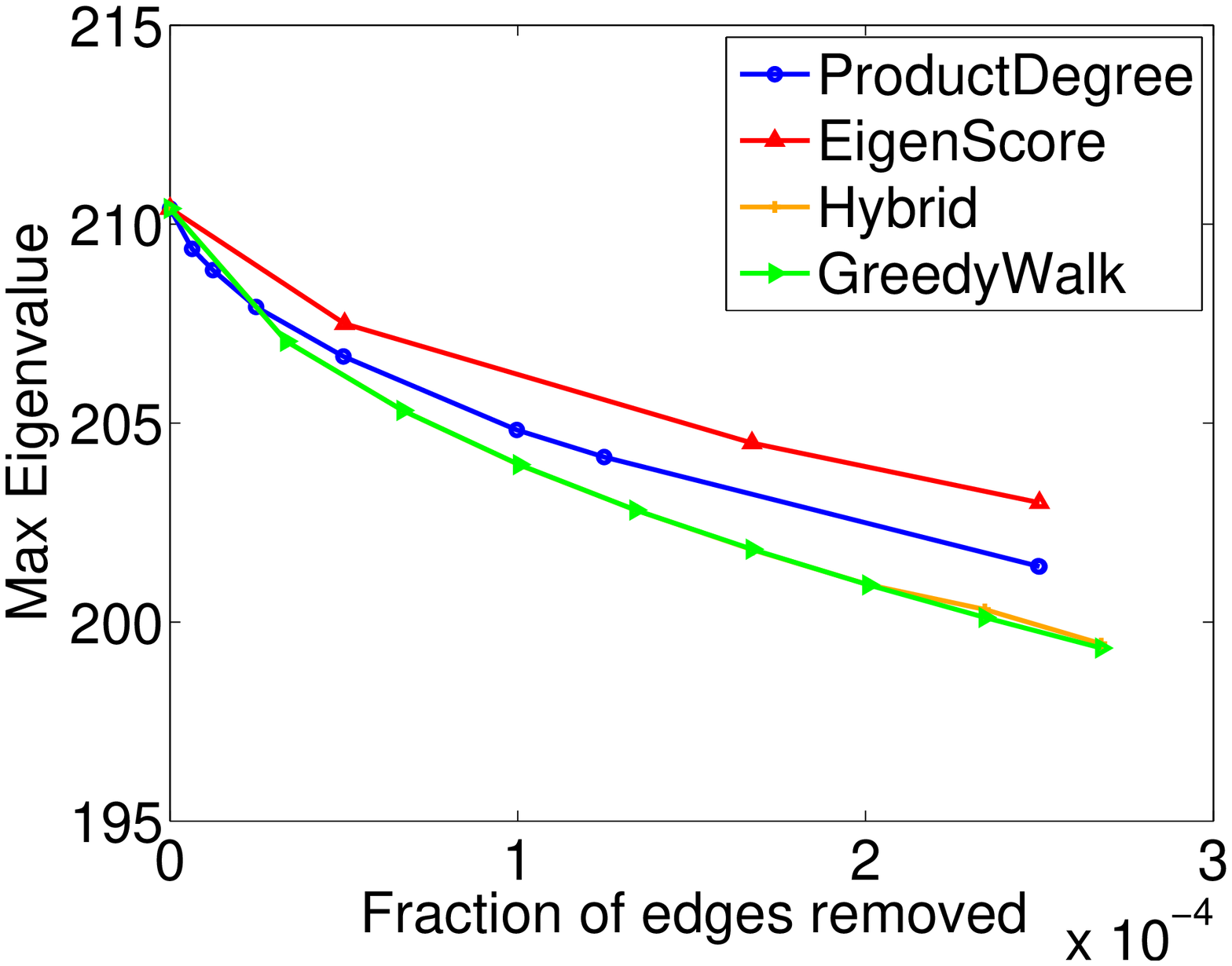}
} & \subfloat[Stanford Web]{ \label{fig:stn}
 \includegraphics[width=0.23\textwidth]{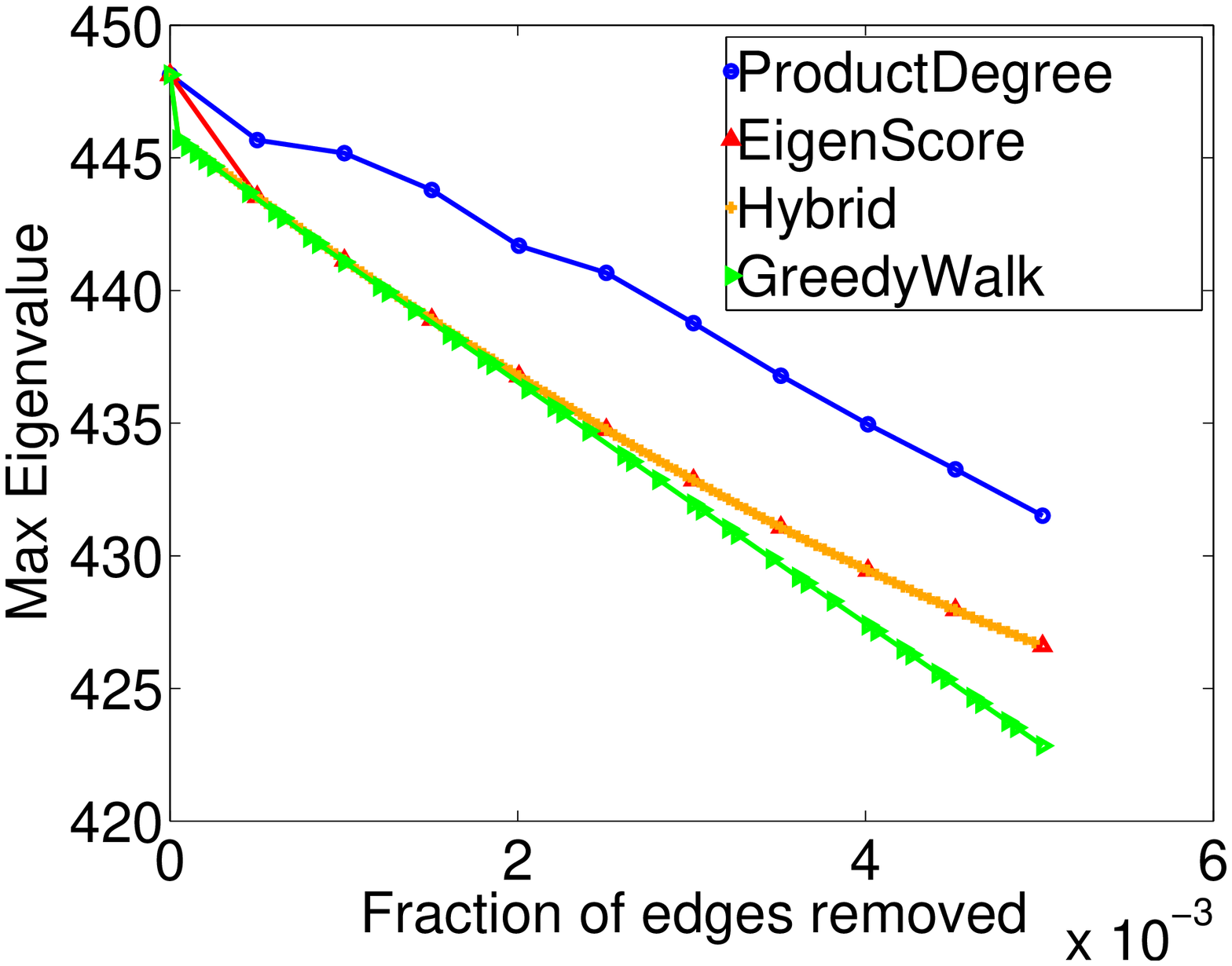}
}
\end{tabular}
\caption{ \label{fig:simple_methods}
Comparison between the \textsc{GreedyWalk}, \textsc{ProductDegree}, \textsc{Eigenscore}, \textsc{LinePagerank}
and \textsc{Hybrid} algorithms for different networks.
Each plot shows the spectral radius (y-axis)
as a function of the fraction of edges removed (x-axis).
%The assortativity of each network is also indicated by $r$.
The \textsc{LinePagerank} heuristic has not been evaluated in
\iffalse \ref{fig:pt},\fi \ref{fig:yt} and \ref{fig:stn} because of the scale of these networks.}
\end{figure*}

}%!iftoggle{fullversion}

\begin{figure}[ht]
\centering
\begin{tabular}{cc}
\subfloat[Collaboration GrQc]{ \label{fig:colgen_gw_vs_pd}
 \includegraphics[width=0.42\columnwidth]{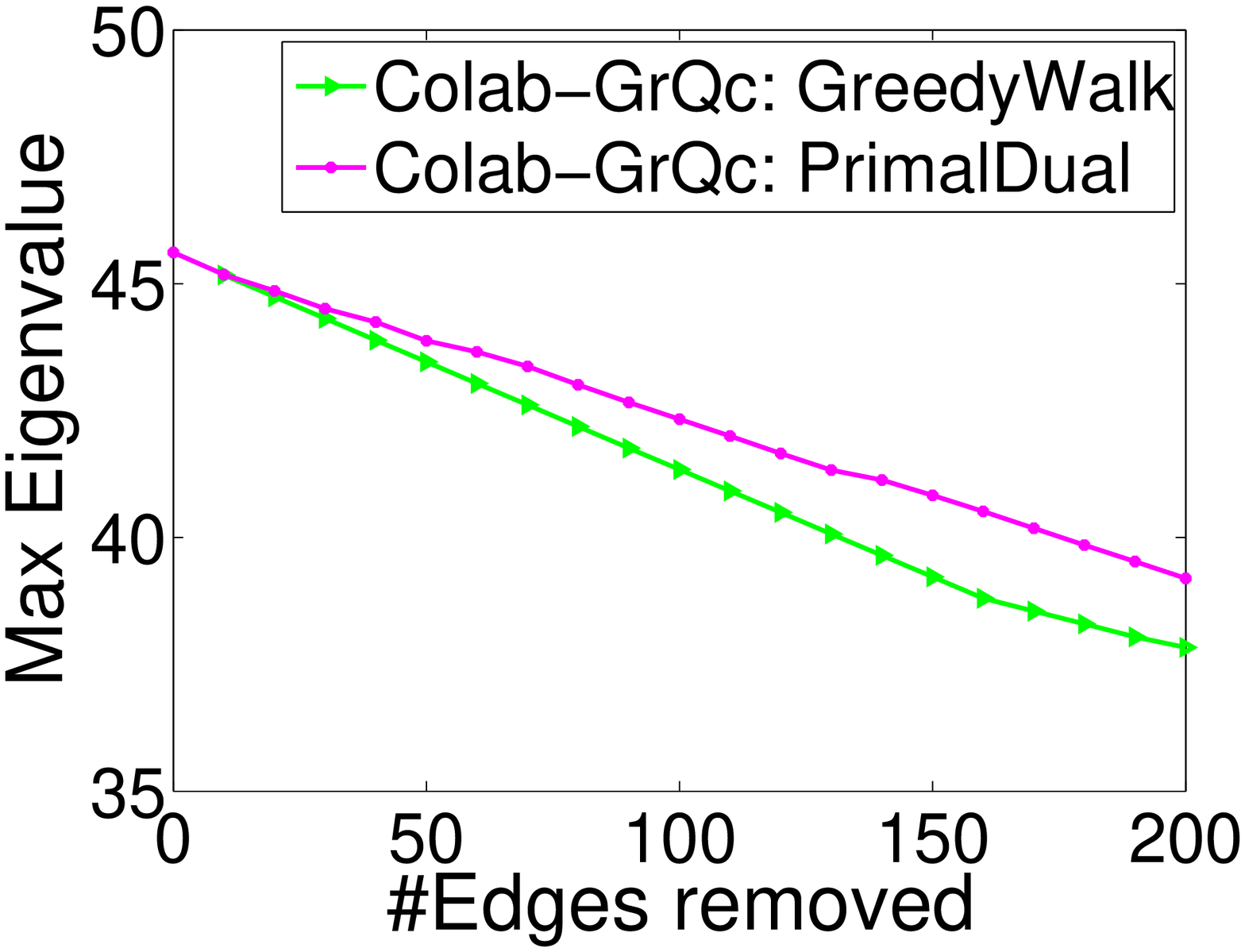}
} &  \subfloat[P2P Gnutella-5]{ \label{fig:5_gw_vs_pd}
 \includegraphics[width=0.42\columnwidth]{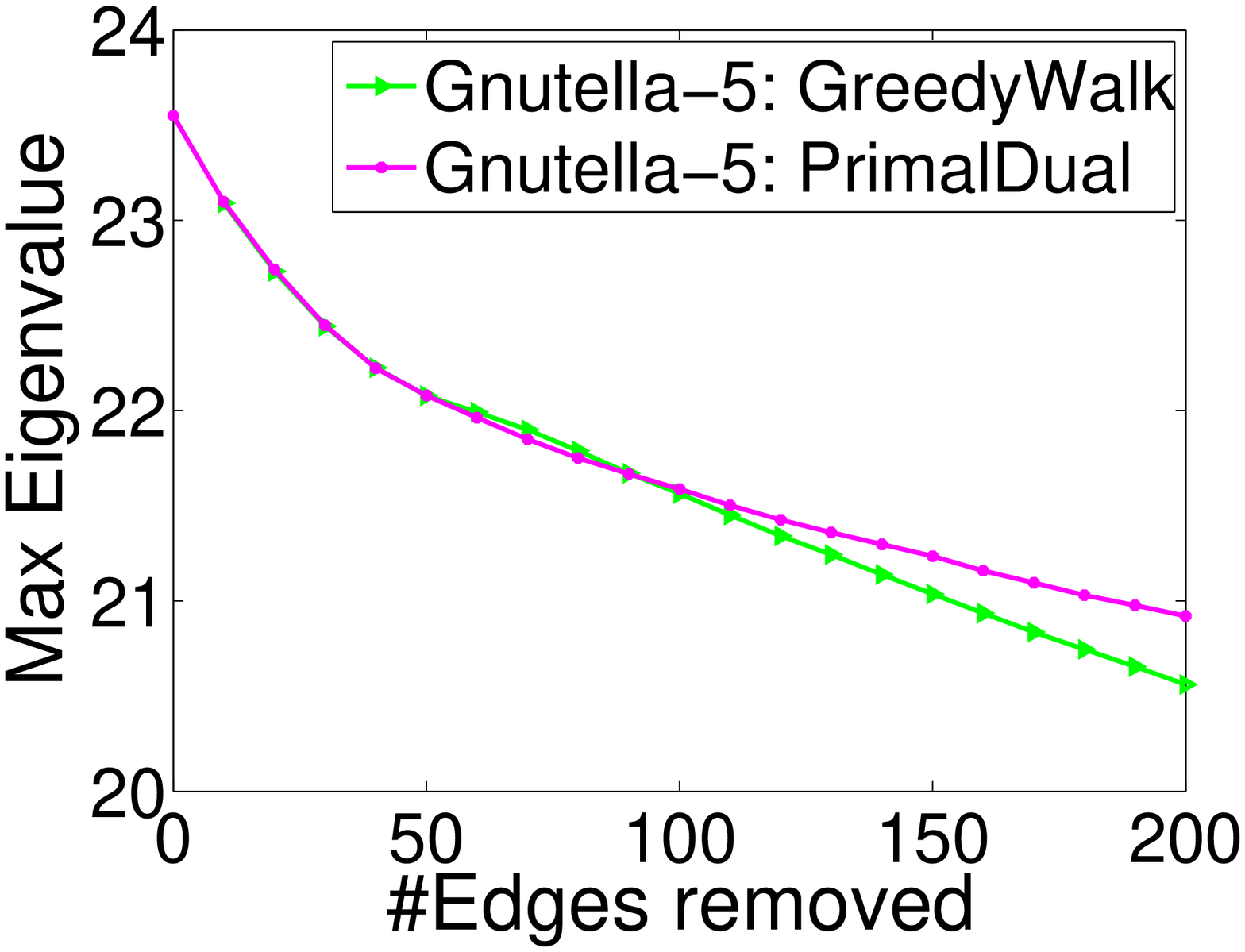}
}

%&  \subfloat[AS Oregon-2]{ \label{fig:as2_gw_vs_pd}
% \includegraphics[scale=0.15]{fig/AS-Oreg2_gw_vs_pd.eps}
%}
\\
\end{tabular}
\caption{ \label{fig:gw_vs_pd} \textsc{GreedyWalk} vs \textsc{PrimalDual}.
Each plot shows the spectral radius (y-axis)
as a function of the number of edges removed (x-axis) using the two methods.
}
\end{figure}

\begin{figure}[ht]
\centering
 \includegraphics[width=0.4\columnwidth]{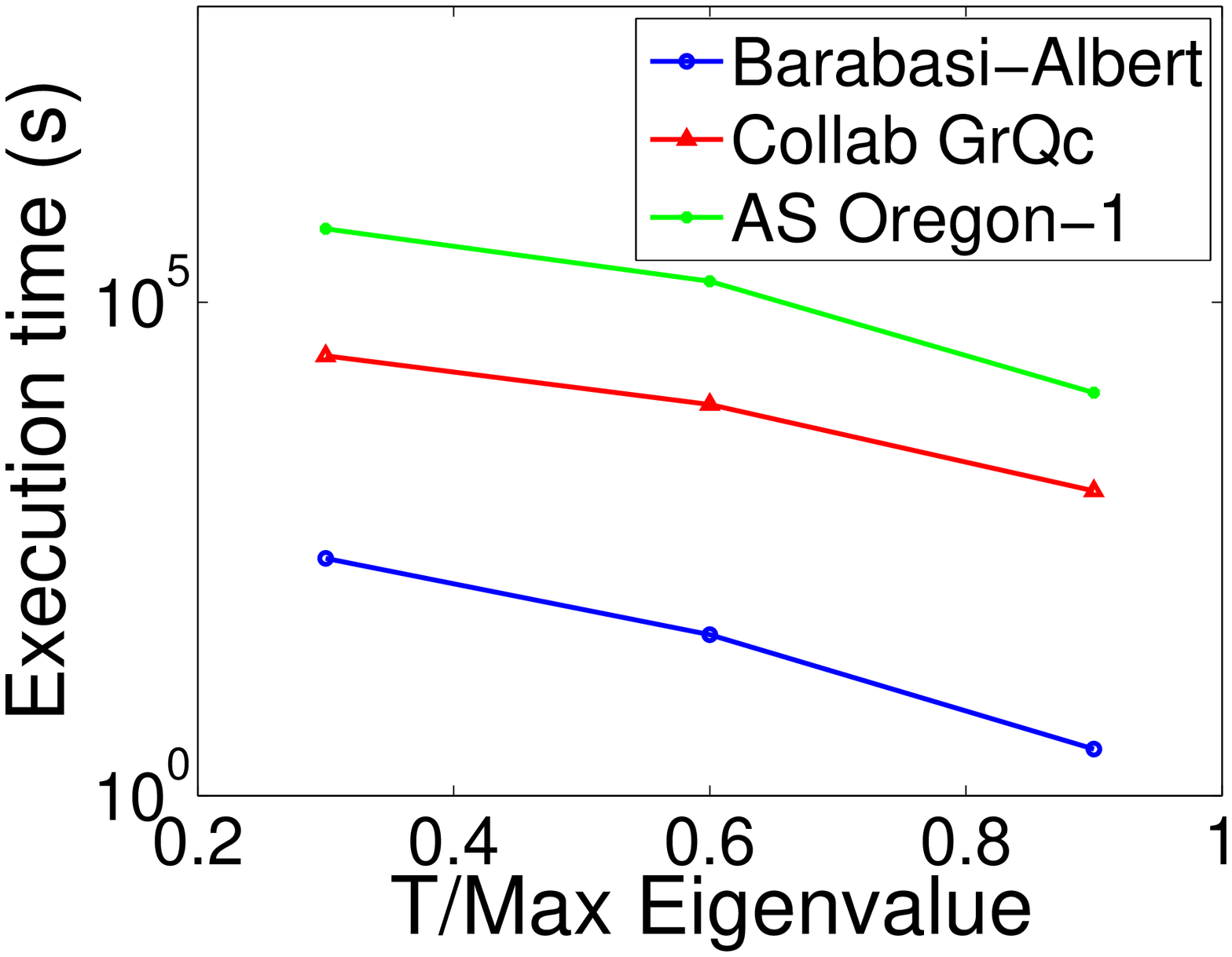}
\caption{ \label{fig:greedywalk_runtime} Total running time of GreedyWalk method (y-axis)
as a function of $T/\rho(G)$ (x-axis), where $T$ is the threshold and $\rho(G)$ is the
spectral radius of the initial graph, without any edges removed.
}
\end{figure}

\par
\noindent
\textbf{Running time and effect of sparsification:}.
Figure~\ref{fig:greedywalk_runtime} shows the total running time of \textsc{GreedyWalk}
for three networks. The time decreases with the increase of
$T$, because the while loop in Algorithm \textsc{GreedyWalk} needs to be run
for fewer iterations.
The high running time motivates faster methods. We evaluate
the performance of the \textsc{GreedyWalkSparse} algorithm. As shown in
Figure \ref{fig:impact_of_sparsification},
\textsc{GreedyWalkSparse} gives almost the same quality of approximation as
\textsc{GreedyWalk}, but improves the running time by up to an order of magnitude,
particularly when $T$ is small.
%For example, in AS Oregon-1 graph for $T=0.05\times \lambda_1$, sparsification reduces the running time by a factor of $30$ (figure \ref{fig:impact_of_sparsification}).

\begin{figure}[ht]
\centering
\begin{tabular}{cc}
\subfloat[\#Edges removed]{ \label{fig:5_edges_sparsification}
 \includegraphics[width=0.42\columnwidth]{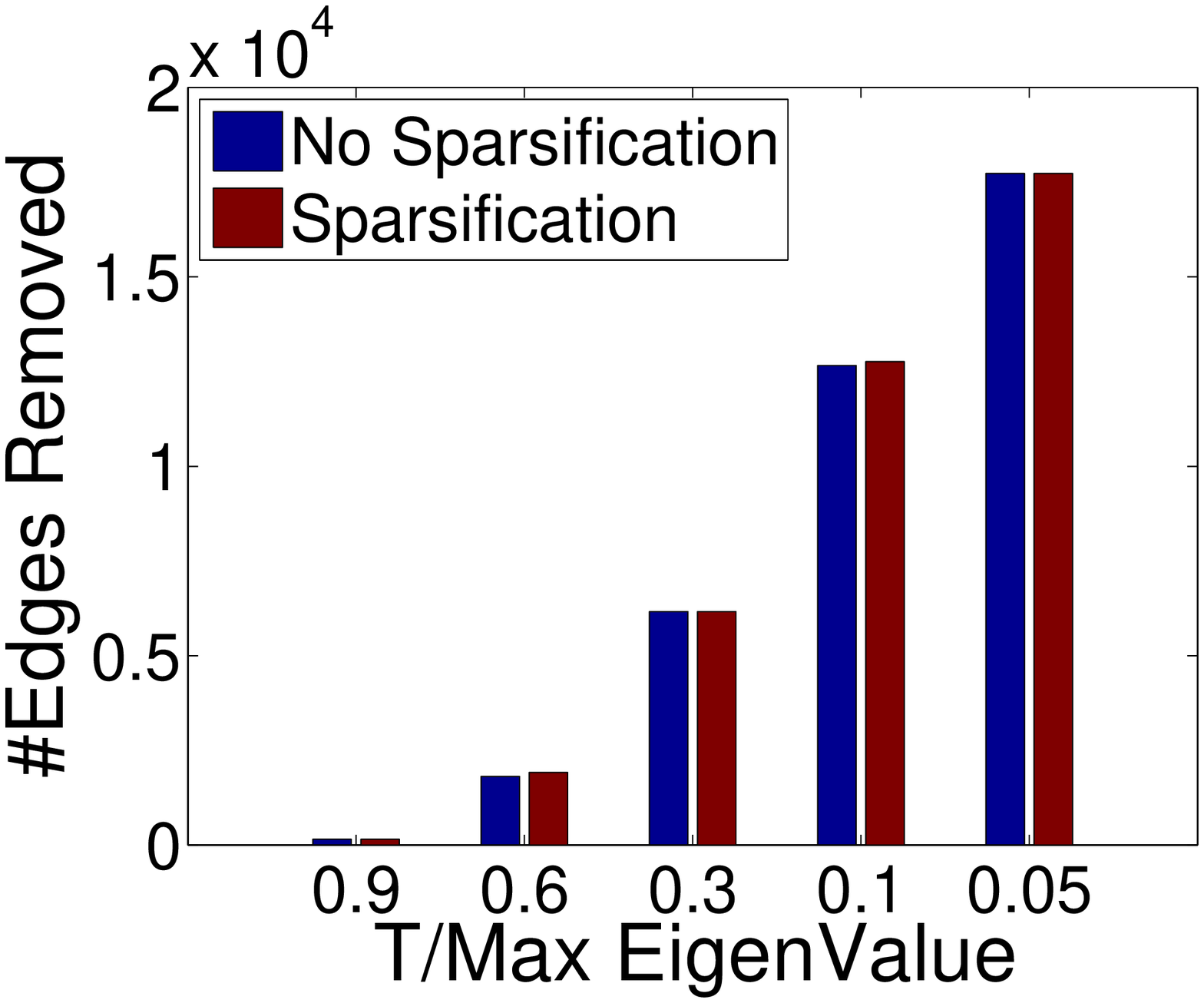}
} &  \subfloat[Execution time]{ \label{fig:5_times_sparsification}
 \includegraphics[width=0.42\columnwidth]{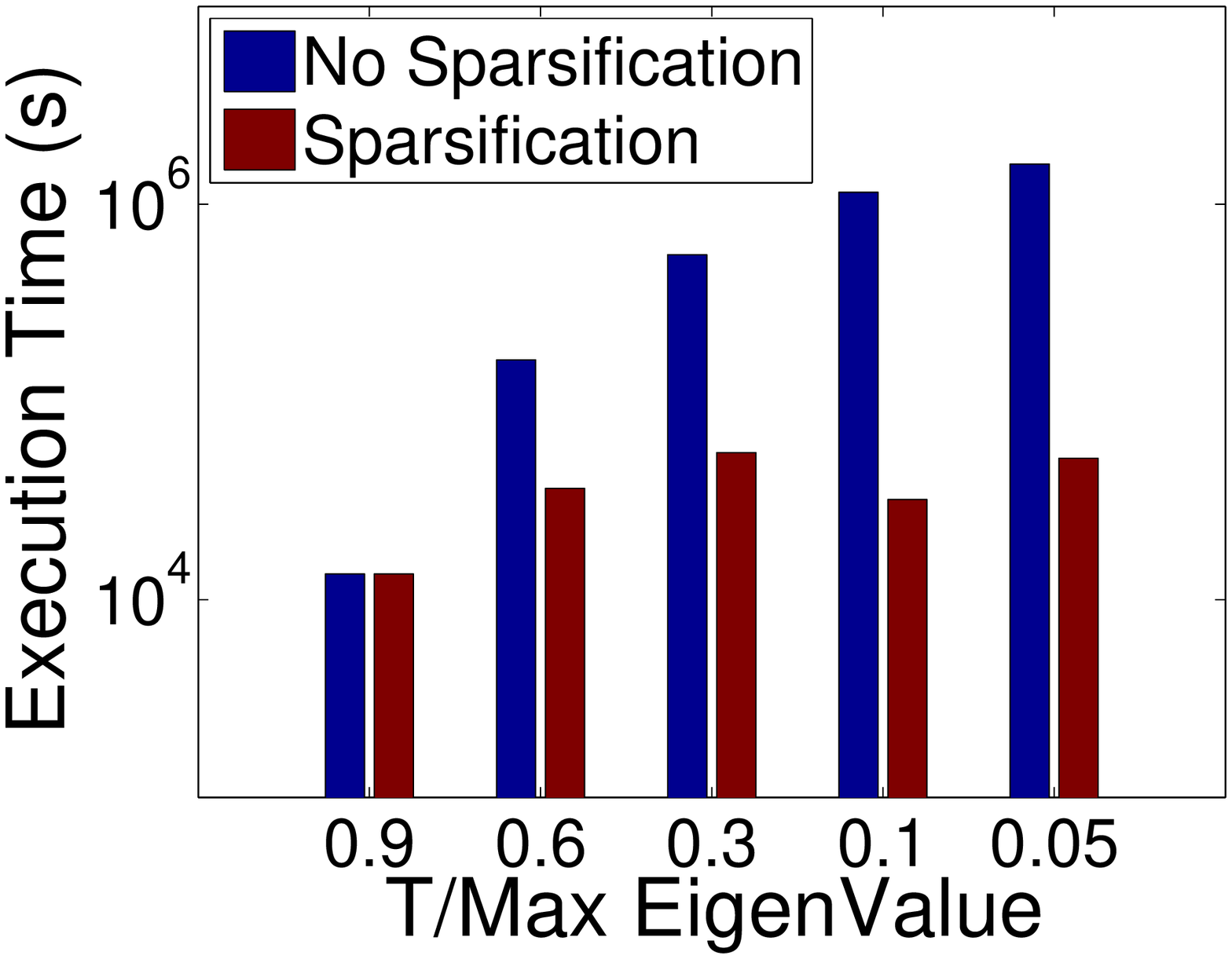}
}\\
\end{tabular}
\caption{ \label{fig:impact_of_sparsification} Impact of sparsification on \textsc{GreedyWalk}.
The plots show for AS Oregon-1 network, (a) the number of edges removed and (b) the execution time on the y-axis,
as a function of $T/\rho(G)$ (x-axis), where $T$ is the threshold and $\rho(G)$ is the
spectral radius of the initial graph, without any edges removed.
%The plot shows results for AS Oregon-1 graph. Sparsification performs almost the same in terms of the number of edges removed and much better in terms of execution time.
}
\end{figure}

\iffalse
\begin{figure}[ht]
\centering
\begin{tabular}{cc}
\subfloat[Collaboration GrQc]{ \label{fig:colgen_gw_vs_pd}
 \includegraphics[width=0.42\columnwidth]{Colab-GrQc_gw_vs_pd.eps}
} &  \subfloat[P2P Gnutella-5]{ \label{fig:5_gw_vs_pd}
 \includegraphics[width=0.42\columnwidth]{Gnutella-5_gw_vs_pd.eps}
}
\end{tabular}
\end{figure}
\fi

\iftoggle{fullversion}
{
\par
\noindent
\textbf{Effect of varying walk lengths:}
%\label{sec:k}
As discussed in Section \ref{sec:greedywalk}, the walk length parameter $k$ is critical
for the performance of \textsc{GreedyWalk}.  Figure \ref{fig:impact_k}
shows the approximation quality in the Oregon-2 and collaboration networks. We find that
as $k$ becomes smaller, the approximation quality degrades significantly, and the
best performance occurs at $k$ close to $2\log{n}$.

\begin{figure}[ht]
\centering
\begin{tabular}{cc}
\subfloat[AS Oregon-2]{\label{fig:as2_impact_k}
 \includegraphics[width=0.42\columnwidth]{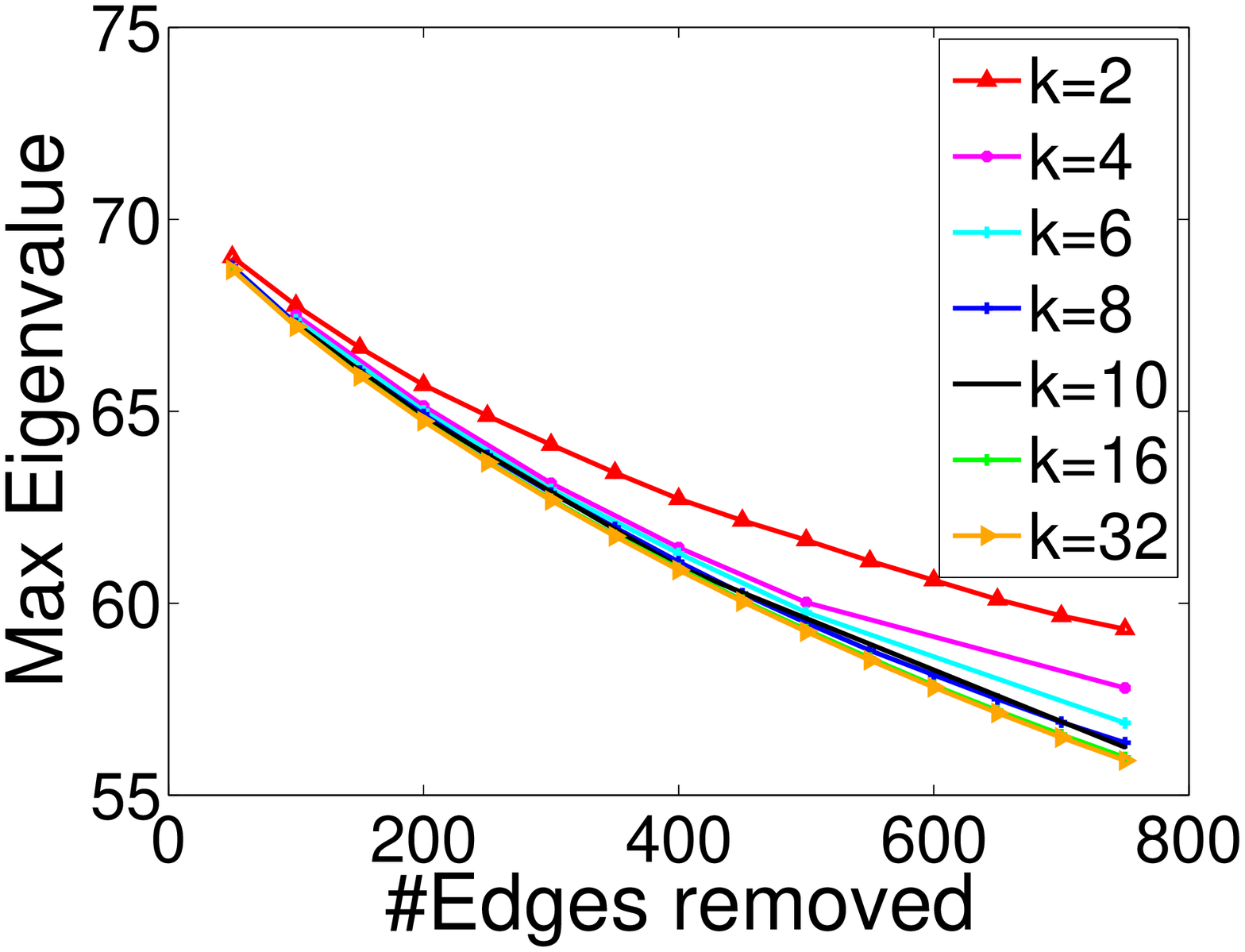}
} &  \subfloat[Collaboration GrQc]{ \label{fig:colgen_impact_k}
 \includegraphics[width=0.42\columnwidth]{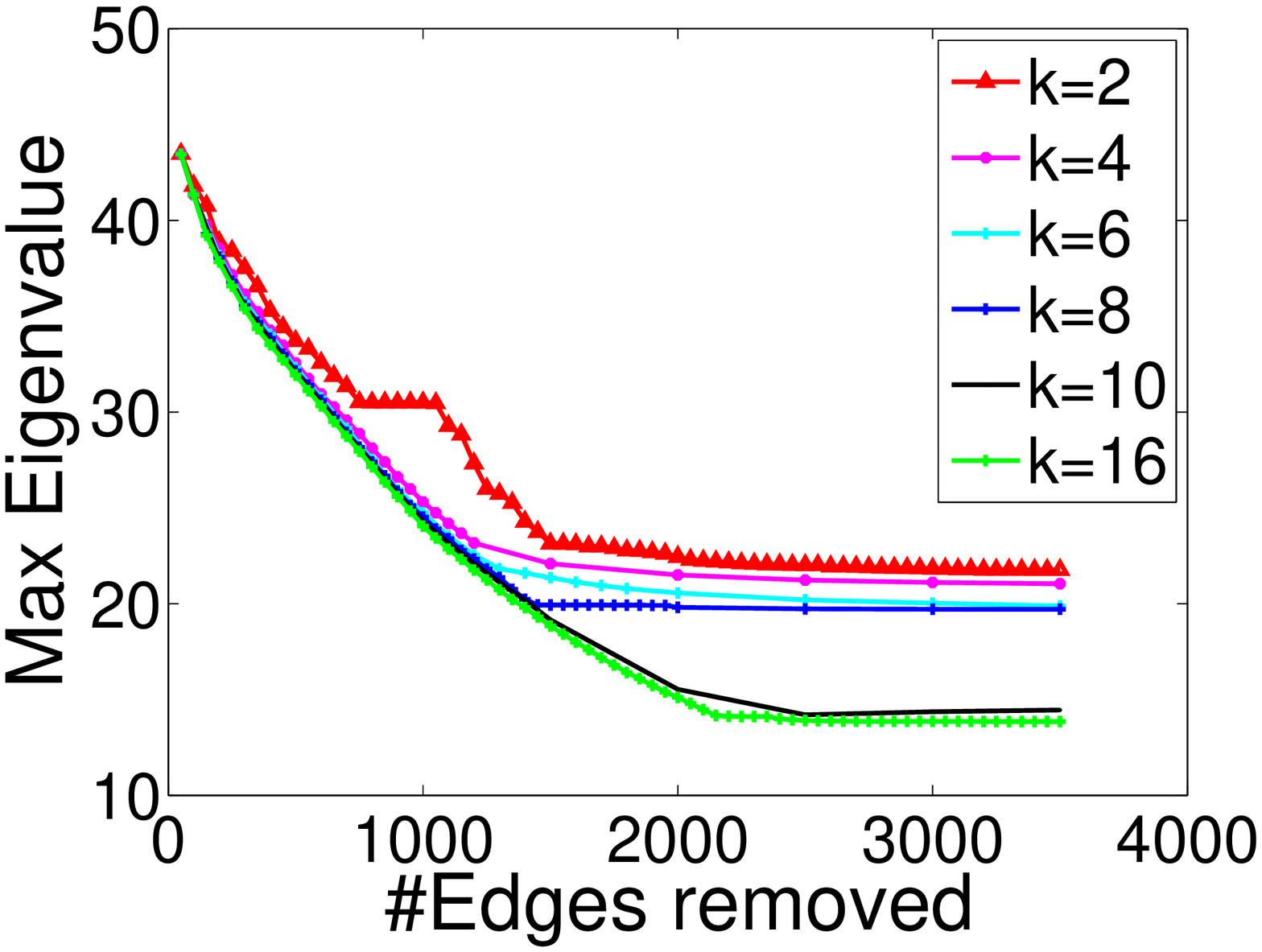}
}
% & \subfloat[Gnutella-06]{ \label{fig:10_impact_k}
% \includegraphics[scale=0.27]{fig/10_impact_k.eps}
%}
\\
\end{tabular}
\caption{ \label{fig:impact_k} Impact of walk length on \textsc{GreedyWalk} performance.
Each plot shows the drop in spectral radius (y-axis) with number of edges removed
(x-axis), for different values of $k$, ranging from $2$ to $2\log{n}$, for the corresponding
networks.
}
\end{figure}
}%iftoggle{fullversion}
{
}

\iftoggle{fullversion}
{
\par
\noindent
\textbf{Extensions:}
For the \textsc{SRME-nonuniform} problem, we compare the adaptation of \textsc{GreedyWalk},
as discussed in Section \ref{sec:extensions}, with the \textsc{Eigenscore} heuristic
run on the matrix $B$ of transmission rates.
As shown in Figure \ref{fig:colgen_nonuni}, we
find that \textsc{GreedyWalk} performs much better. Next we consider the \textsc{SRMN} problem,
and compare the \textsc{GreedyWalk}, as adapted in Section \ref{sec:extensions}, with
the node versions of the \textsc{Degree} and \textsc{EigenScore} heuristic \cite{tong:cikm12}.
As shown in Figure \ref{fig:colgen_node}, \textsc{GreedyWalk} performs consistently better. For results in other networks, see the full version \cite{spectral-approx-extended}.\iffalse see Appendix~\ref{subsec:other-plots}.\fi
}
{
\par
\noindent
%\textbf{Extension:}
\textbf{SRMN:}
For the \textsc{SRMN} problem,
we compare the adaption of \textsc{GreedyWalk}, as discussed in Section \ref{sec:extensions}, with
the node versions of the \textsc{Degree} and \textsc{EigenScore} heuristic \cite{tong:cikm12}.
As shown in Figure \ref{fig:srmn_experiment}, \textsc{GreedyWalk} performs consistently better.
}

\iftoggle{fullversion}
{
\begin{figure}[ht]
\centering
\begin{tabular}{cc}
 \subfloat[ \textsc{SRME-nonuniform}: Barabasi-Albert]{ \label{fig:br_nonuni}
 \includegraphics[width=0.42\columnwidth]{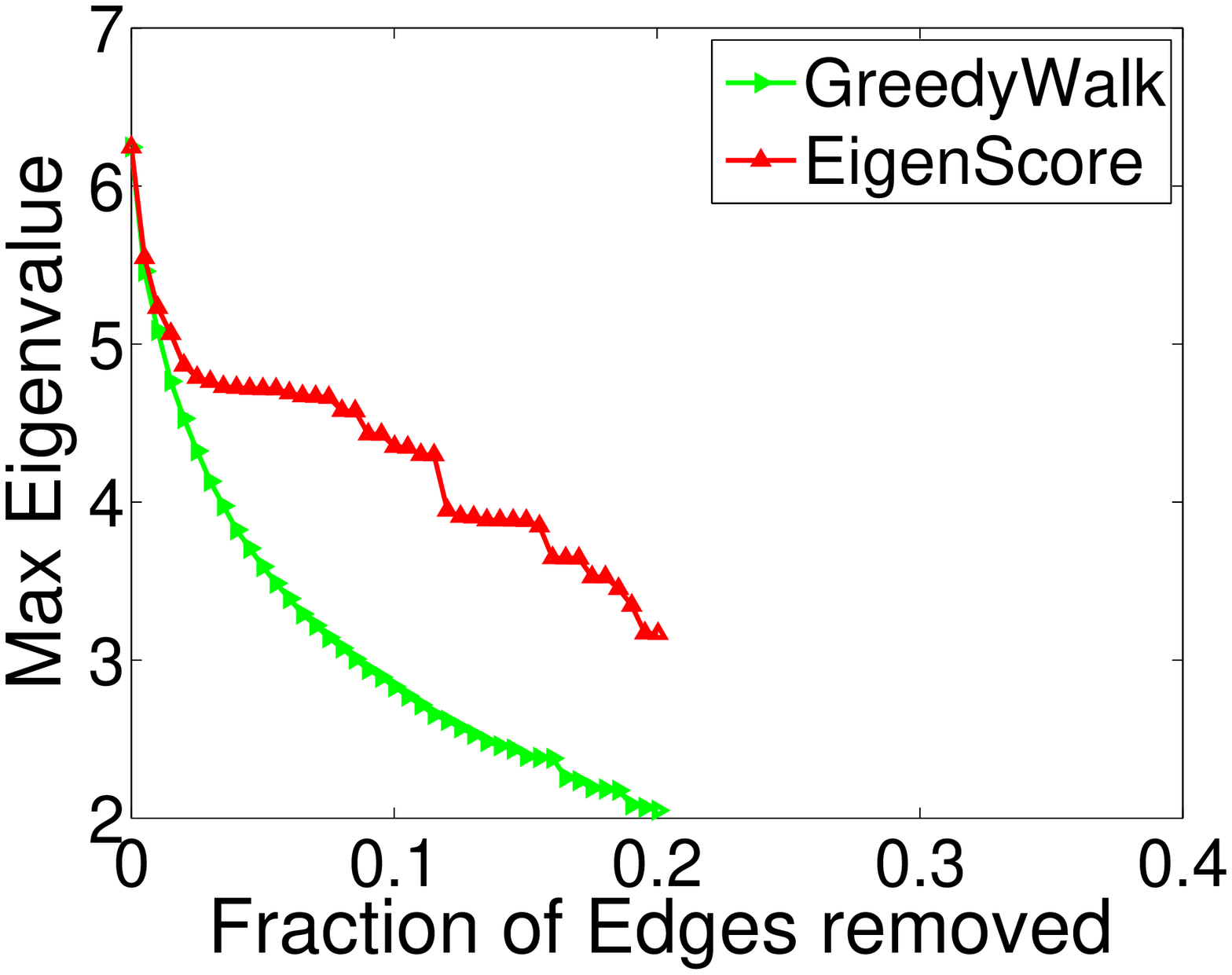}
} & \subfloat[\textsc{SRME-nonuniform: Collaboration GrQc}]{ \label{fig:colgen_nonuni}
 \includegraphics[width=0.42\columnwidth]{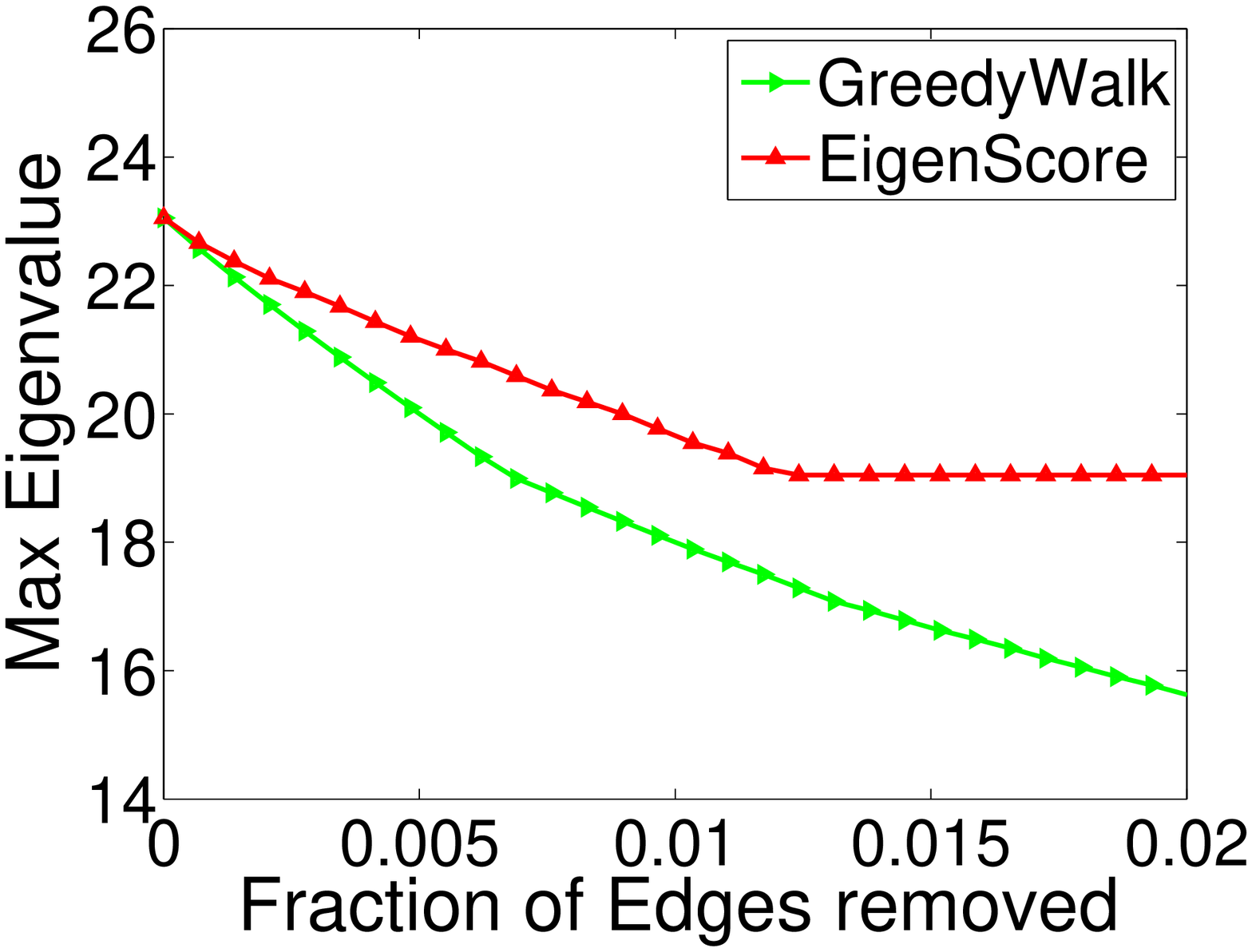}
}
 \\ \subfloat[\textsc{SRMN}: Oregon-1]{ \label{fig:5_node}
 \includegraphics[width=0.42\columnwidth]{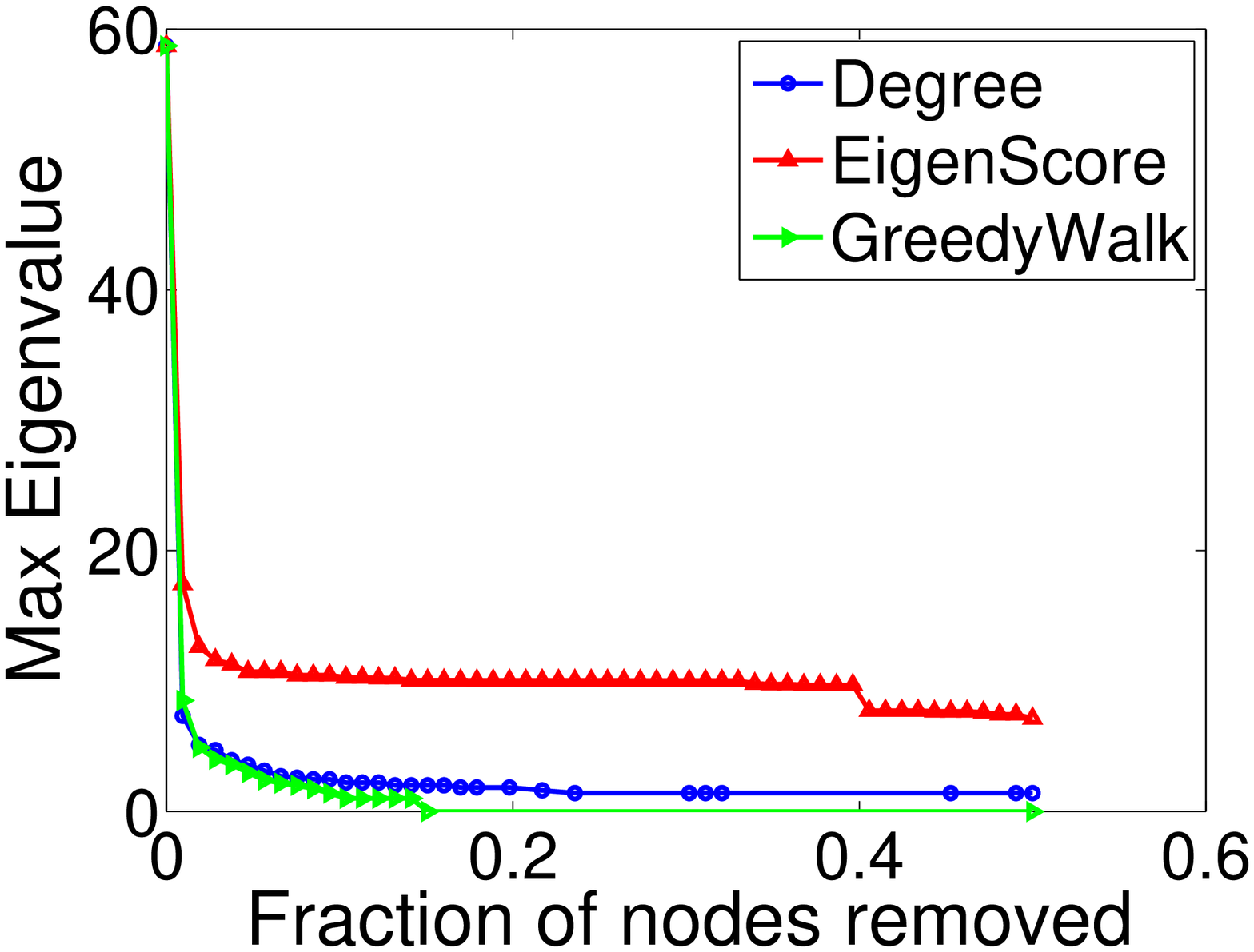}
} & \subfloat[\textsc{SRMN}: Collaboration GrQc]{ \label{fig:colgen_node}
 \includegraphics[width=0.42\columnwidth]{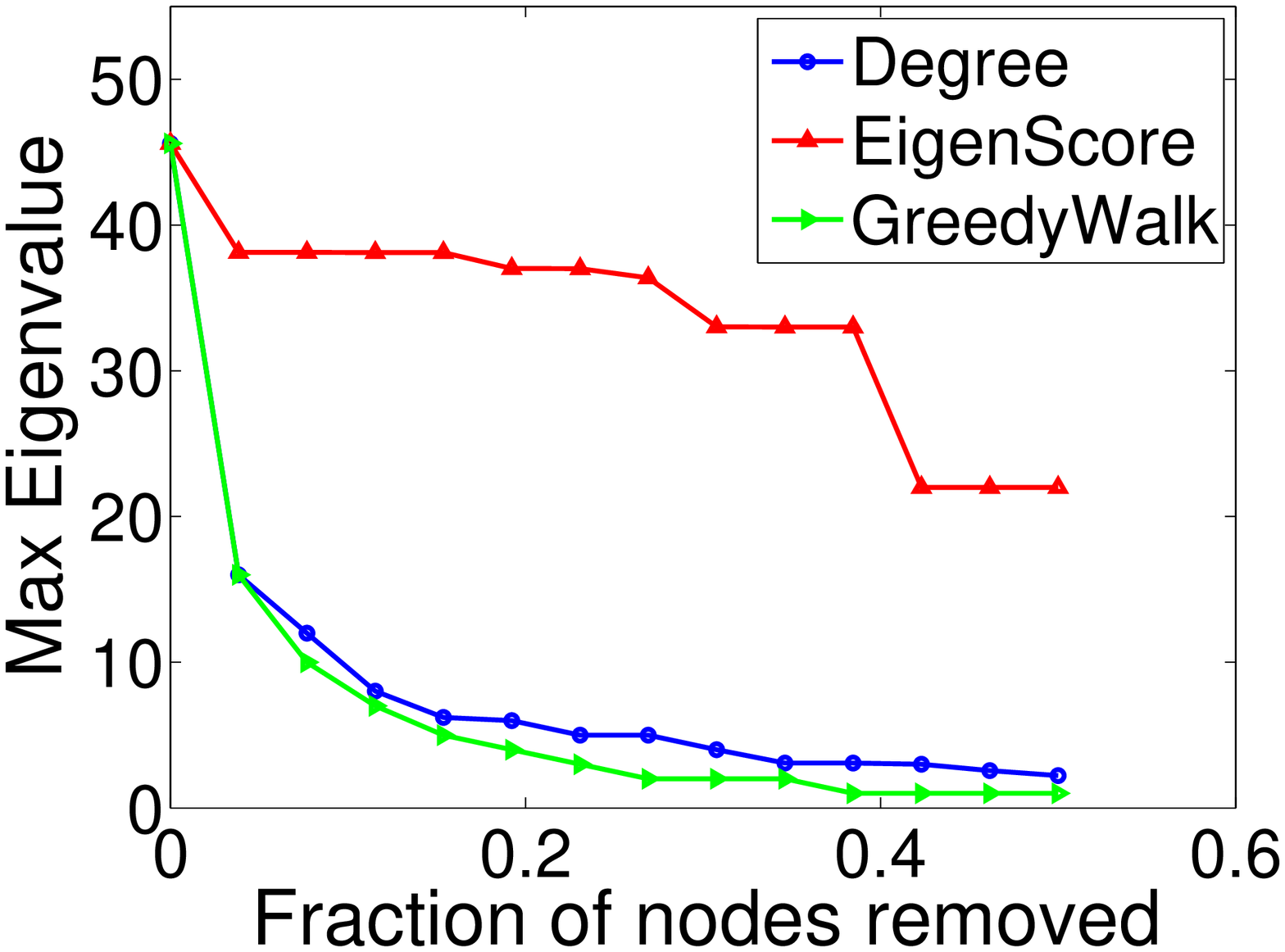}
} \\
\end{tabular}
\caption{ \label{fig:nonuni_srmn} Computing solutions for \textsc{SRME-nonuniform} (\ref{fig:br_nonuni},\ref{fig:colgen_nonuni}) and \textsc{SRMN} (\ref{fig:5_node},\ref{fig:colgen_node})
problem on different networks with \textsc{GreedyWalk} algorithm and \textsc{Degree} and \textsc{EigenScore} heuristics as adapted in Section \ref{sec:extensions}.
The plots show the resultant spectral radius (y-axis) as fractions of edges/nodes are removed (x-axis) with different methods.
%\iffalse For the fraction of edges or nodes removed (x-axis) the plots show the spectral radius of the residual graph (y-axis). The methods are adapted to these settings as described in Section \ref{sec:extensions}.\fi
}
\end{figure}
}
{

%\begin{figure}[ht]
%\centering
%\begin{tabular}{cc}
% \subfloat[\textsc{SRME-nonuniform}]{ \label{fig:colgen_nonuni}
% \includegraphics[width=0.42\columnwidth]{fig/colgen_nonuni.eps}
%} & \subfloat[\textsc{SRMN}]{ \label{fig:colgen_node}
% \includegraphics[width=0.42\columnwidth]{fig/colgen_node.eps}
%} \\
%\end{tabular}
%\caption{ \label{fig:nonuni_srmn} Computing solutions for \textsc{SRME-nonuniform} (\ref{fig:colgen_nonuni}) and \textsc{SRMN} (\ref{fig:colgen_node}) problem on collaboration GrQc network with \textsc{GreedyWalk} algorithm and \textsc{Degree} and \textsc{EigenScore} heuristics as adapted in Section \ref{sec:extensions}. The plots show the resultant spectral radius (y-axis) as fractions of edges/nodes are removed (x-axis) with different methods.  \iffalse For the fraction of edges or nodes removed (x-axis) the plots show the spectral radius of the residual graph (y-axis). The methods are adapted to these settings as described in Section \ref{sec:extensions}.\fi
%}
%\end{figure}

\begin{figure}[ht]
\centering
\begin{tabular}{cc}
\subfloat[Oregon-1]{ \label{fig:5_node}
 \includegraphics[width=0.42\columnwidth]{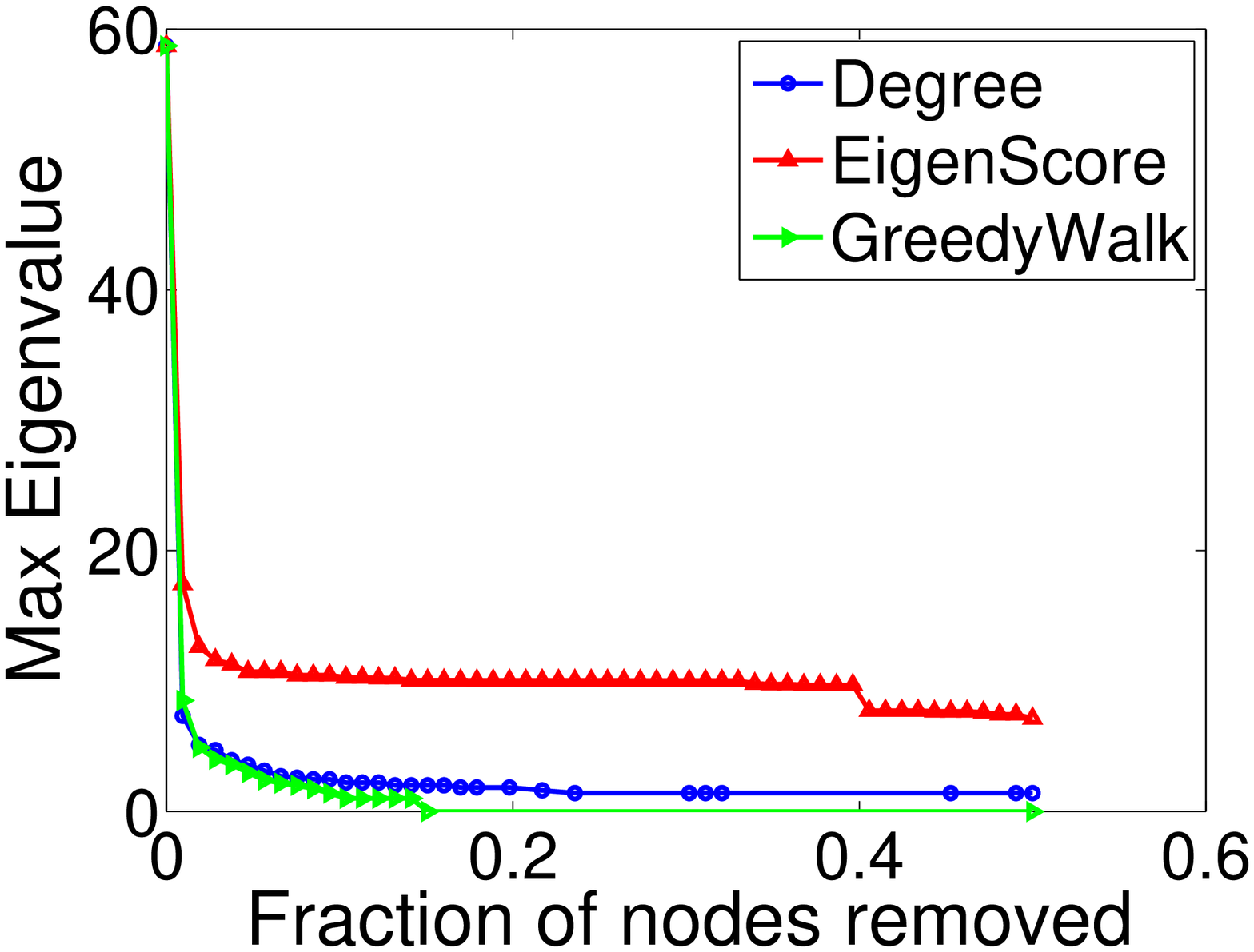}
} & \subfloat[Collaboration GrQc]{ \label{fig:colgen_node}
 \includegraphics[width=0.42\columnwidth]{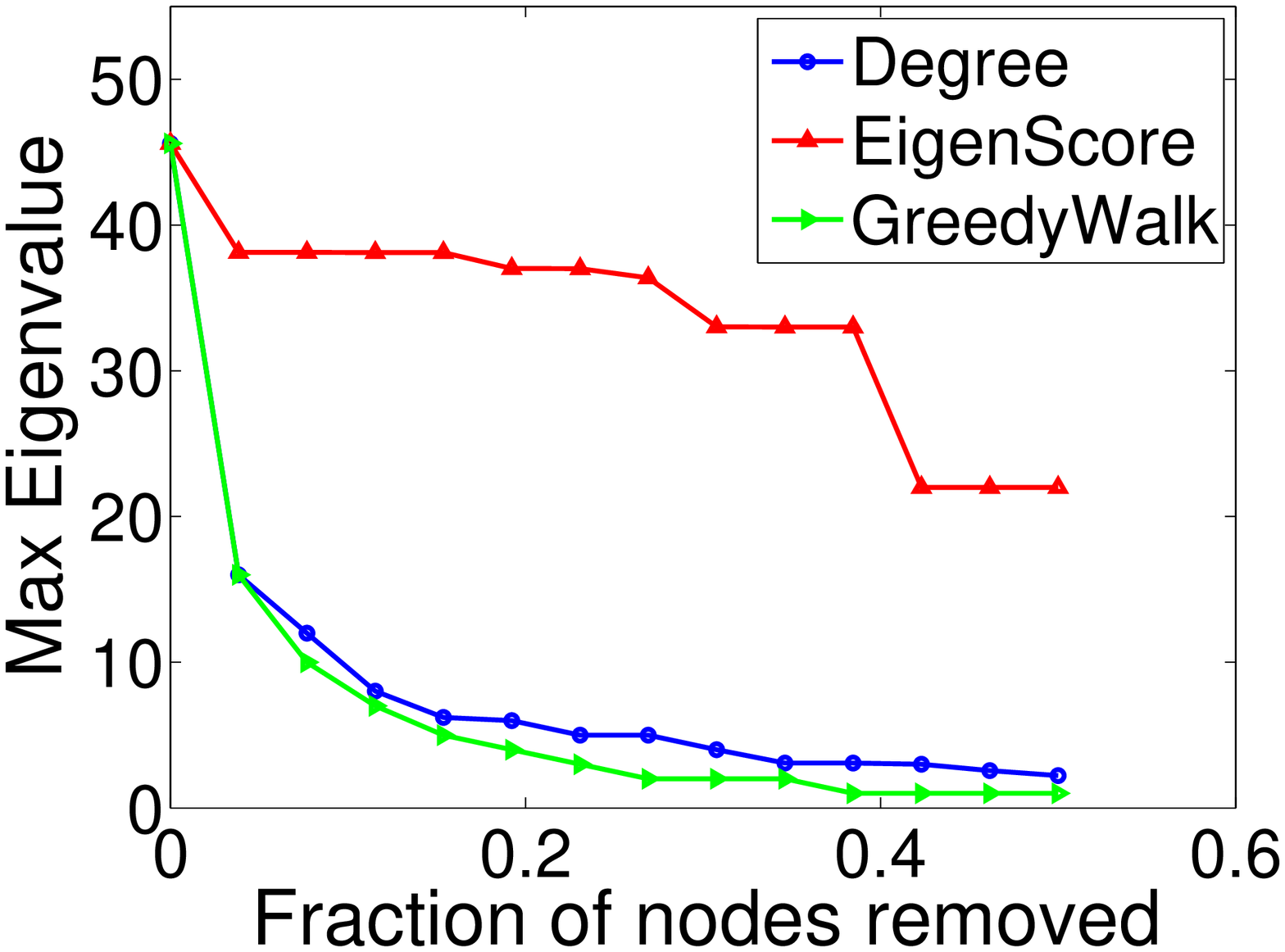}
}\end{tabular}
\caption{ \label{fig:srmn_experiment} Computing solutions for \textsc{SRMN} problem on AS oregon-1 and collaboration GrQc network with \textsc{GreedyWalk} algorithm and \textsc{Degree} and \textsc{EigenScore} heuristics, adapted as shown in Section \ref{sec:extensions}. The plots show the resultant spectral radius (y-axis) as fractions of nodes are removed (x-axis) with different methods.
}
\end{figure}

}

%\iffalse
%\begin{figure}[ht]
%\centering
%\begin{tabular}{cc}
%\subfloat[Barabasi-Albert]{ \label{fig:br_nonuni}
% \includegraphics[scale=0.21]{fig/br_nonuni.eps}
%} &  \subfloat[Collaboration GrQc]{ \label{fig:colgen_nonuni}
% \includegraphics[scale=0.21]{fig/colgen_nonuni.eps}
%}\\
%\end{tabular}
%\caption{ \label{fig:nonuni} \textsc{GreedyWalk} vs \textsc{EigenScore} for
%the \textsc{SRME-nonuniform} problem.  Each plot shows the spectral radius after edge removal (y-axis)
%vs the fraction of edges removed (x-axis). The \textsc{GreedyWalk} method was adapted to this setting
%as in Section \ref{sec:extensions}, and the \textsc{EigenScore} method was run on the matrix
%$B$ of transmission rates.
%}
%\end{figure}
%
%\begin{figure}[ht]
%\centering
%\begin{tabular}{cc}
%\subfloat[Collaboration GrQc]{ \label{fig:5_node}
% \includegraphics[scale=0.2]{fig/5_node.eps}
%} &  \subfloat[AS Oregon-1]{ \label{fig:colgen_node}
% \includegraphics[scale=0.2]{fig/colgen_node.eps}
%}\\
%\end{tabular}
%\caption{ \label{fig:nodemethods}
%Comparison between \textsc{GreedyWalk}, \textsc{Degree} and \textsc{EigenScore} for the
%\textsc{SRMN} problem.
%Each plot shows the spectral radius after node removal (y-axis) vs the fraction of nodes removed.
%}
%\end{figure}
%\fi

\iftoggle{fullversion}
{
$\newline$
\noindent
\textbf{Demographic properties of removed nodes and edges:}
\textsc{GreedyWalk} can also help in getting non-network surrogates for picking nodes/edges. We analyzed the demographic properties of the nodes and edges removed by
\textsc{GreedyWalk} on the Portland contact network \cite{cinet:14}. By doing so, we can hope to use such demographic properties directly, for quicker implementation
and/or when the entire network is not readily available. Figure
\iftoggle{fullversion}{\ref{fig:agegroup_matrix}}{\ref{fig:cw_e_age_mat}} shows the age groups of the end points of the
top $1500$ selected edges by \textsc{GreedyWalk} as a matrix.
Age-groups are partitioned according to \cite{medlock2009}
\iftoggle{fullversion}
{
and shown in table \ref{table:agegroups}. \iffalse appendix~\ref{appsubsec:agegroup}.\fi
}
{ and shown in \cite{spectral-approx-extended}.}
As the figure shows, the edges among age-group \#$11$ (ages $45-49$) and with age-groups \#$8$ (age $30-34$) and
\#$17$ (age $75+$) are picked to a greater extent by \textsc{GreedyWalk}.
\iffalse
Appendix~\ref{subsec:other-plots} shows the results for edges picked by other strategies
\fi
We observe that the edges picked by \textsc{GreedyWalk} have substantially different properties compared to other heuristics \iftoggle{fullversion}{}{(see \cite{spectral-approx-extended})}.
\iftoggle{fullversion}{Figure \ref{fig:agegroup_bar}}{Figure \ref{fig:agegrp_bar}} shows the age groups of the nodes removed by the \textsc{GreedyWalk}
algorithm for the \textsc{SRMN} problem, along with the age group distribution of the entire population.
Observe that more people are selected in age-group numbers 7 to 11 which correspond to ages 25-49.
}
{
}

\iftoggle{fullversion}
{
\begin{table}[ht]
\centering
\caption{\small Age-groups \cite{medlock2009}}
\label{table:agegroups}
\small
\begin{tabular}{c|c||c|c}
\toprule
Age-group & Age & Age-group & Age\\
\midrule
1 & 0 & 10 & 40-44\\ \midrule
2 & 1-4 & 11 & 45-49\\ \midrule
3 & 5-9 & 12 & 50-54\\ \midrule
4 & 10-14 & 13 & 55-59\\ \midrule
5 & 15-19 & 14 & 60-64\\ \midrule
6 & 20-24 & 15 & 65-69\\ \midrule
7 & 25-29 & 16 & 70-74\\ \midrule
8 & 30-34 & 17 & 75+\\ \midrule
9 & 35-39 &  & \\ \bottomrule
\end{tabular}
\end{table}
}{}

\iftoggle{fullversion}
{
\begin{figure}[ht]
\centering
\begin{tabular}{cc}
 \subfloat[Removed Edges]{ \label{fig:cw_e_age_mat}
 \includegraphics[width=0.44\columnwidth]{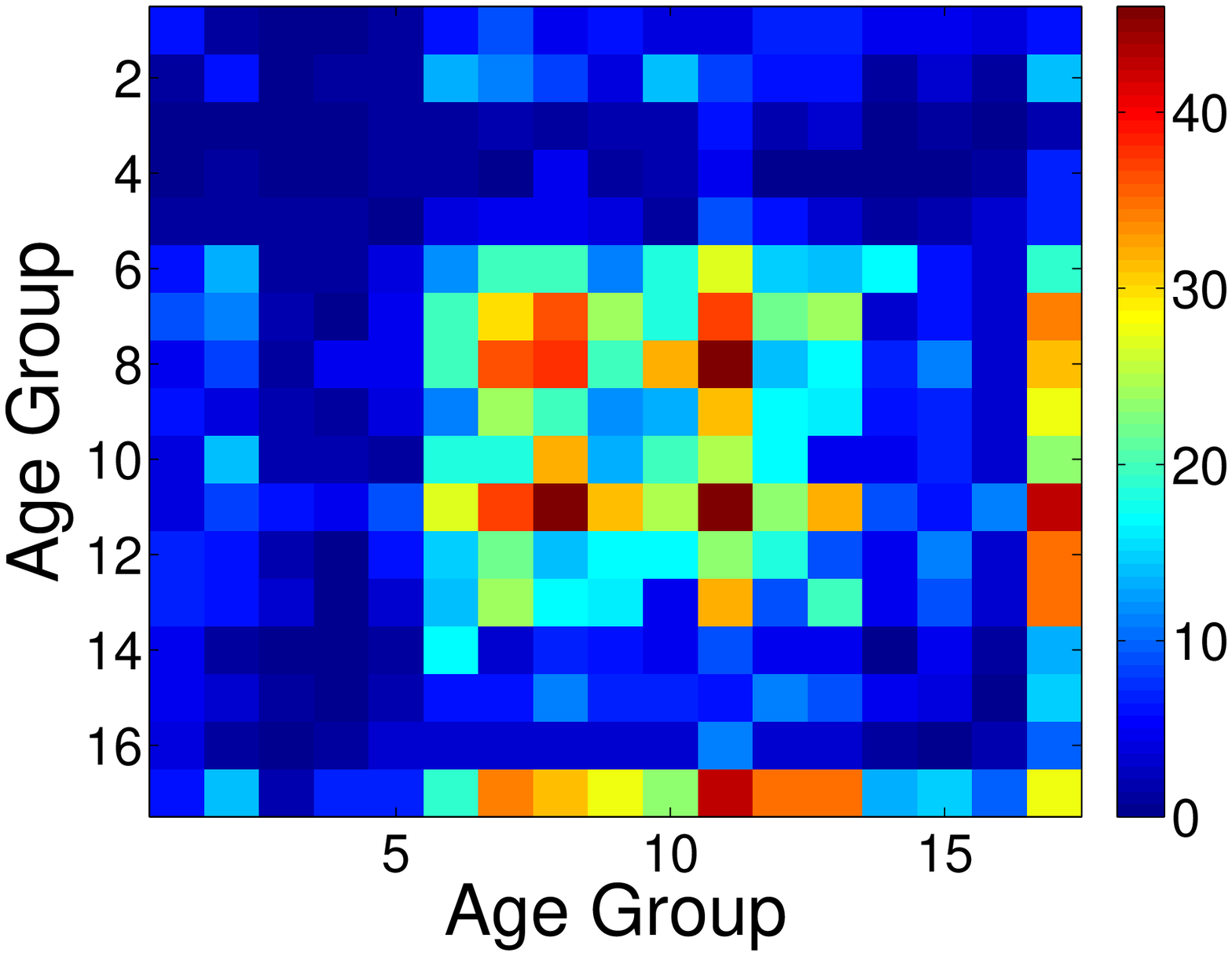}
} & \subfloat[ProductDegree]{ \label{fig:deg_e_age_mat}
 \includegraphics[width=0.44\columnwidth]{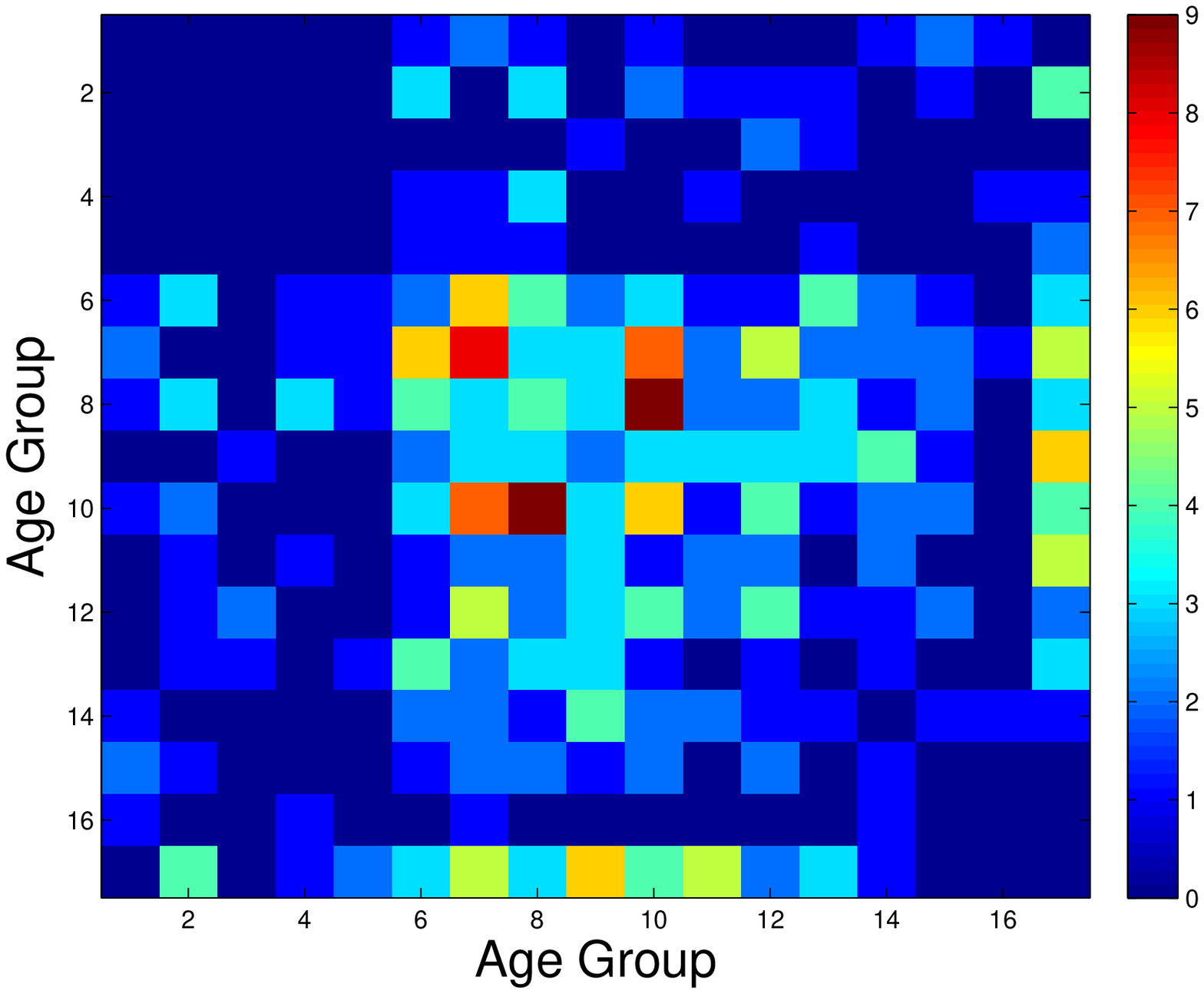}
} \\ \subfloat[ProductEigenscore]{ \label{fig:eig_e_age_mat}
 \includegraphics[width=0.44\columnwidth]{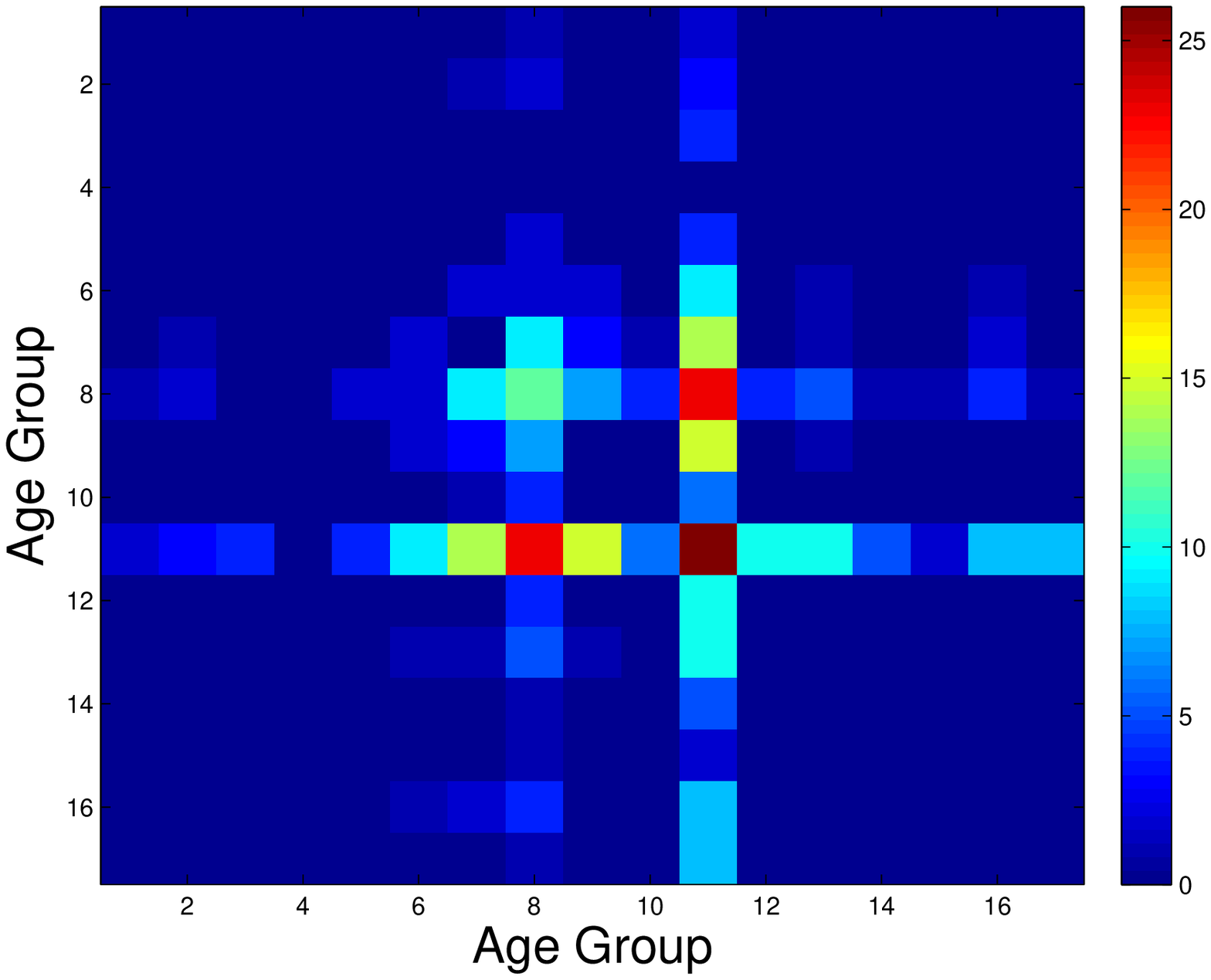}
} & \subfloat[Hybrid]{ \label{fig:hbd_e_age_mat}
 \includegraphics[width=0.44\columnwidth]{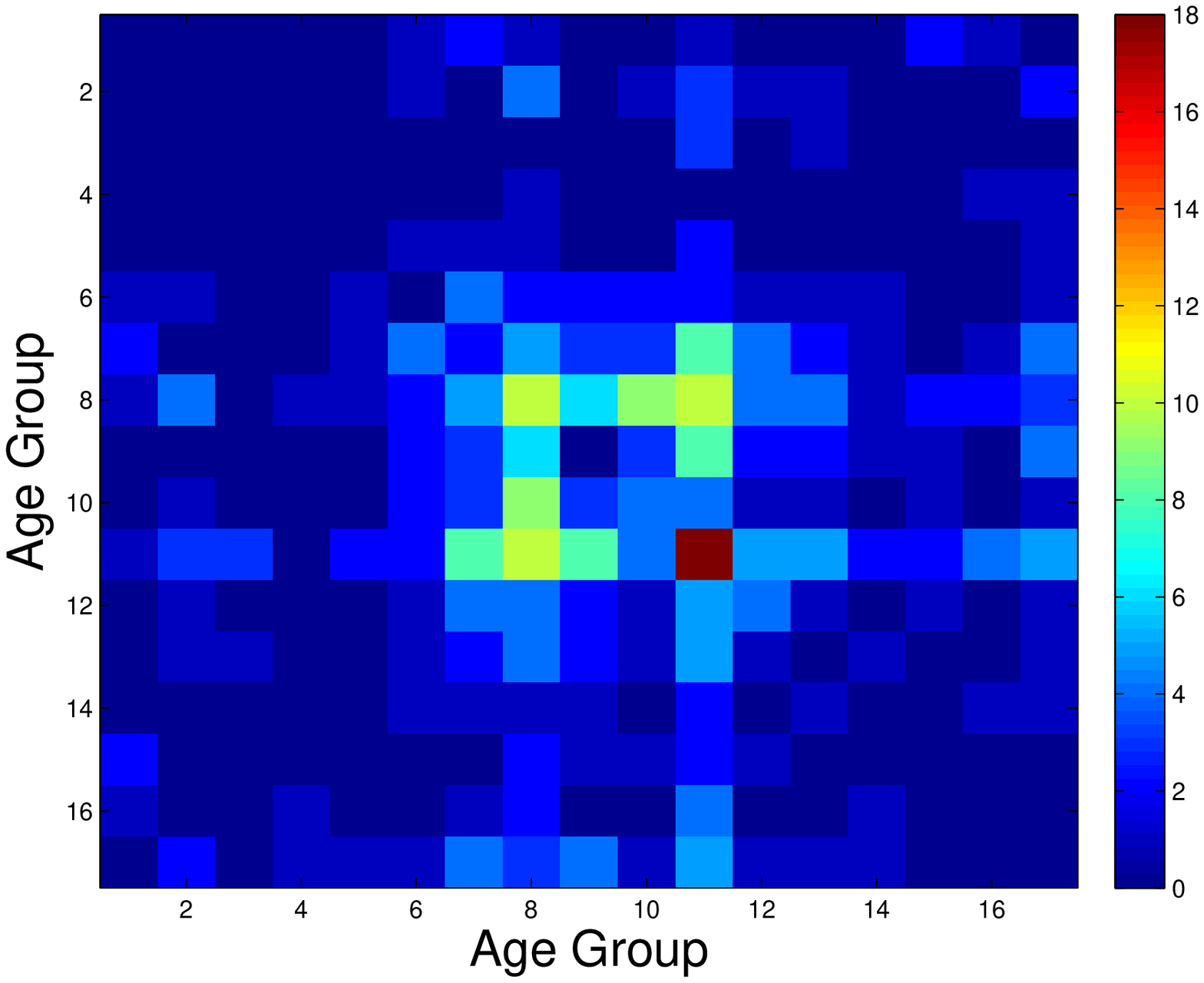}
}
\end{tabular}
\caption{ \label{fig:agegroup_matrix}
Age-Group matrix of the top 1500 removed edges}
\end{figure}

\begin{figure}[ht]
\centering
\begin{tabular}{c}
 \subfloat[Removed Nodes]{ \label{fig:agegrp_bar}
 \includegraphics[width=0.44\columnwidth]{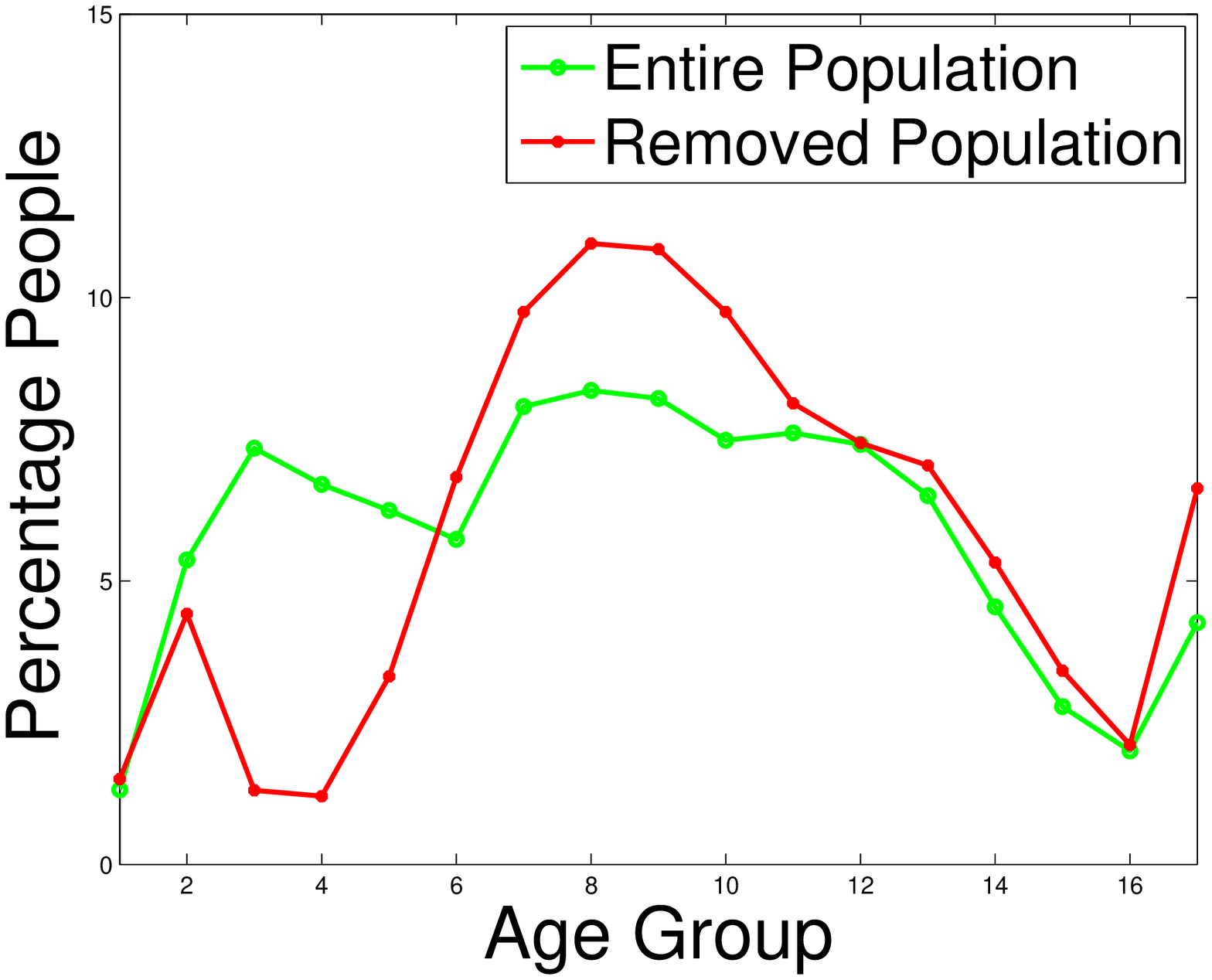}
}
\end{tabular}
\caption{ \label{fig:agegroup_bar} Age-group of 1500 removed nodes  with \textsc{GreedyWalk} from Portland contact graph.}
\end{figure}

\begin{figure}[ht]
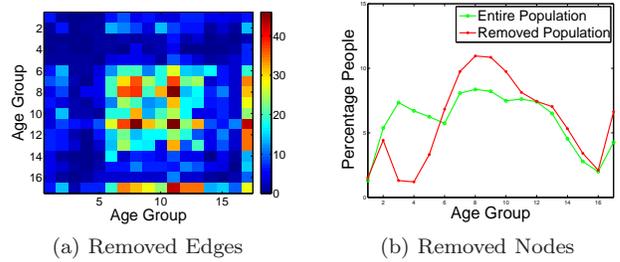

\centering
\begin{tabular}{cc}
\subfloat[Removed Edges]{ \label{fig:cw_e_age_mat}
 \includegraphics[width=0.44\columnwidth]{cw_e_age_mat.eps}
} & \subfloat[Removed Nodes]{ \label{fig:agegrp_bar}
 \includegraphics[width=0.44\columnwidth]{agegroup_pt_cw_fplot.eps}
}
\end{tabular}
\caption{ \label{fig:agegroup_matrix_bar} (\ref{fig:cw_e_age_mat}) Age-Group matrix of the top 1500 removed edges
and (\ref{fig:agegrp_bar}) Age-group of 1500 removed nodes  with \textsc{GreedyWalk} from Portland contact graph. }
\end{figure}
} %!iftoggle{fullversion}
{
}

\iffalse
\begin{figure}[ht]
\centering
\begin{tabular}{c}
\subfloat[]{ \label{fig:cw_e_age_mat}
 \includegraphics[width=0.7\columnwidth]{cw_e_age_mat.eps}
}
\end{tabular}
\caption{ \label{fig:agegroup_matrix}
\small Age-Group \cite{medlock2009} matrix of the top 1500 removed edges with \textsc{GreedyWalk} from Portland contact graph}
\end{figure}

\fi

$\newline$
\noindent
\textbf{Main observations:}
%%%In summary, the main observations from the experiments are as follows.

\iftoggle{fullversion}
{
\noindent
1.~\textsc{GreedyWalk} performs consistently better than existing heurisitics in removing nodes or edges in both static and variable transmission rate settings.

\noindent
2.~Sparsification helps in improving the speed of \textsc{GreedyWalk} without effecting the solution quality.

\noindent
3.~\textsc{GreedyWalk} performs best for walk-lengths of $k=2\log n$.

\noindent
4.~\textsc{GreedyWalk} can potentially help in picking more accurate non-network surrogates.
}
{
\noindent
1.~\textsc{GreedyWalk} performs consistently better than existing heurisitics in removing nodes or edges.
%in both static and variable transmission rate settings.

\noindent
2.Performance of just one iteration of ~\textsc{PrimalDual} comes very close to ~\textsc{GreedyWalk}.

\noindent
3.~Sparsification helps in improving the speed of \textsc{GreedyWalk} without effecting the solution quality.

}
%%\noindent
%%4.~Edges and nodes involving people of age between $45-49$ in $25-49$ respectively in Portland contact network are more effective in reducing spectral radius.

%%
\section{Related Work}
\label{sec:related}

Related work comes from multiple areas: epidemiology, immunization algorithms and other optimization algorithms. There is  general research interest in studying dynamic processes on large graphs, (a) blogs and propagations~\cite{Gruhl:2004,Kumar:2003}, (b) information cascades~\cite{Goldenberg:2001,Granovetter:1978} and (c) marketing and product penetration~\cite{Rogers:1962}. These dynamic processes are all closely related to virus propagation.

\par \noindent
\textbf{Epidemiology:}
A classical text on epidemic models and analysis is by May and Anderson~\cite{andersonmay}.
%%%The classical texts on epidemic models and analysis are May and Anderson~\cite{andersonmay} and Hethcote~\cite{hethcote2000}.
Most work in epidemiology is focused on {\em homogeneous models}~\cite{Bailey1975Diseases,andersonmay}.
%%%Widely-studied epidemiological models include {\em homogeneous
%%%models}~\cite{Bailey1975Diseases,McKendrick1926Medical,andersonmay}
%%%  which assume that every individual has equal contact with others in the population.
Here we study network based models.
Much work has gone into in finding epidemic thresholds (minimum virulence of a
virus which results in an epidemic) for a variety of networks~\cite{SV2002Finite,Wang03Epidemic,ganesh+topology05,aditya12}.

\par \noindent
\textbf{Immunization:}  There has been much work on finding optimal strategies for vaccine allocation~\cite{Briesemeister03Epidemic,Madar04Immunization,Chen:2010}. Cohen et al \cite{Cohen03Efficient} studied the popular {\em acquaintance} immunization
policy (pick a random person, and immunize one of its neighbors at random).
Using game theory, Aspnes et al.~\cite{Aspnes:2005} developed inoculation strategies for victims of viruses under random starting points.
Kuhlman et al.~\cite{Kuhlman:2013} studied two formulations of the problem of blocking a contagion through edge removals under the model of discrete dynamical systems.
As already mentioned Tong et al.~\cite{Tong@ICDM10,tong:cikm12}, Van Miegham et al.~\cite{vanmieghem:ton12}, Prakash et al.~\cite{prakash:fracimmu:2013} and Chakrabarti et al.~\cite{chakrabarti08} proposed various node-based and edge-based immunization algorithms based on minimizing the largest eigenvalue of the graph. Other non-spectral approaches for immunization have been studied by Budak et al~\cite{budak11}, He et al~\cite{He12} and Khalil et al.~\cite{khalil14}.

\par \noindent
\textbf{Other Optimization Problems:} Other diffusion based optimization problems include the influence maximization problem, which was introduced by Domingos and Richardson~\cite{richardson2002mining},  and formulated by  Kempe et. al.~\cite{Kempe03Maximizing}  as a combinatorial optimization problem. They proved it is NP-Hard and also gave a simple $1-1/e$ approximation based on the submodularity of expected spread of a set of starting seeds. Other such problems where we
  wish to select a subset of `{\em important}' vertices on graphs, include `outbreak detection'~\cite{Leskovec@KDD07} and `finding most-likely culprits of epidemics'~\cite{lappas:10:effectors,Prakash@ICDM12}.

\section{Conclusions}
\label{sec:conc}
We study the problem of reducing the spectral radius of a graph to control the spread of
epidemics by removing edges (the \textsc{SRME} problem) or nodes
(the \textsc{SRMN} problem).
We have developed a suite of algorithms for these problems, which give the first
rigorous bounds for these problems.  Our main algorithm \textsc{GreedyWalk}
performs consistently better than all other heuristics for these problems, in
all networks we studied. We also develop variants that improve the running time by sparsification,
and improve the approximation guarantee using a primal dual approach.
These algorithms exploit the connection between the graph spectrum and closed
walks in the graph, and perform better than all other heuristics.
Improving the running time of these algorithms is a direction for further research.
We expect these techniques could potentially help in optimizing other objectives
related to spectral properties, e.g., \emph{robustness} \cite{chansdm2014}, and
in other problems related to the design of interventions to control the spread of epidemics.
\iffalse
We also develop a hybrid heuristic that builds on two
prior strategies involving the first eigenvector and degrees, which performs
very well in most graphs.
There are several directions for further research.  First, we would
like to extend the \textsc{PrimalDual}
algorithm and its analysis to the more general versions (involving
node and edge labels) of the spectral radius minimization problem.
We would also like to explore more general analyses of
the {\sc ProductDegree} and {\sc EigenScore} heuristics.  Finally, we
would like analyze other measures useful for epidemic spread,
including the expansion of the graph, the isoperimetric constant of
the network, and incorporate intervention efficacy rates in the
problem formulations.
\fi

\noindent
\textbf{Acknowledgments}.
This work has been partially supported by the following grants:
DTRA Grant HDTRA1-11-1-0016,
DTRA CNIMS Contract HDTRA1-11-D-0016-0010,
NSF Career CNS 0845700, NSF ICES CCF-1216000,
NSF NETSE Grant CNS-1011769,
DOE DE-SC0003957,
National Science Foundation Grant IIS-1353346 and
Maryland Procurement Office contract H98230-14-C0127.
Also supported by the Intelligence Advanced Research Projects Activity (IARPA) via Department of Interior National Business Center (DoI/NBC) contract number D12PC000337, the US Government is authorized to reproduce and distribute reprints for Governmental purposes notwithstanding any copyright annotation thereon.\\
Disclaimer: The views and conclusions contained herein are those of the authors and should not be interpreted as necessarily representing the official policies or endorsements, either expressed or implied, of IARPA, DoI/NBC, or the US Government.

\bibliographystyle{abbrv}
\bibliography{references,all-aditya,reference}

%\end{document}

\iftoggle{fullversion}
{
\newpage
\appendix
%\large{\textbf{Appendix}}
\section{Appendix}

\subsection{\textsc{GreedyWalk} with Dynamic Programming Approach}
\label{appsubsec:dp}

$\newline$
\noindent
\textbf{Main idea}: we adapt a dynamic programming approach in sparse graphs
to avoid matrix multiplication, that leads to lower space complexity, thereby
allowing us to scale to larger graphs. We then observe that the number of walks
does not need to be recomputed each time an edge is deleted.

Let $H_{\overrightarrow{uv}}(G,x,l)$ denote the number of walks of length
$l$ from node $u$ through edge $(u,v)$ as the first edge to node $x$ in $G$.
It is easy to see that, $H_{\overrightarrow{uv}}(G,u,k)=\walks(e, G,k)$.
Algorithm \textsc{ClosedWalkDP} describes how to
compute $H_{\overrightarrow{uv}}(G,u,l)=\walks(e,G,k)$. In the algorithm, $N(x)$ denotes the neighbors of node $x$ in $G$.

\begin{algorithm}{}
\label{alg:k-ClosedWalkDP}
\caption{$\textsc{ClosedWalkDP}(G,(u,v),k)$}
\begin{algorithmic}[1]
\INPUT{$G, (u,v), k\ge 2 $}
\OUTPUT{Number of closed walks of length $k$ in $G$ containing $(u,v)$}
\STATE Let $H_{\overrightarrow{uv}}(G,v,1)=1$, $H_{\overrightarrow{uv}}(G,x,1)=0$, $\forall x\in V \setminus \{v\}$ \\
%% $H_{\overrightarrow{vu}}(G,u,1)=1$ \;
%% $H_{\overrightarrow{vu}}(G,x,1)=0$,  $\forall x\in V \setminus \{u\}$ \;
\FOR{$l=2$ to $k$}
\STATE \mbox{$H_{\overrightarrow{uv}}(G,x,l)=\sum_{y\in N(x)}H_{\overrightarrow{uv}}(G,y,l-1)$, $\forall x\in V$}\\
\ENDFOR
%% \For{$l=2$ to $k$}{
%% $H_{\overrightarrow{vu}}(G,x,l)=\sum_{y\in N_G(x)}H_{\overrightarrow{vu}}(G,y,l-1)$, $\forall x\in V(G)$\;
%% }
\STATE return $H_{\overrightarrow{uv}}(G,u,k)$ %%+H_{\overrightarrow{vu}}(G,v,k)$
\end{algorithmic}
\end{algorithm}

Next, we describe in Algorithm \textsc{GreedyEdgeChoice} how the
greedy edge choice in line 4 of Algorithm \textsc{GreedyWalk}
is implemented efficiently.  We make use of the fact that
$\walks(e,G',k)\le\walks(e,G,k)$ for any $G'\subset G$.  In every
iteration of Algorithm \textsc{GreedyEdgeChoice}, potentially, we need to
update $f(\cdot)$ for all edges in $E\setminus E'$. However, in practice,
we observe that the number of such updates is very small compared to
$|E\setminus E'|$.

%% \iffalse
%% \noindent
%% \emph{Updating after edge removal:}
%% Note that,
%% $\walks(e,G',k)\le\walks(e,G,k)$ for any $G'\subset G$. We can use this
%% simple observation to speed up the process further. The method is in
%% Algorithm~\ref{alg:edgeUpdate}.
%% %% Let $e_1, e_2, e_3, ...,e_m$ be the edges of $G$ such that
%% %% $\walks(e_1,G,k)\ge\walks(e_2,G,k)\ge\cdots\ge\walks(e_m,G,k)$. This
%% %% implies that $e_1$ is removed from the graph. In the residual graph $G'$,
%% %% we will recompute walks of the other edges in the order $e_2,e_3,...,e_m$.
%% %% Suppose $\walks(e_2,G,k)\ge \walks(e_p,G,k)$ for some $p$ ($3\le p\le m$).
%% %% Then, we need to compute the walks only for $2\le i \le p-1$.
%% %% In practice, we observe that $p\ll m$.
%% \fi
%%
%%

\begin{algorithm}{}
\label{alg:edgeUpdate}
\caption{\textsc{GreedyEdgeChoice}}
\begin{algorithmic}[1]
\INPUT{$G, T, c(\cdot)$}
\OUTPUT{Edge set $E'$}
\STATE Initialize $E'\leftarrow \phi$ and $\forall e\in E$, let $f(e)=\walks(e,G,k)$ \tcp{computed by \textsc{ClosedWalkDP}}\\
\WHILE{$W_k(G[E\setminus E'])\geq nT^k$}
\STATE Order edges of $E\setminus E'$ in the decreasing order of $f(.)$ values. Let $e_1$ be the first edge.
\STATE $E'\leftarrow E'\cup\{e_1\}$
\FOR {$j=2,\ldots,|E\setminus E'|$}
\STATE  Update $f(e_j)=\walks(e_j,G[E\setminus E'],k)$.\\
\IF {$f(e_j)\ge f(e_{j+1})$}
\STATE Exit from the for loop
\ENDIF
\ENDFOR
\ENDWHILE
\end{algorithmic}
\end{algorithm}

\noindent
\emph{Running time and space complexity:} Let $n=|V|$, $m=|E|$. Note that, $\textsc{ClosedWalkDP}(G,e,k)$ takes $2mk$ time to
compute $\walks(e,G,k)$. Therefore, computing $\walks(e,G,k)$ for all the edges takes $2m^2 k = O(n^2 k)$, assuming
$m=\Theta(n)$ in real world networks.
Since, for computing $H_{\overrightarrow{uv}}(G,x,l)$, $\forall x\in V$,  $\textsc{ClosedWalkE}(k)$
needs to look only at $H_{uv}(G,y,l-1)$, $\forall y\in V$, therefore, the
space complexity is $\Theta(n)$.

\subsection{Non-uniform transmission rates}
\label{sec:non-uniform-appendix}
\iffalse
We now consider the more general setting, where the $\beta_{ij}$'s are non-uniform,
and might depend, e.g., on the demographics of the end-points $i$ and $j$, such as
age, as is commonly assumed in epidemiology \cite{medlock2009}.
\fi

$\newline$
Let $B=(\beta_{ij})$ denote the matrix of the transmission rates.
We assume the rates are symmetric, i.e., $\beta_{ij}=\beta_{ji}$.
In this case, the
sufficient condition for the epidemic to die out is slightly different,
and is stated below.

\begin{lemma}
\label{lemma:nonuniform}
Let $B$ be the matrix of transmission rates, and let $\delta$ be the recovery rate
in the SIS model. If $\rho(B) <\delta$, the time to extinction, $\tau$ satisfies
\[
\expect[\tau]\leq \frac{\log{n}+1}{\delta - \rho(B)}
\]
\end{lemma}

For the case of uniform costs, i.e., $c(e)=1$ for all edges $e$, this motivates the following problem:
\begin{Definition}{\textsc{SRME-nonuniform} problem}
Given an undirected graph $G=(V, E)$, with transmission rate $\beta_{ij}$ for each $(i,j)\in E$
and recovery rate $\delta$, find the smallest set $E'\subseteq E$ such that
$\rho(B(G[E-E'])) \leq\delta$.
\end{Definition}

In this section, we use $\eopt$ to denote the optimum solution to
\textsc{SRME-nonuniform}$(G, B, \delta)$.  Our algorithm
\textsc{GreedyWalk-nonuniform} adapts \textsc{GreedyWalk} to a
weighted covering problem. We need to refine the definitions used
earlier. For walk $w\in\mathcal{W}_k(G)$, let $f(w)=\prod_{e=(ij)\in
  E(w)} \beta_{ij}^{\ct(e,w)}$ denote its weight, where $\ct(e,w)$ is the
number of occurrences of edge $e$ in walk $w$; for a set $W'$ of
walks, let $f(W')=\sum_{w\in W'} f(w)$ denote the total weight of
$W'$.  In the algorithm, we will need to compute $f(W_k(G))$, which is
done by modifying the recurrence used in Algorithm
\textsc{CountWalks}$(G)$ to compute $W_k(G)$:
\[
f(W_k(G)) = B^k_{nn} + f(W_k(G[V-\{n\}]).
\]

Let $f(e, G)=\sum_{w: e\in w} f(w)$ denote the total weight of walks
containing edge $e$; $f(e,G)=B^k_e$.  Algorithm
\textsc{GreedyWalk-nonuniform} involves the following steps:
\begin{itemize}
\item
$E'=\phi$
\item
while $f(W_k(G[E-E']))\geq n\delta$:
\begin{itemize}
\item
Pick the $e\in E\setminus E'$ that maximizes $(\min\{n\delta -
f(W_k(G[E-E'])), f(e, G[E\setminus E'])\})/c(e)$.
\item
$E'\leftarrow E'\cup\{e\}$
\end{itemize}
\end{itemize}

\begin{lemma}
\label{lemma:greedywalk-nonuniform}
Let $E'$ denote the set of edges found by Algorithm
\textsc{GreedyWalk-nonuniform}.  Given any constant $\epsilon > 0$, let $k$ be an even integer greater than
$\log{n}/\log(1+\epsilon/3)$, we have $\rho(B(G[E\setminus E'])) \leq (1+\epsilon)\delta$
and $|c(E')| = O(c(\eopt)\log n\log \Delta)$.
\end{lemma}

%The proof of Lemma \ref{lemma:greedywalk-nonuniform} is discussed in ??.
%Appendix \ref{sec:nonuniform-appendix}.

\begin{proof}
The bound on $\rho(B(G[E\setminus E']))$ follows on the same lines as
the proof of Lemma \ref{lemma:greedywalk}. The main difference is that
the proof of \cite{slavic:ipl97} does not consider the case of weights
associated with elements.  But, as we argue now, the same approach for
analyzing greedy algorithms extends to our case, and we show
$c(E')=O(c(E_\text{HITOPT}) \log n)$.

We partition the iterations of Algorithm
\textsc{GreedyWalk-nonuniform} into $O(\log n)$ phases.  Each phase,
ends at the first iteration when the total weight that needs to be
further covered goes down by a factor of at least $2$.  So if $F$ is
the weight that needs to be covered at the start of the phase, in
every iteration of the phase, there exists an edge $e$ (which is in an
optimum solution) such that $f(e,G[E\setminus E'])/c(e) \ge
F/(2c(E_\text{HITOPT}))$.  Thus, the total cost of the edges selected
in the phase is at most $2c(E_\text{HITOPT})$.  Since the ratio of
$n\delta$ over the minimum weight of a walk is polynomial in $n$, the
total number of phases is $O(\log n)$.  Adding over all phases then
yields the desired bound on $c(E')$.  Putting this together with the
rest of the proof of Lemma \ref{lemma:greedywalk} yields the desired
bound.
\end{proof}

\subsection{Node version: \textsc{SRMN} problem}
\label{appsubsec:srmn}
$\newline$
Recall the definition of $\walks(v, G, k)$ from Section \ref{sec:preliminaries}.
Let $G[V'']$ denote the subgraph of $G=(V, E)$ induced by subset $V''\subset V$.
We modify Algorithm \textsc{GreedyWalk} \iffalse \ref{alg:greedywalk} note this \fi to work for the \textsc{SRMN} problem
in the following manner:

\begin{algorithm}{}
\label{alg:srmn}
\caption{Algorithm \textsc{GreedyWalkSRMN}}
\begin{algorithmic}[1]
\STATE Initialize $V'\leftarrow \phi$
\WHILE {$W_k(G[V\setminus V'])\geq nT^k$}
\STATE $r \leftarrow W_k(G[E\setminus E']) - nT^k$
\STATE Pick $v\in V\setminus V'$ that maximizes $\frac{\min \{r, \walks(v, G[V\setminus V'],k)\}}{c(v)}$
\STATE $V'\leftarrow V'\cup\{v\}$
\ENDWHILE
\end{algorithmic}
\end{algorithm}

%\iffalse
%\begin{enumerate}
%\item
%Initialize $V'\leftarrow \phi$
%\item
%while $W_k(G[V\setminus V'])\geq nT^k$:
%\begin{enumerate}
%\item
%$r \leftarrow W_k(G[E\setminus E']) - nT^k$
%\item
%Pick $v\in V\setminus V'$ that maximizes $\frac{\min \{r, \walks(v, G[V\setminus V'],k)\}}{c(v)}$
%\item
%$V'\leftarrow V'\cup\{v\}$
%\end{enumerate}
%\end{enumerate}
%\fi

It can be shown on the same lines as Lemma \ref{lemma:greedywalk} that this gives
a solution of cost $O(c(\eopt(T))\log n\log \Delta)$, where $c(\eopt(T))$ denotes the
cost of the optimal solution to \textsc{SRMN} problem. Further, the same running time
bounds as in Sections \ref{sec:matrix} and \ref{sec:dynamic} hold.

\subsection{Proof of Theorem~\ref{thm:productdegree}}\label{sec:worstcaseproof}
$\newline$
\noindent
\textbf{Construction:} We construct a graph $G$ for which the statement holds. For convenience
let us assume that $T'$ is a positive integer. $G$ contains (1)~a clique
$G_1$ on ${T'+1}$ nodes; (2)~a caterpillar tree $G_2$, which comprises of a path $v_1v_2\cdots v_{q-1}$
with $v_i$ adjacent to $T'$ leaves each and (3)~$G_3$, a star graph
with $(T'+1)^2$ leaves and central vertex denoted by $v_q$. We connect
$G_1$ to $G_2$ by $(v_0,v_1)$ where, $v_0$ is some node in $G_1$ and
$G_2$ is connected to $G_3$ by the edge $(v_q,v_{q-1})$. Note that
$q=\frac{n-(T'+1)^2-T'}{T'}$ and $\lambda_1(G)\ge\lambda_1(G_3)=T'+1$. Again,
here we assume that $q$ is an integer.

$\newline$
\noindent
\textbf{Bound on $c(\eopt)$:} We will show that $c(\eopt)\le 2T'+3$.
Removing the edges $(v_0,v_1)$ and $(v_{q-1},v_q)$ isolates
the components $G_1$, $G_2$ and $G_3$. $G_1$ is a clique on $T'+1$
nodes and on removing one edge, its spectral radius decreases below
$T'$. $G_2$ is a star with $(T'+1)^2$ leaves and therefore, on removing
at most $(T'+1)^2-(T'^2+1)$ edges, its spectral radius decreases
below $T'$. It can be shown that $\lambda_1(G_2)\le \sqrt{T'}+2$.
%% Now we will show that $\lambda_1(G_2)\le \sqrt{T'}+2$
%% and therefore, assuming $T'$ is sufficiently high, no edge needs to
%% be removed from it. We partition the edge set of $G_2$ as follows:
%% (a)~$G_2'$ contains the edges $(v_i,v_{i+1})$, for~$i=1,\ldots,q-1$
%% and therefore corresponds to the path $v_1v_2\cdots v_{q-1}$ and
%% (b)~$G_2''$ is induced by the remaining set of edges and corresponds to
%% a disjoint collection of star graphs with $T'$ leaves each. Therefore,
%% $\lambda_1(G_2)\le\lambda_1(G_2')+\lambda_1(G_2'')\le\sqrt{T'}+2$. Hence,
%% $OPT\le 2T'+3$.

Now we will demonstrate that all the four algorithms score the edges
$(v_i,v_{i+1})$, $i=0,\ldots,q-2$ above any edge belonging to the clique
$G_1$. However, the spectral radius cannot be brought down below $T'$
until at least one edge in $G_1$ is removed. Therefore, at least $q$
edges will be removed by all the algorithms.  By the initial assumption
that $T'<c\sqrt{n}$, it follows that $q=\Omega\big(\frac{n}{T'}\big)$,
while, by ${c(\eopt)}=O(T')$, hence completing the
proof. Now we analyze each algorithm separately.

\noindent
\textbf{\textsc{ProductDegree}}: For all $u\in V(G_1)$, $d(u)\le T'+1$ while,
for each $i=1,\ldots,q$, $d(v_i)\ge T'+2$. Therefore,
$(v_i,v_{i+1})$, $i=0,\ldots,q-2$ has higher score than any
edge in $G_1$.

\noindent
\textbf{\textsc{EigenScore}}:
Let $x$ denote the unit eigenvector corresponding to $\lambda_1(G)$ and
for any $v\in V(G)$, let $x(v)$ denote the $v$th component of $x$. We
will show that $x(v_{q-1})>x(v_{q-2})>\cdots>x(v_0)>x(v')$ where $v'$
is any vertex in $G_1$ other than $v_0$. This implies that all
the edges $(v_i,v_{i+1})$, $i=0,\ldots,q-2$ have eigenscore greater than
the edges in $G_1$.

Let $\lambda:=\lambda_1(G)$.  By symmetry, all $v'\in
V(G_1)\setminus\{v_0\}$ have the same eigenvector component $x(v')$ and all
leaves of $v_i$ have the same component $x(l_i)$. Let
$A$ be the adjacency matrix of $G$. Since $Ax=\lambda x$, we have

\begin{subequations}\label{eqn:eig}
\begin{align}
\lambda x(v')&=(T'-1)x(v')+x(v_0) \label{eqn:v0}\\
\lambda x(v_0)&=T' x(v')+x(v_1) \label{eqn:v1}\\
\lambda x(v_i)&=x(v_{i-1})+x(v_{i+1})+T' x(l_i), \text{ $1\le i\le q-1$}\label{eqn:vip1}\\
\lambda x(v_q)&=x(v_{q-1})+(T'+1)^2x(l_q)\label{eqn:vq}\\
\lambda x(l_i)&=x(v_i), \text{ $1\le i\le q$}\,. \label{eqn:leaf}
\end{align}
\end{subequations}
From~\eqref{eqn:v0} and the fact that $\lambda\ge T'+1$,
\begin{align}\label{eqn:v0lb}
x(v_0)=(\lambda-T'+1)x(v')\ge 2x(v')\,.
\end{align}
By induction on $i$, we will show that $x(v_i)\ge\frac{T'}{2}x(v_{i-1})$ for
$i=1,\ldots,q-1$. The base case is $i=1$.
Using~\eqref{eqn:v1},~\eqref{eqn:v0lb} and the bound $\lambda\ge T'+1$,
\begin{align}\label{eqn:v1lb}
x(v_1)=\lambda x(v_0)-T' x(v')\ge> \frac{T'}{2}x(v_0)\,.
\end{align}
Assuming $x(v_i)\ge\frac{T'}{2}x(v_{i-1})$
and applying~\eqref{eqn:vip1},~\eqref{eqn:leaf} and again $\lambda\ge T'+1$,

\begin{align}\label{eqn:v2lb}
x(v_{i+1})&=\lambda x(v_i)-x(v_{i-1})-T' x(l_i)\\
&\ge\bigg(T'+1-\frac{2}{T'}-\frac{T'}{\lambda}\bigg)x(v_i)\\
&\ge\bigg(T'+1-\frac{2}{T'}-\frac{T'}{T'+1}\bigg)x(v_i)>\frac{T'}{2}x(v_i).
\end{align}
From~\eqref{eqn:v0lb} and~\eqref{eqn:v2lb}, it follows that
$x(v_{q-1})>x(v_{q-2})>\cdots>x(v_0)>x(v')$.

\noindent
\textbf{\textsc{Hybrid}}: Since both \textsc{ProductDegree} and \textsc{EigenScore}
rate edges $(v_i,v_{i+1})$, $i=0,\ldots,q-2$, higher than any edge in
$G_1$, it follows that the same holds for \textsc{Hybrid} as well.
%%%%%%%%
\newcommand{\ec}{e^c}
\newcommand{\ecvo}{e^c_{v_0}}

\noindent
\textbf{\textsc{LinePagerank}}:
Let $\pi(e)$ denote the pagerank of edge $e$. We will show that
$\pi(v_{q-1}v_{q})=\pi(v_{q-2}v_{q-1})=\cdots=\pi(v_1v_2)>\pi(v_0v_1)>\pi(\ecvo)>\pi(\ec)$
where $\pi(\ecvo)$ (by symmetry) is the pagerank of every edge in clique
$G_1$ incident with $v_0$ while $\pi(\ec)$ (again by symmetry) is the
pagerank of every other edge in the clique. Let $l_i$ denote the leaf
edges incident with $v_i$ for $i=1,\ldots,q$. Pagerank of each edge is
computed as follows: $\pi(e)=\sum_{e'\in N(e)}\frac{\pi(e')}{d(e')}$
where, $N(e)$ and $d(e)$ denote the set of neighbors and degree
respectively of $e$ in the line graph.

In the line graph, the degrees of each edge of $G$ are as follows:
$d(\ec)=2(T'-1);\,d(\ecvo)=2T'-1;
d(v_0v_1)=2T'+1;\,d(v_{q-1}v_q)=(T'+1)^2+1;
d(v_iv_{i+1})=2(T'+1),\,i=1,\ldots,q-2;
d(l_i)=T'+1,\,i=1,\ldots,q-1$.
The pageranks of the relevant edges are as follows:
%\vspace{-0.1in}
\begin{subequations}
\label{eqn:pr}
\begin{align}
\pi(\ec)&=\frac{2(T'-1)-2}{2(T'-1)}\pi(\ec)+\frac{2\pi(\ecvo)}{2(T'-1)+1} \label{eqn:ec}\\
\pi(\ecvo)&=\frac{\pi(\ec)}{2}+\frac{T'-1}{2T'-1}\pi(\ecvo)+\frac{\pi(v_0v_1)}{2T'+1} \label{eqn:ecvo}\\
\pi(v_0v_1)&=\frac{T'\pi(\ecvo)}{2T'-1}+\frac{\pi(v_1v_2)}{2(T'+1)}+\frac{T'\pi(l_1)}{T'+1}\label{eqn:v0v1}\\
\pi(v_1v_2)&=\frac{\pi(v_0v_1)}{2T'+1}+\frac{\pi(v_2v_3)}{2(T'+1)}+\frac{T'(\pi(l_1)+\pi(l_2))}{T'+1}\label{eqn:v1v2}\\
\pi(v_iv_{i+1})&=\frac{\pi(v_{i-1}v_i)+\pi(v_{i+1}v_{i+2})}{2(T'+1)}+\frac{T'(\pi(l_i)+\pi(l_{i+1}))}{T'+1},\nonumber\\
&\hspace{10em}i=2,\ldots,q-2\label{eqn:vivj}\\
\pi(l_1)&=\frac{T'-1}{T'+1}\pi(l_1)+\frac{\pi(v_0v_1)}{2T'+1}+\frac{\pi(v_1v_2)}{2(T'+1)}\label{eqn:l1}\\
\pi(l_i)&=\frac{T'-1}{T'+1}\pi(l_1)+\frac{\pi(v_{i-1}v_i)+\pi(v_{i}v_{i+1})}{2(T'+1)}\nonumber\\
&\hspace{11em}i=2,\ldots,q-2\,.\label{eqn:li}
\end{align}
\end{subequations}
Using~\eqref{eqn:pr}, we have the following:
\begin{subequations}\label{eqn:prelim}
\begin{align}
&\eqref{eqn:ec}\Rightarrow \pi(\ecvo)=\frac{2(T'-1)+1}{2(T'-1)}\pi(\ec), \label{eqn:ecvolb}\\
&\text{\eqref{eqn:ecvo} and~\eqref{eqn:ecvolb}}\Rightarrow \pi(v_0v_1)=\frac{2T'+1}{2T'-1}\pi(\ecvo),\label{eqn:vov1lb}\\
&\text{\eqref{eqn:v0v1},~\eqref{eqn:l1} and~\eqref{eqn:vov1lb}}\Rightarrow \pi(v_1v_2)=\frac{2(T'+1)}{2T'+1}\pi(v_0v_1),\label{eqn:v1v2lb}\\
&\text{\eqref{eqn:v1v2},~\eqref{eqn:l1},~\eqref{eqn:li} and~\eqref{eqn:v1v2lb}}\Rightarrow \pi(v_2v_3)=\pi(v_1v_2)\,.\label{eqn:v2v3lb}
\end{align}
\end{subequations}
Now, by induction on $i$ we can show that
$\pi(v_iv_{i+1})=\pi(v_{i-1}v_{i})$, for $i=2,\ldots,q-2$. The base
case $i=1$ is covered in~\eqref{eqn:v2v3lb}. For any $k\ge 2$, applying
$\pi(v_kv_{k+1})=\pi(v_{k-1}v_k)$ in~\eqref{eqn:v1v2} (with $i=k$) and~\eqref{eqn:li},
it follows that $\pi(v_{k+1}v_{k+2})=\pi(v_{k}v_{k-1})$.

Hence, proved.

%%%%%%%%%%%%%%%%%%%%%%%%

}%iftoggle{fullversion}
{}

\end{document}